\documentclass[letterpaper,10pt,twoside]{article}
\usepackage[nottoc]{tocbibind} 

\usepackage[utf8]{inputenc}  
\usepackage[english]{babel} 

\usepackage[T1]{fontenc} 
\usepackage{lmodern}

\usepackage{mathtools,amssymb,amsthm} 

\usepackage{geometry} 
\usepackage[frame,arc,knot,matrix]{xy}
\usepackage{cancel}

\usepackage{tikz}
\usepackage{pgf}
\usepackage{pgfplots}
\usepgfmodule{plot}
\usetikzlibrary{patterns}
\usetikzlibrary{decorations}
\usetikzlibrary{decorations.pathreplacing}
\usetikzlibrary{decorations.markings}

\newcommand{\be}{\begin{equation}}
\newcommand{\ee}{\end{equation}}

\newcommand{\im}{\mathrm{i}}

\newcommand{\R}{\mathbb{R}}
\newcommand{\bC}{\mathbb{C}}
\newcommand{\defeq}{\coloneqq}
\newcommand{\tens}{\otimes}
\DeclareMathOperator{\id}{id}
\newcommand{\xd}{\mathrm{d}}
\newcommand{\xD}{\mathcal{D}}
\newcommand{\cH}{\mathcal{H}}
\newcommand{\pol}{\mathrm{pol}}
\newcommand{\ls}{\ell}
\newcommand{\ms}{\mathsf{m}}
\newcommand{\toi}{\hookrightarrow}

\theoremstyle{definition}
\newtheorem{dfn}{Definition}[section]

\theoremstyle{plain}
\newtheorem{lem}[dfn]{Lemma}
\newtheorem{prop}[dfn]{Proposition}

\newtheorem{cor}[dfn]{Corollary}

\allowdisplaybreaks

\begin{document}

\begin{titlepage}
\title{\textbf{The vacuum as a Lagrangian subspace}}
\author{%
  Daniele Colosi\footnote{email: dcolosi@enesmorelia.unam.mx}\\
  Escuela Nacional de Estudios Superiores, Unidad Morelia,\\
  Universidad Nacional Autónoma de México,\\
  C.P.~58190, Morelia, Michoacán, Mexico
  \and Robert Oeckl\footnote{email: robert@matmor.unam.mx}\\
  Centro de Ciencias Matemáticas,\\
  Universidad Nacional Autónoma de México,\\
  C.P.~58190, Morelia, Michoacán, Mexico}
\date{UNAM-CCM-2019-2\\ 19 March 2019\\ 11 November 2019 (v2)}

\maketitle

\vspace{\stretch{1}}

\begin{abstract}
  
We unify and generalize the notions of vacuum and amplitude in linear quantum field theory in curved spacetime. Crucially, the generalized notion admits a localization in spacetime regions and on hypersurfaces.
The underlying concept is that of a Lagrangian subspace of the space of complexified germs of solutions of the equations of motion on hypersurfaces. Traditional vacua and traditional amplitudes correspond to the special cases of definite and real Lagrangian subspaces respectively.
Further, we introduce both infinitesimal and asymptotic methods for vacuum selection that involve a localized version of Wick rotation.
We provide examples from Klein-Gordon theory in settings involving different types of regions and hypersurfaces to showcase generalized vacua and the application of the proposed vacuum selection methods. A recurrent theme is the occurrence of mixed vacua, where propagating solutions yield definite Lagrangian subspaces and evanescent solutions yield real Lagrangian subspaces. The examples cover Minkowski space, Rindler space, Euclidean space and de~Sitter space.
A simple formula allows for the calculation of expectation values for observables in the generalized vacua.

\end{abstract}

\vspace{\stretch{1}}
\end{titlepage}

\tableofcontents

\section{Introduction}

In non-relativistic quantum theory a \emph{vacuum state} can simply be identified with a lowest-energy state. In a relativistic context the absence of a unique notion of time and consequently energy, makes this less straightforward. Minkowski space has a rich isometry group (the Poincaré group) that helps to fix a notion of vacuum by demanding its invariance. However, generic curved spacetimes do not admit isometries. This makes the question of how to choose a vacuum state rather important, as well as the understanding of what such a choice means physically. A further important question about the vacuum concerns its ``localizability'' properties. Usually, a vacuum is seen as encoding global information about spacetime. This is reinforced by the Reeh-Schlieder theorem \cite{ReSc:unitequiv}. However, one can ask to which extent a vacuum might encode information just about a spacetime region or (as we shall see) a hypersurface neighborhood. This question is particularly important from the point of view of Segal's axiomatic approach to quantum field theory, which posits that quantum amplitudes in composite spacetime regions may be decomposed into amplitudes in component regions \cite{Seg:cftproc,Seg:cftdef,Oe:gbqft}. More recently, this approach has been generalized to include observables \cite{Oe:feynobs} and general processes \cite{Oe:posfound}. A third question we want to raise here concerns the generalization of the notion of vacuum to a context where no background metric is fixed from the outset. This is relevant in particular for quantum gravity. With the present work we aim to make some contribution to addressing each of these questions.

To be able to make some headway we restrict in this work purely to linear (i.e., free) field theory. We recall (Section~\ref{sec:revquant}) that a standard quantization method in curved spacetime \cite{BiDa:qftcurved} starts with selecting a set of modes (i.e., solutions of the equations of motion) that satisfy certain completeness and orthogonality properties (\ref{eq:propmodes}) with respect to an inner product (\ref{eq:iplc}) that derives from the symplectic form on the solution space. A choice of such modes amounts to selecting a vacuum. Equivalently, we may encode this choice in terms of a complex structure on solution space with certain properties. As we emphasize in this work, a third way of encoding this information is in terms of a particular type (that we call \emph{definite}) of \emph{Lagrangian subspace} of the solution space. In Minkowski space with the standard vacuum, this Lagrangian subspace is precisely the space of ``positive energy solutions'' and its conjugate that of ``negative energy solutions''. In order to move towards a more local picture and away from a restriction to Minkowski space we recall that there is a natural symplectic form on the space of germs of solutions on any hypersurface in spacetime (Appendix~\ref{sec:lagingreds}). A vacuum can then be encoded as a definite Lagrangian subspace on any hypersurface. If the hypersurface is spacelike and spacetime globally hyperbolic this can be brought into correspondence with the more traditional global perspective.

There is another, apparently completely distinct setting where Lagrangian subspaces occur in (purely classical) field theory (Section~\ref{sec:classlag}). This is the symplectic framework of Kijowski and Tulczyjew \cite{KiTu:symplectic}, axiomatized in the linear case in \cite{Oe:holomorphic}. The key insight is that the solutions of a sufficiently simple field theory in a spacetime region form a Lagrangian subspace of the space of germs of solutions on the boundary.\footnote{``Sufficiently simple'' means here for example that there are no gauge symmetries. In the presence of gauge symmetries a refined scheme can be applied that involves symplectic reduction, see e.g.\ \cite{DiOe:qabym}.} The Lagrangian subspaces in question are \emph{real} subspaces in contrast to the definite ones for vacua which are necessarily complex (and defined on the \emph{complexified} space of germs).
Our core proposal (Section~\ref{sec:vaclag}) is that, nevertheless, both occurrences of Lagrangian subspaces are really special cases of a common unified structure, which, for simplicity we continue to call vacuum. To this end, we show on the classical level that the definite Lagrangian subspaces are naturally associated to ``sufficiently'' non-compact regions of spacetime, complementing the real Lagrangian subspaces for compact and ``mildly'' non-compact regions. Crucially, also Lagrangian subspaces that are neither definite nor real (but are a mixture of both) occur naturally, as we show. The unification becomes really compelling at the quantum level, where we show that the wave function for a standard vacuum state takes exactly the same form as the wave function encoding the state dual to the amplitude for a region. This is most easily seen by using the Schrödinger representation and the Feynman path integral. Expectation values of observables (defined as functions on spacetime field configurations) on all the generalized vacua can be evaluated by reducing to Weyl observables and then applying a simple path integral formula (\ref{eq:veweyl}).

A second component of the present work consists of the proposal of methods for vacuum selection (Section~\ref{sec:vchoice}). These are inspired by \emph{Euclidean methods} and incorporate notions of \emph{Wick rotation}. We observe that real Lagrangian subspaces occur naturally in association with decaying asymptotic boundary conditions. This suggests to view the definite Lagrangian subspaces of traditional vacua as arising through a Wick rotation of boundary conditions. Concretely, we propose an infinitesimal and an asymptotic method for fixing a vacuum. While this works straightforwardly when solutions show a decaying behavior, it requires a Wick rotation when solutions show oscillatory behavior. The latter case recovers traditional methods of vacuum selection using timelike vector fields.

In order to motivate our proposal we showcase the natural occurrence of generalized vacua in simple examples and demonstrate the application of our vacuum selection methods. This is partly in the spirit of the reverse engineering approach to quantum field theory, where we use known tools and methods to extract underlying structure \cite{Oe:reveng}. For simplicity, all examples are based on (massive or massless) Klein-Gordon theory. The examples involve different regions and hypersurfaces (including timelike ones) in Minkowski space (Sections~\ref{sec:hypcyl} and \ref{sec:tlhp}), Rindler space (Section~\ref{sec:Rindler}), a Euclidean space (Section~\ref{sec:2deucl}), and de~Sitter space (Section~\ref{sec:deSitter}). An intriguing repeated pattern is the occurrence of \emph{evanescent waves} with a decaying behavior and corresponding real Lagrangian subspaces along with the \emph{oscillating waves} with corresponding definite Lagrangian subspaces. It is only the latter that occur in the traditional approach to the vacuum.

While aspects of the example applications are novel, their purpose is limited to providing an initial proof of concept for the proposed generalized notion of vacuum and selection methods. The real interest of these new concepts and methods lies in their applicability to situations which lie outside the scope of standard methods or where such methods lack conceptual clarity or present technical difficulty. Of particular interest are spacetimes that are not globally hyperbolic, such as anti-de~Sitter space, black hole spacetimes or certain cosmological spacetimes. On the other hand, although this is not emphasized explicitly in this work, a wide range of boundary conditions may be understood in terms of our generalized notion of vacuum. This might lead to a completely different class of applications such as to the Casimir effect and related problems.
We notice, in accordance with previous remarks, another potential area of application in terms of quantum theory (such as quantum gravity) on spacetimes without background metric. While we focus the discussion in this work on standard quantum field theories and the methods of vacuum selection proposed in Section~\ref{sec:vchoice} rely to some extent on a metric, the framework of Section~\ref{sec:vaclag} is in principle applicable also in the absence of a metric.
For further discussion of results and a more detailed outlook, see Section~\ref{sec:outlook}.

We emphasize that the present work is focused on certain aspects of the notion of vacuum only. Other important aspects such as whether a Hadamard condition \cite{KaWa:qfstatesbifurcate} is satisfied, relevant for obtaining a renormalized energy-momentum tensor, are not touched upon. This does not mean that they are not interesting, but that their relation to the presented concepts and methods is outside of the scope of this work and should be the subject of future investigation.

Some mathematical details on Lagrangian subspaces are collected in Appendix~\ref{sec:mathlag}. This includes Proposition~\ref{prop:rdlagcompl}, which is instrumental in ensuring well-definedness and uniqueness in the application of formula (\ref{eq:veweyl}) for vacuum expectation values. In Appendix~\ref{sec:caxioms} an axiomatization of our notion of generalized vacuum is presented, generalizing the axiomatic framework \cite{Oe:holomorphic} that formalizes the mentioned Lagrangian approach of Kijowski and Tulczyjew \cite{KiTu:symplectic} in the linear case.

\section{Quantization, complex structure, Lagrangian subspaces}
\label{sec:revquant}

In the present section we briefly review aspects of the conventional approach to quantization of bosonic field theory in curved spacetime \cite{BiDa:qftcurved}. A more precise treatment of some of the mathematical structures used in this section is provided in Appendix~\ref{sec:mathlag}.

\subsection{Modes and complex structure}
\label{sec:modecs}

Consider a classical field theory on a globally hyperbolic spacetime. The \emph{phase space} $L$ is the \emph{space of solutions} of the equations of motion. It can also be identified with the space of \emph{initial data} on a spacelike hypersurface. We suppose that $L$ is a \emph{real vector space}. That is, we deal with linear or ``free'' field theory. Any interactions would be treated perturbatively. An important ingredient of the Lagrangian description of the field theory is the \emph{symplectic form} $\omega:L\times L\to\R$ on $L$. This is a non-degenerate anti-symmetric bilinear form.

We denote by $L^\bC= L\oplus \im L$, the \emph{complexification} of $L$. This is the complex vector space whose elements take the form $a+\im b$ for $a,b\in L$. It carries a \emph{complex structure}, i.e., we know what it means to multiply with $\im$. It also carries a \emph{real structure}, i.e., we know what it means to complex conjugate an element, namely $\overline{a+\im b}\defeq a-\im b$ for $a,b\in L$.
Using the symplectic form $\omega$ we may define a sesquilinear hermitian inner product on $L^\bC$, given for $\phi,\phi'\in L^{\bC}$ by,
\begin{equation}
  (\phi,\phi')\defeq 4\im \omega(\overline{\phi},\phi') .
  \label{eq:iplc}
\end{equation}
Note that this inner product is not, and cannot be positive-definite.\footnote{Suppose that $(\phi,\phi)>0$ for some $\phi\in L^\bC$. Then, $(\overline{\phi},\overline{\phi})<0$.}

A standard way \cite{BiDa:qftcurved} to construct a quantization starts with a complete set $\{u_k\}_{k\in I}$ of elements of $L^\bC$, called \emph{modes}, with the following orthogonality properties,
\begin{equation}
  (u_k,u_l)=\delta_{k,l},\quad
  (\overline{u}_k,\overline{u}_l)=-\delta_{k,l},\quad
  (u_k,\overline{u}_l)=0,\quad\forall k,l\in I .
  \label{eq:propmodes}
\end{equation}
Denote by $L^+$ and $L^-$ the complex subspaces of $L^\bC$ spanned by the modes $\{u_k\}_{k\in I}$ and $\{\overline{u}_k\}_{k\in I}$ respectively. $L^+$ and $L^-$ are thus orthogonal subspaces that span all of $L^\bC$ and the inner product $(\cdot,\cdot)$ is \emph{positive-definite} in $L^+$ and \emph{negative-definite} in $L^-$. $L^+$ and $L^-$ are complex conjugates of each other, $\overline{L^+}=L^-$ and $\overline{L^-}=L^+$. We assume that the spaces $L^+$ and $L^-$ are \emph{complete} with respect to the inner product.
The state space of the quantum theory is then constructed as a Fock space $\cH$ with \emph{creation} operators $\{a^\dagger_k\}_{k\in I}$ and \emph{annihilation} operators $\{a_k\}_{k\in I}$ corresponding to the modes $\{\overline{u}_k\}_{k\in I}$ and $\{u_k\}_{k\in I}$ respectively, with commutation relations
\begin{equation}
  [a_k,a_l]=0,\quad [a^\dagger_k,a^\dagger_l]=0,\quad
  [a_k,a^\dagger_l]=\delta_{k,l} .
\end{equation}
The \emph{vacuum} state is characterized by the property that it is annihilated by all annihilation operators. Thus, two different sets of modes give rise to the same vacuum precisely if the space $L^+$ (or equivalently $L^-$) is the same for both sets.

An important property of the spaces $L^+$ and $L^-$ is that they are \emph{Lagrangian subspaces} of $L^\bC$. This means (as shown here for $L^+$) that they are \emph{isotropic}, i.e.,
\begin{equation}
  \omega(\phi,\eta)=0,\quad\forall\phi,\eta\in L^+,
  \label{eq:isotrop}
\end{equation}
and \emph{coisotropic}, i.e.,
\begin{equation}
  \omega(\phi,\eta)=0,\quad\forall\phi\in L^+ \Longrightarrow \eta\in L^+ .
  \label{eq:coisotrop}
\end{equation}
Isotropy follows from the third property in expression (\ref{eq:propmodes}) while coisotropy follows from the first two (which are in fact equivalent). A choice of vacuum might be characterized as follows in the present quantization scheme: Choose a Lagrangian subspace $L^+\subseteq L^\bC$ in such a way that the inner product (\ref{eq:iplc}) is positive-definite on $L^+$. We call such a subspace a \emph{positive-definite} Lagrangian subspace. An orthonormal basis $\{u_k\}_{k\in I}$ of $L^+$ then yields a set of modes with the properties (\ref{eq:propmodes}).

Define a complex linear operator $J:L^\bC\to L^\bC$ as follows,
\begin{equation}
  J u_k=\im u_k,\quad J \overline{u}_k=-\im \overline{u}_k,\quad\forall k\in I .
\end{equation}
Expressed differently, $L^+$ and $L^-$ are the eigenspaces of $J$ with eigenvalues $\im$ and $-\im$ respectively. It is then easy to verify that,
\begin{equation}
  J^2=-\id,\quad\text{and}\quad \omega(J \phi,J \phi')=\omega(\phi,\phi'),\quad\forall \phi,\phi'\in L^\bC .
  \label{eq:cstruc}
\end{equation}
The orthogonal projection operators $P^{\pm}:L^\bC\to L^\bC$ onto the subspaces $L^{\pm}$ can be written in terms of the operator $J$,
\begin{equation}
  P^{\pm}\phi =\frac{1}{2}\left(\phi\mp\im J\phi\right) .
\end{equation}
It is also easy to see that $J$ commutes with complex conjugation on $L^\bC$. This implies that it is the complexification of a real linear operator $L\to L$ that we shall also denote by $J$. $J$ is a \emph{complex structure} on $L$. That is, it makes $L$ (not $L^\bC$) into a complex vector space by defining the multiplication with $\im$ to be the application of $J$.
Combining with the symplectic form $\omega$, we can construct on $L$ a sesquilinear and hermitian inner product with respect to this complex structure $J$,
\begin{equation}
  \{\phi,\phi'\}\defeq 2\omega(\phi,J \phi')+2\im\omega(\phi,\phi') .
  \label{eq:ipl}
\end{equation}
Define $w_k\defeq u_k+\overline{u}_k$. Then $w_k$ is an element of $L$ and as is easy to verify,
\begin{equation}
  \{w_k,w_l\}=\delta_{k,l} .
\end{equation}
In particular, the complex inner product $\{\cdot,\cdot\}$ is \emph{positive-definite} and the set $\{w_k\}_{k\in I}$ forms an \emph{orthonormal basis} of $L$ viewed as a complex vector space. Conversely, we can recover the modes $u_k$ and $\overline{u}_k$ from $w_k$ as,
\begin{equation}
  u_k=P^+ w_k,\quad \overline{u}_k=P^- w_k .
\end{equation}

An equivalent construction of the Fock space $\cH$ starts from the space $L$, viewed as a complex inner product space with the inner product (\ref{eq:ipl}). Thus, the $n$-particle space is then a symmetrized $n$-fold tensor product of copies of $L$. $\cH$ is the completed direct sum of all these $n$-particle spaces with $n$ ranging from $0$ to infinity. In this context, equipping $L$ with the structure of a complex Hilbert space of 1-particle states is also referred to as ``first quantization'' and the construction of the Fock space $\cH$ over $L$ as ``second quantization''.

To summarize, we have two equivalent ways to determine a quantization, i.e., a vacuum:
\begin{itemize}
\item Choose a positive-definite Lagrangian subspace $L^+\subseteq L^\bC$, i.e., a Lagrangian subspace that is positive-definite with respect to the inner product $(\cdot,\cdot)$ given by (\ref{eq:iplc}).
\item Choose a complex structure $J:L\to L$, satisfying conditions (\ref{eq:cstruc}), and such that the inner product $\{\cdot,\cdot\}$ given by (\ref{eq:ipl}) is positive-definite. We call this a \emph{positive-definite} complex structure.
\end{itemize}
We have simplified the mathematical treatment slightly here. The precise statement of equivalence is given in terms of Propositions~\ref{prop:dlagtoj} and \ref{prop:jtodlag} in Appendix~\ref{sec:mathlag}. Notably, this implies a completeness property that is mostly left implicit in the exposition of the present section.

\subsection{Time, energy, and complex structure}
\label{sec:tec}

Different choices of vacua might lead to different physics. It is thus important to identify criteria for and implications of different such choices. We recall how this works in Minkowski space and aspects of its extension to curved spacetime.

Suppose that we work in Minkowski space with a fixed inertial coordinate system. Time translations lead to induced transformations on the space $L$ of solutions. Infinitesimally, this gives rise to a derivative operator $\partial_0:L\to L$. Its exponentiation describes time evolution in $L$. Thus evolution for a time $\Delta_t$ corresponds to the operator $e^{\Delta_t \partial_0}$ on $L$. In the quantum theory, time evolution is to be described through a Hamiltonian operator $H$ on the Fock space $\cH$. According to the Schrödinger equation, to a time $\Delta_t$ corresponds the operator $e^{-\Delta_t\im H}$ on $\cH$.

This suggests to construct first a Hamiltonian operator $h$ on $L$ in the ``first quantization'' step, based on the classical operator $\partial_0$. Then, in the second quantization step, $H$ is taken to be the operator on the Fock space $\cH$ induced by $h$ on its 1-particle subspace $L$. The simplest way to construct $h$ is to precisely match the quantum with the classical time evolution on $L$. That is,
\begin{equation}
  e^{- \Delta_t J h}= e^{\Delta_t\partial_0} ,\quad\text{implying},\quad
  h=J \partial_0
  \label{eq:tevolh}
\end{equation}
We recall that the multiplication with $i$ in $L$ is given by the complex structure $J$. If we impose the usual requirement of self-adjointness and non-negativity on the Hamiltonian operator $h$, then the complex structure $J$ can be determined with condition (\ref{eq:tevolh}).

For concreteness consider the Klein-Gordon theory with mass $m$, given by the action,
\begin{equation}
 S(\phi)=\frac{1}{2}\int \xd t\,\xd^3 x\, \left((\partial_0 \phi)
 (\partial_0 \phi)- \sum_i (\partial_i \phi)(\partial_i \phi)
 -m^2\phi^2\right) .
\end{equation}
We can expand complexified solutions, i.e., elements of $L^\bC$ in terms of plane waves,
\begin{equation}
 \phi(t,x)=\int\frac{\xd^3 k}{(2\pi)^3 2E}
  \left(\phi^\text{a}(k) e^{-\im(E t-k x)}+\phi^\text{b}(k) e^{\im(E t-k x)}\right) .
\end{equation}
Real solutions, i.e., elements of $L$ have the property $\overline{\phi^\text{b}(k)}=\phi^\text{a}(k)$. The symplectic form is given by,\footnote{Note that the sign of the symplectic form depends on a choice of orientation, see Appendix~\ref{sec:lagingreds}. Here the orientation is chosen to correspond to considering equal-time hypersurfaces as boundaries of future half-spaces.}
\begin{align}
 \omega_t(\phi_1,\phi_2)  & =\frac{1}{2}\int\xd^3 x\,
  \left(\phi_1(t,x) (\partial_0 \phi_2)(t,x) - \phi_2(t,x)(\partial_0\phi_1)(t,x)\right) \label{eq:sfkget} \\
 & =\frac{\im}{2}\int\frac{\xd^3 k}{(2\pi)^3 2E}
 \left(\phi_1^\text{a}(k)\phi_2^\text{b}(k)-\phi_2^\text{a}(k)\phi_1^\text{b}(k)\right) .
\end{align}
(Note that the value of $t$ in the first line is arbitrary.)
In terms of the momentum modes the operator $\partial_0$ acts as,
\begin{equation}
  (\partial_0 \phi)^\text{a}(k)=-\im E\,\phi^\text{a}(k) ,\qquad (\partial_0 \phi)^\text{b}(k)=\im E\,\phi^\text{b}(k) .
\end{equation}
In particular, $\partial_0$ does not admit a spectral decomposition on the solution space $L$, but only on its complexification $L^\bC$, with imaginary eigenvalues. In order for $J \partial_0$ to have non-negative eigenvalues we thus need,
\begin{equation}
  (J\phi)^\text{a}(k)=\im\phi^\text{a}(k),\qquad (J\phi)^\text{b}(k)=-\im\phi^\text{b}(k)
\end{equation}
It is easily verified that $J$ defined in this way is a compatible complex structure satisfying conditions (\ref{eq:cstruc}). What is more, the associated inner product (\ref{eq:ipl}) is given by,
\begin{equation}
   \{\phi_1,\phi_2\}=2\int\frac{\xd^3 k}{(2\pi)^3 2E} \phi_2^\text{a}(k) \phi_1^\text{b}(k) .
\end{equation}
This is easily seen to be positive-definite on $L$. That is, we have a valid quantization in the standard sense. The eigenspaces of $J$ are,
\begin{equation}
  L^+=\{\phi\in L^\bC : \phi^\text{b}(k)=0\,\forall k\},\quad\text{and}\quad
  L^-=\{\phi\in L^\bC : \phi^\text{a}(k)=0\,\forall k\} .
\end{equation}
$L^+$ is called the space of \emph{positive frequency} or \emph{positive energy} modes, while $L^-$ is called the space of \emph{negative frequency} or \emph{negative energy} modes. While we might perfectly well consider orthonormal bases $\{u_k\}_{k\in I}$ and $\{\overline{u}_k\}_{k\in I}$ of these spaces with the properties (\ref{eq:propmodes}), it is simpler and more customary to use complex plane wave solutions. These are eigenvectors of the operator $\partial_0$ with continuous momentum space labels,
\begin{equation}
  \phi_k(x,t)=e^{-\im (Et-kx)},\quad\text{and}\quad \overline{\phi_k(x,t)}=e^{\im (Et-kx)} .
\end{equation}
They are not actually elements of the space $L^\bC$ as they are not normalizable. Instead they satisfy delta-function orthogonality relations, which, however, are otherwise similar to the relations (\ref{eq:propmodes}).

The present procedure extends straightforwardly to other bosonic field theories in Minkowski space. In particular, the time-derivative operator $\partial_0$ has imaginary eigenvalues and the eigenspaces of the complex structure $J$ correspond to the two different signatures of these imaginary eigenvalues.

\subsection{Quantization on hypersurfaces}
\label{sec:quanthyp}

Generically, curved spacetimes do not admit a time-translation symmetry. This complicates considerably the issue of finding a suitable positive-definite complex structure or even of quantization in general. In order to address this, it turns out to be more fruitful to start with individual hypersurfaces rather than spacetime as a whole when quantizing. Thus, denote the space of \emph{germs of solutions} of the equations of motions on a hypersurface $\Sigma$ by $L_{\Sigma}$. As before, we suppose that $L_{\Sigma}$ has the structure of a real vector space. Also we suppose that it carries a \emph{symplectic form} $\omega_{\Sigma}:L_{\Sigma}\times L_{\Sigma}\to\R$ that is bilinear, anti-symmetric and non-degenerate (see Appendix~\ref{sec:lagingreds}).

In the remainder of this section we restrict to the case that $\Sigma$ is a spacelike hypersurface and that the equations of motion admit a well posed initial value problem. Then we can interpret $L_{\Sigma}$ also as the space of \emph{initial data} on $\Sigma$. Also, the restriction map $I_{\Sigma}:L\to L_{\Sigma}$ from the space $L$ of global solutions to $L_{\Sigma}$ is then an isomorphism. Its inverse is given by the evolution of initial data. For fixed $\Sigma$ this isomorphism induces a symplectic form $\omega$ on $L$ from the symplectic form $\omega_{\Sigma}$ on $L_{\Sigma}$. Due to a conservation law (see Section~\ref{sec:classlag}), this induced symplectic form $\omega$ is the same no matter what spacelike hypersurface $\Sigma$ we choose. This is how the symplectic form on $L$ that we have mentioned previously arises.

We may now proceed to perform a quantization exactly as outlined in Section~\ref{sec:modecs}, except that we replace the symplectic vector space $(L,\omega)$ with $(L_{\Sigma},\omega_{\Sigma})$ for each spacelike hypersurface $\Sigma$. In particular, we may formulate this in terms of positive-definite Lagrangian subspaces $L_{\Sigma}^{\pm}\subseteq L_{\Sigma}^\bC$ and positive-definite complex structures $J_{\Sigma}$. We obtain a Fock space $\cH_{\Sigma}$ of quantum states for each spacelike hypersurface $\Sigma$. Time evolution generalizes in this setting to the evolution between different spacelike hypersurfaces. Consider an initial spacelike hypersurface $\Sigma$ and a final spacelike hypersurface $\Sigma'$. Then, classically, the evolution from $\Sigma$ to $\Sigma'$ is described by the isomorphism $T_{\Sigma,\Sigma'}\defeq I_{\Sigma'}\circ I_{\Sigma}^{-1}:L_{\Sigma}\to L_{\Sigma'}$. As previously mentioned, this preserves the symplectic form. The simplest way to quantize this requires the complex structures to be preserved as well, i.e., we require $J_{\Sigma'}=T_{\Sigma,\Sigma'}\circ J_{\Sigma}\circ T_{\Sigma,\Sigma'}^{-1}$. Then, $T_{\Sigma,\Sigma'}$ becomes unitary with respect to the inner products (\ref{eq:ipl}) and induces a unitary map $U_{\Sigma,\Sigma'}:\cH_{\Sigma}\to\cH_{\Sigma'}$ between the corresponding Fock spaces.

In the absence of time-translation symmetries we may consider more general flows on spacetime. Even if these do not preserve solutions in spacetime, they may induce infinitesimal actions on the spaces of germs of solutions on hypersurfaces. That is, we might be able to construct operators $L_{\Sigma}\to L_{\Sigma}$ that represent such flows infinitesimally. These in turn can then be linked to complex structures through their spectrum, similarly as we have seen this for time translations.

\subsection{Path integral with past and future boundaries}
\label{sec:feynpf}

In this section we take a different starting point and review the basics of the quantization of observables in quantum field theory through the Feynman path integral. We focus on how this links to the choice of vacuum as discussed in the previous section. Most of the content of this section can be found in standard text books on quantum field theory such as \cite{ItZu:qft}.

We first recall how transition amplitudes in quantum field theory are obtained from the path integral. We work in Minkowski space with the standard quantization of Section~\ref{sec:tec}, applied on equal-time hypersurfaces as in Section~\ref{sec:quanthyp}. Thus, consider an initial state $\psi_1\in\cH_{t_1}$ at time $t_1$ and a final state $\psi_2\in\cH_{t_2}$ at time $t_2$. We keep the convention from Section~\ref{sec:quanthyp} to equip state spaces with labels that indicate hypersurfaces, which in this case are parametrized by the time variable. When referring to objects associated to the spacetime region $[t_1,t_2]\times\R^3$ we indicate this with a subscript $_{[t_1,t_2]}$. The corresponding \emph{transition amplitude} is the matrix element of the time-evolution operator $U_{[t_1,t_2]}:\cH_{t_1}\to\cH_{t_2}$ given by,
\begin{equation}
  \langle \psi_2, U_{[t_1,t_2]} \psi_1\rangle
  = \int_{K_{[t_1,t_2]}}\xD\phi\, \psi_1(\phi_1) \overline{\psi_2(\phi_2)} e^{\im S_{[t_1,t_2]}(\phi)} .
  \label{eq:tamplpi}
\end{equation}
The integral is over field configurations $\phi\in K_{[t_1,t_2]}$ in the spacetime region $[t_1,t_2]\times\R^3$ with $\phi_i$ denoting the configuration at time $t_i$. The action $S$ is evaluated in the same region. We use here the Schrödinger representation, where states on a hypersurface are wave functions on the space of field configurations on this hypersurface.

When the field theory is interacting, i.e., the action includes terms that are of higher order than two in the fields, it is not known how to directly evaluate a path integral such as (\ref{eq:tamplpi}), except for some very special cases. Instead, one sets up perturbation theory around a free theory described by an action $S$ quadratic in the fields. An important intermediate object in this construction arises by adding to the action a \emph{source term}. That is, one replaces the action $S$ with the action $S_{\mu}\defeq S + D_{\mu}$, where,
\begin{equation}
  D_{\mu}(\phi)\defeq \int\xd^4 x\, \mu(x)\phi(x) .
  \label{eq:src}
\end{equation}
For simplicity we have chosen here the notation of a real scalar field, where the \emph{source} $\mu$ is a map from spacetime to the real numbers. (This readily generalizes to fields with internal degrees of freedom.)
As in Section~\ref{sec:modecs} we denote by $L$ the real vector space of global solutions of the equations of motion determined by the free action $S$. Similarly, we denote by $A_\mu$ the space of global solutions of the equations of motion determined by the modified action $S_\mu$. Note that $A_\mu$ is not a real vector space in general, but a real \emph{affine} space. Indeed, for $S$, the equations of motion are generally homogeneous partial differential equations, while for $S_\mu$ they are inhomogeneous. In the example of the Klein-Gordon theory (compare Section~\ref{sec:tec}) the equations of motion are given by,
\begin{equation}
  (\Box+m^2)\phi(x)=\mu(x) .
\end{equation}
As before we denote by $L^\bC$ the complexification of $L$. We also denote by $A_{\mu}^\bC=A_{\mu}\oplus \im L$ the complexification of $A_{\mu}$. Elements of this space may be written as $c+\im\, d$ with $c\in A_\mu$ and $d\in L$. The addition of and element $a+\im\, b\in L^\bC$ results in $(c+a)+\im (d+b)$.

\begin{figure}
\centering
\begin{tikzpicture}[scale=1]
\fill[gray!10]  (-1,1) rectangle (6,3);
\draw[->] (0,-2) -- (0,5) node [left] {$t$};
\draw[->] (-1,-1) -- (6,-1) node [right] {$x$};
\draw (-1,1) node [left] {$t_1$} -- (6,1);
\draw (-1,3) node [left] {$t_2$} -- (6,3);
\draw[->] (1,0.8) -- (1,-0.3);
\draw[->,xshift=-0.5cm] (1,0.8) -- (1,-0.3);
\draw[->,xshift=-1.5cm] (1,0.8) -- (1,-0.3);
\draw[->,xshift=3cm] (1,0.8) -- (1,-0.3);
\draw[->,xshift=3.5cm] (1,0.8) -- (1,-0.3);
\draw[->,xshift=4cm] (1,0.8) -- (1,-0.3);
\draw[->,xshift=4.5cm] (1,0.8) -- (1,-0.3);
\draw[<-,yshift=3.5cm] (1,0.8) -- (1,-0.3);
\draw[<-,yshift=3.5cm,xshift=-0.5cm] (1,0.8) -- (1,-0.3);
\draw[<-,yshift=3.5cm,xshift=-1.5cm] (1,0.8) -- (1,-0.3);
\draw[<-,yshift=3.5cm,xshift=3cm] (1,0.8) -- (1,-0.3);
\draw[<-,yshift=3.5cm,xshift=3.5cm] (1,0.8) -- (1,-0.3);
\draw[<-,yshift=3.5cm,xshift=4cm] (1,0.8) -- (1,-0.3);
\draw[<-,yshift=3.5cm,xshift=4.5cm] (1,0.8) -- (1,-0.3);
\node at (1.5,2) {$M$};
\draw[fill=gray] plot [smooth cycle] coordinates { (2.5,1.9) (3,2.5) (3.6,2.4) (4,2)  (3.65,1.5) (2.9,1.4)};
\node at (2.5,0.5) {negative energy};
\node at (2.5,0) {solutions};
\node at (2.5,4) {positive energy};
\node at (2.5,3.5) {solutions};
\node at (3.2,2) [fill=white] {$\mu$};
\end{tikzpicture}
\caption{The region $M$ determined by the time interval $[t_1,t_2]$ in Minkowski space. A source $\mu$ is located within the region. The solution $\eta$ of the corresponding inhomogeneous equations of motion is a negative energy solution before $t_1$ and a positive energy solution after $t_2$.}
\label{fig:st_int_pos_neg}
\end{figure}
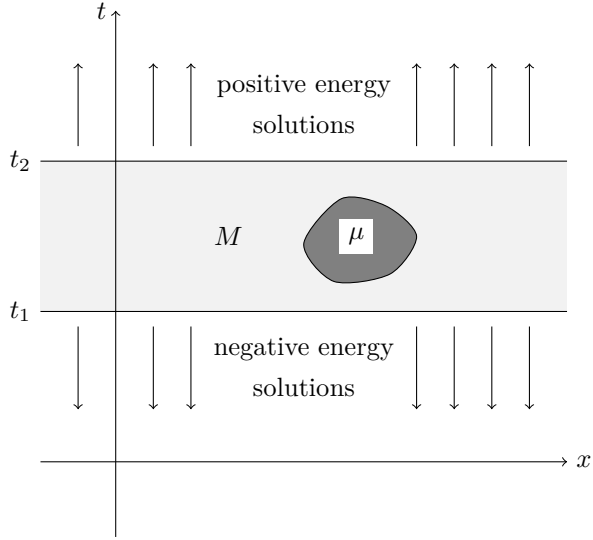

The path integral (\ref{eq:tamplpi}) with source takes a particularly simple form in the case where both the initial and the final state are given by the Fock vacuum $\psi_0$. As before, we consider the path integral between an initial time $t_1$ and a final time $t_2$. Correspondingly, we suppose that the source $\mu$ has support in the spacetime region $[t_1,t_2]\times\R^3$ only. Then,
\begin{equation}
  \langle \psi_0,U_{\mu,[t_1,t_2]}\psi_0\rangle
  =\exp\left(\frac{\im}{2} \int_{[t_1,t_2]\times \R^3} \xd^4 x\, \mu(x)\eta(x)\right) .
  \label{eq:vevsource}
\end{equation}
Here $\eta\in A_{\mu}^\bC$ is a particular complexified solution of the inhomogeneous equations of motion (i.e., with source). More precisely, $\eta$ is the \emph{unique} inhomogeneous solution with the following properties:
\begin{itemize}
\item For $t\le t_1$, $\eta$ is a negative energy solution, i.e., $\eta|_{t\le t_1}=\eta^-\in L^-$.
\item For $t\ge t_2$, $\eta$ is a positive energy solution, i.e., $\eta|_{t\ge t_2}=\eta^+\in L^+$.
\end{itemize}
See Figure~\ref{fig:st_int_pos_neg} for an illustration.
These \emph{boundary conditions} for $\eta$ are precisely tied to the choice of vacuum. At the $t=t_2$ hypersurface we restrict to positive energy solutions according to our choice of positive-definite Lagrangian subspace $L^+\subseteq L^\bC$ determining here the standard vacuum in Minkowski space, compare Section~\ref{sec:tec}. At the $t=t_1$ hypersurface we choose the same (time translated) vacuum. However, since we are dealing with an initial rather than a final hypersurface, its orientation is opposite. This means that the symplectic form changes its sign as it is given by a local integral over the hypersurface (compare expression~(\ref{eq:sympl}) in Appendix~\ref{sec:lagingreds}). Similarly, the complex structure changes sign as the normal derivative changes direction, compare expression (\ref{eq:tevolh}). Consequently, the complex inner product (\ref{eq:ipl}) on $L$ is complex conjugated. This also implies that the positive-definite Lagrangian subspace $L^+$ is complex conjugated and we obtain a restriction to negative energy modes, i.e., to $L^-=\overline{L^+}$.

Crucially, if we decrease $t_1$ or increase $t_2$ without changing the source $\mu$, the solution $\eta$ will not change. This is precisely because the definite Lagrangian subspaces $L^+\subseteq L^\bC$ and $L^-\subseteq L^\bC$ of positive and negative energy solutions are \emph{time translation invariant}, compare Section~\ref{sec:tec}. What is more, the integral in expression (\ref{eq:vevsource}) does not change either, since $\mu$ lacks support outside of $[t_1,t_2]\times\R^3$.
Consequently, the quantity (\ref{eq:vevsource}) does not depend on the choice of initial and final time as long as the support of $\mu$ is contained in $[t_1,t_2]\times\R^3$. In particular, we may formally send $t_1$ to minus infinity and $t_2$ to plus infinity. We can then drop the restriction on the support of $\mu$, although we must keep in mind that integrability on the right hand side might not be guaranteed if $\mu$ does not have compact support. We may view the condition on the early and late time behavior of $\eta$ as an \emph{(temporal) asymptotic boundary condition} determined by a choice of \emph{(temporal) asymptotic vacuum}.

The resulting identity for the path integral may be brought into the following suggestive form,
\begin{equation}
  Z(\mu)\defeq
  \int_{K}\xD\phi\, e^{\im \left(S(\phi)+D_{\mu}(\phi)\right)}
  =\exp\left(\frac{\im}{2} \int \xd^4 x\, \mu(x)\eta(x)\right) .
  \label{eq:pisrc}
\end{equation}
Here $K$ is now the space of field configurations in all of spacetime and the integral on the right hand side is also over all of spacetime. While such a path integral is commonly written down in quantum field theory text books, it is clear that its notation is ambiguous since the boundary conditions are not indicated.
On the other hand, the right hand side can be used to \emph{define} the path integral, also in curved spacetime. Thus, in a globally hyperbolic spacetime we may choose future and past \emph{asymptotic vacua} in the form of \emph{definite Lagrangian subspaces} $L^+$ and $L^-$ of the complexified global space of solutions $L^\bC$ making the inner product (\ref{eq:iplc}) positive-definite and negative-definite respectively. (Recall that we have a change of orientation for the past asymptotics that flips the sign of $\omega$, converting negative to positive-definiteness.) We may even allow the vacua to be different in the sense that $\overline{L^-}$ does not agree with $L^+$, as long as $L^+$ and $L^-$ are complementary. Complementary means that the intersection of $L^+$ and $L^-$ is $\{0\}$, but together they generate $L^\bC$. Then, $\eta$ on the right hand side of equation (\ref{eq:pisrc}) is determined as above. That is, it is the unique complexified solution of the inhomogeneous equations of motion with source $\mu$ that satisfies the boundary conditions of lying in $L^+$ in the asymptotic future and in $L^-$ in the asymptotic past.

Instead of working with the special solution $\eta$ it is often more convenient to use the \emph{Feynman propagator} $G_F$ which is a complex symmetric distribution on two copies of Minkowski space satisfying (for Klein-Gordon theory),
\begin{equation}
  (\Box_x+m^2) G_F(x,y)=\delta^4(x-y) .
  \label{eq:Feyndelta}
\end{equation}
Moreover, with one argument held fixed, $G_F$ as a function in the other argument is a complexified solution satisfying the same boundary conditions as $\eta$ for early and for late times. (Early or late times are understood with respect to the fixed argument.) This determines $G_F$ uniquely. Then,
\begin{equation}
  \eta(x)=\int\xd^4 y\, G_F(x,y) \mu(y) .
\end{equation}
In particular, we may rewrite the right hand side of equation (\ref{eq:pisrc}) as,
\begin{equation}
  Z(\mu)
  =\exp\left(\frac{\im}{2} \int \xd^4 x\,\xd^4 y\, \mu(x)G_F(x,y)\mu(y)\right) .
  \label{eq:pisrc2}
\end{equation}

An important tool in quantum field theory are the time ordered \emph{$n$-point functions}. In a standard text book one may find this expressed in terms of the path integral as follows,
\begin{equation}
  \langle\psi_0, \mathbf{T} \phi(x_1)\cdots\phi(x_n) \psi_0\rangle
  = \int\xD\phi\, \phi(x_1)\cdots\phi(x_n) e^{\im S(\phi)}
  \label{eq:npointpi}
\end{equation}
Here, $\psi_0$ symbolizes the vacuum, $\phi(x_i)$ on the left hand side are \emph{field operators} labeled by spacetime points $x_i$ and $\mathbf{T}$ denotes \emph{time-ordering}. The expression on the left hand side is to be understood in the Heisenberg picture. As is easy to see with (\ref{eq:src}), this quantity can be obtained for the free theory under consideration here from $Z(\mu)$ given by (\ref{eq:pisrc2}) by applying functional derivatives. The simplest case with $n=2$ recovers the Feynman propagator,
\begin{equation}
  \langle\psi_0, \mathbf{T} \phi(x_1)\phi(x_2) \psi_0\rangle
  = \left.\left(-\im\frac{\partial}{\partial \mu(x_1)}\right)\left(-\im\frac{\partial}{\partial \mu(x_2)}\right) Z(\mu) \right|_{\mu=0}
  =-\im\, G_F(x_1,x_2) .
\end{equation}

\section{Quantization in spacetime}
\label{sec:genquant}

In the present section we review the more recent generalization of quantization and the path integral to a spacetime local formulation \cite{Oe:boundary,Oe:gbqft}. In doing so we rely mainly on the results of \cite{Oe:feynobs}. We consider classical field theory first, and move on to quantization subsequently.

\subsection{Classical field theory and Lagrangian subspaces}
\label{sec:classlag}

In order to describe a field theory locally we consider for any spacetime region $M$ (possibly restricted to some sufficiently large class), the space $L_M$ of \emph{solutions} of the equations of motion in $M$. This is a real vector space since we are dealing with free field theory. Crucially, these solutions are defined only in $M$. There is no need or assumption that they extend to global solutions. (We may allow for contexts where a notion of global solution or even global spacetime does not exist.)
The physical content of a field theory is not so much in the structure of these solution spaces as such, but in the relation they have with each other. Here, the powerful field theoretic principle of \emph{locality} comes into play. To encode this we also need to consider hypersurfaces $\Sigma$ and associated solution spaces $L_{\Sigma}$. More precisely, $L_{\Sigma}$ is the space of \emph{germs of solutions} on $\Sigma$, i.e., solutions defined in an arbitrarily small neighborhood of $\Sigma$. In Lagrangian field theory, this space comes equipped with a symplectic form  $\omega_{\Sigma}:L_{\Sigma}\times L_{\Sigma}\to\R$ (which we assume to be non-degenerate). This symplectic form arises as the second variation of the Lagrangian on the hypersurface $\Sigma$ \cite{Woo:geomquant}, see Appendix~\ref{sec:lagingreds} for details. Crucially, the hypersurface $\Sigma$ carries an \emph{orientation}. While the orientation reversed hypersurface, denoted by $\overline{\Sigma}$, carries the same space of germs of solutions $L_{\overline{\Sigma}}=L_{\Sigma}$, the associated symplectic form changes sign, $\omega_{\overline{\Sigma}}=-\omega_{\Sigma}$.
Next, we realize that a solution in a region $M$ can be restricted to a solution in the neighborhood of the boundary $\partial M$. This gives rise to a linear map $L_M\to L_{\partial M}$. Crucially, (for well behaved regions) the image of $L_M$ under this map (which we also denote by $L_M$ when no confusion can arise) is a \emph{Lagrangian subspace} of $L_{\partial M}$. This is a powerful principle of Lagrangian field theory in spacetime \cite{KiTu:symplectic}, generalizing (as we recall below) the well known conservation of the symplectic form between spacelike hypersurfaces. Note that regions are also oriented and boundary hypersurfaces inherit an orientation from the region they bound.

We may now express the notion of \emph{composition}.
Consider two adjacent spacetime regions $M_1$ and $M_2$ that are in contact through a common hypersurface $\Sigma$. Then, we may express the relation between the corresponding solution spaces through the following \emph{exact sequence},
\begin{equation}
  L_{M_1\cup M_2} \to L_{M_1}\times L_{M_2} \rightrightarrows L_{\Sigma} .
  \label{eq:exacomp}
\end{equation}
The arrow on the left hand side means: Take a solution in the union $M_1\cup M_2$ and restrict it on the one hand to a solution in $M_1$ and on the other hand to a solution in $M_2$. The arrows on the right hand side mean: Either restrict the solution in $M_1$ to a neighborhood of $\Sigma$ or do this with the solution in $M_2$. The whole expression being an exact sequence expresses the following simple fact: A pair of solutions in $M_1$ and $M_2$ arises through restriction from a solution in the union $M_1\cup M_2$ precisely if these restrictions agree near the hypersurface $\Sigma$.

The spacetime regions and hypersurfaces on the one hand, and solution spaces with their properties on the other hand can be organized into an axiomatic system \cite[Section~4.1]{Oe:holomorphic}. This may in fact be used as a \emph{definition} of classical field theory using algebraic language instead of the usual description in terms of differential geometric structures and differential equations.

We recall how the time-evolution picture fits into this spacetime framework. Consider a globally hyperbolic spacetime and let $M$ be a region bounded by two spacelike hypersurfaces, $\Sigma_1$ and $\Sigma_2$, with the latter in the future of the former. (Figure~\ref{fig:st_int_pos_neg} shows an example.) We take $\Sigma_1$ and $\Sigma_2$ to have the same orientation with respect to a global choice of time direction. Then, as components of the boundary $\partial M$ with orientations induced from $M$, one hypersurface (by convention here $\Sigma_2$) appears with inverted orientation, $\partial M=\Sigma_1\sqcup \overline{\Sigma}_2$. Note that $L_M$, $L_{\Sigma_1}$ and $L_{\Sigma_2}$ are all naturally identified with the global solution space $L$. Correspondingly, the symplectic forms $\omega_{\Sigma_1}$ and $\omega_{\Sigma_2}$ can be viewed as forms on $L$. The map $L_M\to L_{\partial M}=L_{\Sigma_1}\times L_{\Sigma_2}$ that restricts solutions to the boundary can be written as, $\phi\mapsto (\phi,\phi)$. The isotropy property (\ref{eq:isotrop}) implies for $\phi,\phi'\in L$,
\begin{equation}
  0=\omega_{\partial M}((\phi,\phi),(\phi',\phi'))=\omega_{\Sigma_1}(\phi,\phi')+\omega_{\overline{\Sigma}_2}(\phi,\phi')=\omega_{\Sigma_1}(\phi,\phi')-\omega_{\Sigma_2}(\phi,\phi') .
\end{equation}
This is just the usual conservation property of the symplectic form between spacelike hypersurfaces. On the other hand, the coisotropy property (\ref{eq:coisotrop}) implies the non-degeneracy of $\omega_{\Sigma_1}$ and $\omega_{\Sigma_2}$.

\subsection{Quantization in regions and boundary conditions}
\label{sec:stquant}

In the present section we review quantization in general spacetime regions following \cite{Oe:feynobs}.
We recall from Section~\ref{sec:quanthyp} that the quantization prescription of Section~\ref{sec:modecs} can be carried out for individual hypersurfaces $\Sigma$. Thus, we choose a positive-definite Lagrangian subspace $L_{\Sigma}^+\subseteq L_{\Sigma}^\bC$ or positive-definite complex structure $J_{\Sigma}$. This works just as well if $\Sigma$ is not a spacelike hypersurface in a globally hyperbolic spacetime. Of course, in that case $L_{\Sigma}$ will not in general be isomorphic to the space $L$ of global solutions (if it exists). Correspondingly, the Fock space $\cH_{\Sigma}$ is not to be interpreted as a ``global'' state space of the quantum theory, but as the space of states on the hypersurface $\Sigma$. This generalizes the notion of state at a time $t$.

The most convenient way to describe the dynamics of the quantum theory in general spacetime regions is not via an evolution equation, but via amplitudes. Given a spacetime region $M$ with boundary hypersurface $\partial M$, the \emph{amplitude map} $\rho_M$ assigns to a state in $\cH_{\partial M}$ its amplitude, a complex number. It can be conveniently constructed via the Feynman path integral as,
\begin{equation}
  \rho_M(\psi)
    = \int_{K_M}\xD\phi\, \psi\left(\phi|_{\partial M}\right) e^{\im S_M(\phi)} .
  \label{eq:amplpi}
\end{equation}
Here $\psi\in\cH_{\partial M}$ is a state written as a wave function in the Schrödinger representation. The integral is over the space $K_M$ of field configurations in $M$ and $S_M$ is the action evaluated in $M$.

The transition amplitude (\ref{eq:tamplpi}) arises as a special case when $M$ is taken to be $[t_1,t_2]\times\R^3$. We are then in the situation considered at the end of Section~\ref{sec:classlag} where the boundary $\partial M$ of $M$ decomposes into an initial hypersurface at $t_1$ and a final one at $t_2$ as $\partial M=\Sigma_{t_1}\sqcup\overline{\Sigma}_{t_2}$. (Recall also Figure~\ref{fig:st_int_pos_neg}.)
The boundary solution space $L_{\partial M}$ decomposes as a cartesian product or direct sum $L_{\partial M}=L_{t_1}\oplus L_{t_2}$. Upon quantization, the Fock space then decomposes into a tensor product $\cH_{\partial M}= \cH_{t_1}\tens \cH_{t_2}^*$, i.e., $\psi=\psi_1\tens\psi_2^*$. The dualization, indicated by $^*$ is due to orientation reversal. In the Schrödinger representation the tensor product manifests as a factorization of wave functions. With $\phi|_{\partial M}=(\phi_1,\phi_2)$ we have $\psi(\phi_1,\phi_2)=\psi_1(\phi_1)\overline{\psi_2(\phi_2)}$. This recovers the path integral (\ref{eq:tamplpi}) from (\ref{eq:amplpi}), where the notation for the amplitude map and transition amplitude are related as,
\begin{equation}
  \rho_{[t_1,t_2]}(\psi_1\tens\psi_2^*)=\langle \psi_2, U_{[t_1,t_2]} \psi_1\rangle .
\end{equation}

The quantization of observables in spacetime regions is also easily accomplished through the Feynman path integral. Given a complex observable $F:K_M\to\bC$ we use for its quantization the notation $\rho_M^F:\cH_{\partial M}\to\bC$ and call this the corresponding \emph{observable map},
\begin{equation}
  \rho_M^F(\psi)
    = \int_{K_M}\xD\phi\, \psi\left(\phi|_{\partial M}\right) F(\phi)\, e^{\im S_M(\phi)} .
  \label{eq:obsmap}
\end{equation}
If the observable $F$ arises from a \emph{linear} observable $D:K_M\to\bC$ via $F=\exp(\im\, D)$, then we call it a \emph{Weyl observable}. Adding a source term to the action as in Section~\ref{sec:feynpf} is just a special case of such a Weyl observable with $D$ given by expression~(\ref{eq:src}). We are particularly interested in the special case where $\psi$ is the \emph{vacuum state} $\psi_0$. The evaluation of the path integral then yields the simple formula,\footnote{This formula appears in \cite{Oe:feynobs} in a slightly different form as formula (85). It can be derived using formula (\ref{eq:afactsymp}) of Appendix~\ref{sec:lagingreds}.}
\begin{equation}
  \rho_M^F(\psi_0)=\exp\left(\frac{\im}{2} D(\eta)\right),
  \label{eq:veweyl}
\end{equation}
generalizing formula (\ref{eq:vevsource}). Again, $\eta$ is here a special complexified solution of the ``inhomogeneous'' equations of motion in $M$. ``Inhomogeneous'' now refers to the equations of motion generated by the modified action $S_M+D$. We denote the corresponding affine space of solutions in $M$ by $A_M^D$. Its \emph{complexification} is $A_M^D\oplus \im L_M$, where $L_M$ is the space of solutions of the unmodified equations of motion in $M$. By restriction to a neighborhood of the boundary $\partial M$, $\eta$ gives rise to an element in $L_{\partial M}^{\bC}$. By slight abuse of notation we also denote this element by $\eta$. The boundary condition that $\eta$ has to satisfy is now given by the requirement that this restriction of $\eta$ be an element of the positive-definite Lagrangian subspace $L_{\overline{\partial M}}^+\subseteq L_{\partial M}^\bC$ that determines our quantization on $\partial M$ in the sense of Section~\ref{sec:modecs}.\footnote{Note that the orientation of the boundary $\partial M$ implicit in Section~\ref{sec:modecs} and relevant for the choice of the Lagrangian subspaces is opposite to that induced from the region $M$. We write $\overline{\partial M}$ to reflect this. This also affects the complex structure.} There is a unique solution $\eta$ that satisfies this requirement.
It is easy to see how the boundary conditions for $\eta$ in formula (\ref{eq:vevsource}) arise from this perspective. The boundary condition on $M$ can be split into two, corresponding to the two components of the boundary of $M$, initial and final. Correspondingly, the Lagrangian subspace $L^+_{\overline{\partial M}}\subseteq L_{\partial M}^\bC$ and the complex structure $J_{\overline{\partial M}}$ split into two. As explained previously, due to the relatively opposite orientation of past and future boundary, symplectic form and complex structure change sign, leading to a relative complex conjugation of the Lagrangian subspace determining the past boundary condition. More precisely, we have,
\begin{equation}
  J_{\overline{\partial M}}=J_{\overline{t_1}}+J_{t_2}=-J_{t_1}+J_{t_2},\quad\text{and thus},\quad
  L^+_{\overline{\partial M}}=L^+_{\overline{t_1}} \oplus L^+_{t_2}=\overline{L^+_{t_1}} \oplus L^+_{t_2}= L^-_{t_1}\oplus L^+_{t_2} .
\end{equation}
Here, $J_{t_1}$ and $J_{t_2}$ are naturally identified with a global $J$ on $L$. Similarly, $L^{\pm}_{t_1}$ and $L^{\pm}_{t_2}$ are naturally identified with the global $L^\pm\subseteq L^\bC$.

The key assumption underlying the existence and uniqueness of $\eta$ is the Lagrangian subspace property $L_M\subseteq L_{\partial M}$ (compare Section~\ref{sec:classlag}). In particular, the Lagrangian subspaces $L_M^\bC$ and $L_{\overline{\partial M}}^+$ of $L_{\partial M}^\bC$ are necessarily complementary, i.e., $L_{\partial M}^\bC=L_M^\bC\oplus L_{\overline{\partial M}}^+$. Indeed, a (complete) positive-definite and a complexified real Lagrangian subspace are always complementary, see Proposition~\ref{prop:rdlagcompl} in Appendix~\ref{sec:mathlag}. We can thus decompose any $\phi\in L_{\partial M}^\bC$ as $\phi=\phi^{\text{int}}+\phi^{\text{ext}}$ with $\phi^{\text{int}}\in L_{M}^\bC$ and $\phi^{\text{ext}}\in L_{\overline{\partial M}}^+$. Say $\xi\in A_M^D\oplus \im L_M$ is some arbitrary complexified solution of the modified equations of motion in $M$. It is then easy to see that $\eta=\xi^{\text{ext}}$.

Restrict now to the situation where the Weyl observable encodes a source $\mu$ with support in $M$. That is, $F_{\mu}=\exp(\im\, D_{\mu})$ with $D_{\mu}$ given by (\ref{eq:src}). As in Section~\ref{sec:feynpf} we can then rewrite the right-hand side of equation (\ref{eq:veweyl}) in terms of the Feynman propagator as the right hand side of equation (\ref{eq:pisrc2}). This Feynman propagator is here a symmetric distribution $G_F:M\times M\to\bC$ that satisfies the usual equation (\ref{eq:Feyndelta}). In addition, it satisfies the following boundary condition: When one argument, say $y$, is held fixed in the interior of $M$, then $G_F(x,y)$ as a function of $x$ reduces in a neighborhood of the boundary $\partial M$ to an element in the positive-definite Lagrangian subspace $L^+_{\overline{\partial M}}$.

In the present section we have so far not needed to make any reference to spacetime outside a given region $M$. However, usually $M$ will be part of a fixed global spacetime. Then, the choice of vacuum, i.e., positive-definite Lagrangian subspace on $\overline{\partial M}$ may arise as the imprint of an \emph{asymptotic boundary condition} as discussed in Section~\ref{sec:feynpf}. Correspondingly, we might consider $\eta$ and $G_F$ as globally defined objects in this setting. However, note that the present perspective on the meaning of ``asymptotic'' is potentially more general than in Section~\ref{sec:feynpf}, referring to the ``far away'' behavior of solutions, not necessarily in a temporal (far past and far future) sense.

We have also not needed to restrict in this section to spacelike hypersurfaces. Indeed, we have not needed to mention or even imply a metric on spacetime at all. The present considerations may well be applied to field theories that do not require a background spacetime metric. In fact, what we have reviewed in the present section is but a fraction of a beautifully coherent and manifestly local framework for quantizing field theory, including the free theory, quantization of observables and elementary perturbation theory, laid out in the work \cite{Oe:feynobs}. Unfortunately, this framework is seriously limited in its applicability to realistic quantum field theories as we shall see in Section~\ref{sec:hypcyl}.

\subsection{Amplitudes and vacuum in the Schrödinger representation}
\label{sec:schroedvac}

It turns out to be instructive to explore the notion of vacuum in the context of the duality relation between amplitudes and states. Thus, given a spacetime region $M$ we wish to consider a ``state'' $\hat{\rho}_M$ so that for any $\psi\in\cH_{\partial M}$,
\begin{equation}
  \rho_M(\psi)=\langle \hat{\rho}_M,\psi\rangle_{\partial M} .
  \label{eq:amplstd}
\end{equation}
Of course, with $\cH_{\partial M}$ infinite-dimensional, $\rho_M$ is unbounded and $\hat{\rho}_M$ will not be normalizable. But this does not stop us from writing down a Schrödinger wave function for it. Indeed, in the Schrödinger representation the inner product in (\ref{eq:amplstd}) is written as,
\begin{equation}
  \rho_M(\psi)
  = \int_{K_{\partial M}}\xD\varphi\, \psi(\varphi) \overline{\hat{\rho}_M(\varphi)},
\end{equation}
where the integral is over the field configuration space $K_{\partial M}$ on the boundary $\partial M$ of $M$. A formal expression for $\hat{\rho}_M$ thus follows from comparison with the path integral (\ref{eq:amplpi}),
\begin{equation} 
  \overline{\hat{\rho}_M(\varphi)}=Z_M(\varphi)=\int_{K_M, \phi|_{\partial M}=\varphi}\xD\phi\, e^{\im S_M(\phi)} .
  \label{eq:pifp}
\end{equation}
Here the integral is over those field configurations $\phi$ in the interior that match the boundary field configuration $\varphi$. $Z_M$ is called the \emph{field propagator} in $M$.
Since the action $S_M$ is a quadratic form, the integral can be solved explicitly. In particular, let $\phi_{\text{cl}}$ be the classical solution of the equations of motion in $M$ that takes the boundary value $\varphi$. (We assume this to exist and be unique here, see the discussion below.) Then, from the variational principle of the action and formal translation invariance of the path integral we get,
\begin{equation}
  Z_M(\varphi)= \exp\left(\im S_M(\phi_{\text{cl}})\right) ,
  \label{eq:fpropa}
\end{equation}
where we have dropped a numerical factor depending only on $M$.

In order to get a better understanding of the field propagator we take a closer look at the Schrödinger representation from the perspective on classical field theory as reviewed in Section~\ref{sec:classlag} \cite{Oe:affine}. To this end we note that Lagrangian field theory comes equipped with a \emph{symplectic potential} on any hypersurfaces $\Sigma$, see Appendix~\ref{sec:lagingreds} for details. In the present case of linear field theory this is a bilinear form $[\cdot,\cdot]_{\Sigma}:L_{\Sigma}\times L_{\Sigma}\to\R$. For example, for Klein-Gordon theory on an equal-time hyperplane in Minkowski space (oriented as in Section~\ref{sec:tec} as the past boundary of a region), this is,
\begin{equation}
  [\phi,\phi']_t=-\int\xd^3x\, \phi'(t,x) (\partial_0 \phi)(t,x) .
  \label{eq:etsymp}
\end{equation}
The symplectic form is its anti-symmetric part, see equation~(\ref{eq:sympfrompotlin}). We will make use of the fact that the action can be expressed in terms of the symplectic potential, see equation~(\ref{eq:actsymp}).

For a hypersurface $\Sigma$, define subspaces of $L_{\Sigma}$ as follows,
\begin{equation}
   P_{\Sigma}\defeq\{\tau\in L_{\Sigma}: [\xi,\tau]=0\; \forall \xi\in L_{\Sigma}\}\qquad
 Q_{\Sigma}\defeq\{\tau\in L_{\Sigma}: [\tau,\xi]=0\; \forall \xi\in L_{\Sigma}\} .
\label{eq:defmn}
\end{equation}
We assume that $P_{\Sigma}$ and $Q_{\Sigma}$ together generate $L_{\Sigma}$. This is enough to conclude that they are complementary Lagrangian subspaces, i.e., $L_{\Sigma}=P_{\Sigma}\oplus Q_{\Sigma}$. It turns out that the solutions in $P_{\Sigma}$ are characterized by having vanishing field value on $\Sigma$, while those in $Q_{\Sigma}$ have vanishing normal derivative. Also define, $K_{\Sigma}=L_{\Sigma}/P_{\Sigma}$. This is the same \emph{field configuration space} on $\Sigma$ that we have previously introduced. We denote the quotient map by $q_{\Sigma}:L_{\Sigma}\to K_{\Sigma}$. We also note that in view of the definition of $P_{\Sigma}$ we may consider the symplectic potential as a map $[\cdot,\cdot]_{\Sigma}:L_{\Sigma}\times K_{\Sigma}\to\R$. Essentially, its first argument depends only on Neumann data (normal derivatives) and the second argument only on Dirichlet data (field values).

Let $M$ be a spacetime region. We assume that $L_M$ as a subspace of $L_{\partial M}$ is transverse both to $P_{\partial M}$ and $Q_{\partial M}$ which implies $L_{\partial M}=L_M\oplus P_{\partial M}$ and $L_{\partial M}=L_M\oplus Q_{\partial M}$. (This is generically satisfied.) Then, there is a unique linear map $\pol_M:K_{\partial M}\to L_{\partial M}$ with the properties $\pol_M(K_{\partial M})=L_{M}$ and $q_{\partial M}\circ \pol_M=\id$. That is, $\pol_M$ yields for given Dirichlet boundary data the solution in $M$ that matches the data. We may use this and expression (\ref{eq:actsymp}) from Appendix~\ref{sec:lagingreds} to rewrite the field propagator (\ref{eq:fpropa}) as,
\begin{equation}
  Z_M(\varphi)= \exp\left(\frac{\im}{2} [\pol_M(\varphi),\varphi]_{\partial M}\right) .
  \label{eq:samplbdy}
\end{equation}

Consider now the Schrödinger wave function of the vacuum state $\psi_0$ on a spacelike hypersurface $\Sigma$ in the context of the standard quantization of Section~\ref{sec:revquant}. In the specific case of Klein-Gordon theory in Minkowski space on an equal-time hypersurface (as in Section~\ref{sec:tec}), the vacuum wave function is given by \cite{hat:qft},
\begin{equation}
  \psi_0(\varphi)=\exp\left(-\frac{1}{2} \int \xd^3 x\, \varphi(x) \, (E \varphi)(x)\right) .
  \label{eq:svackg}
\end{equation}
Here we use a compact notation where $E$ is to be understood as an operator on the field configuration $\varphi$. It is defined as taking the eigenvalue $E$ on a plane wave mode with energy $E$.

In general, the vacuum is determined by a positive-definite Lagrangian subspace $L_{\Sigma}^+\subseteq L_{\Sigma}^\bC$. Now, $L_{\Sigma}^+$ is necessarily transverse to both $P_{\Sigma}^\bC$ and $Q_{\Sigma}^\bC$ as they are both complexifications of real Lagrangian subspaces (recall Proposition~\ref{prop:rdlagcompl}). There is thus a unique map $\pol_{\Sigma}^+:K_{\Sigma}\to L_{\Sigma}^\bC$ such that $\pol_{\Sigma}^+(K_{\Sigma})=L_{\Sigma}^+$ and $q_{\Sigma}\circ \pol_{\Sigma}^+=\id$. In other words, $\pol_{\Sigma}^+$ selects that element in the Lagrangian subspace $L_{\Sigma}^+$ which has the prescribed Dirichlet data on $\Sigma$. It turns out that the Schrödinger wave function of the vacuum state in general is,
\begin{equation}
  \psi_0(\varphi)= \overline{\exp\left(\frac{\im}{2} [\pol_{\Sigma}^+(\varphi),\varphi]_{\Sigma}\right)} ,
  \label{eq:svacstd}
\end{equation}
where the orientation of $\Sigma$ in the expression for the symplectic potential is that as the boundary of the region to the future of $\Sigma$.
Indeed, in the case of Klein-Gordon theory in Minkowski space, $\pol_{\Sigma}^+$ maps configurations to positive energy solutions. On these, $\partial_0$ yields eigenvalues $-\im E$. Combining this with the expression (\ref{eq:etsymp}) for the symplectic potential and changing sign for orientation reversal of $\Sigma$ recovers the vacuum wave function (\ref{eq:svackg}).

In general, the vacuum wave function (\ref{eq:svacstd}) looks exactly the same as that of a state that represents the amplitude for a spacetime region. The only difference is that the Lagrangian subspace involved in the wave function of the vacuum is a definite one while the one involved in the wave function representing the amplitude is a (complexified) real one. Note also that $\Sigma$ is the boundary of a spacetime region, namely the region to the future (or past) of $\Sigma$.

\section{A case study: The timelike hypercylinder}
\label{sec:hypcyl}

In the present section we consider a first example of boundary conditions determining a choice of vacuum on a \emph{timelike} hypersurface. Recall from Section~\ref{sec:feynpf} that a choice of \emph{temporal asymptotic vacuum} (when available) is imprinted on a global Feynman propagator $G_F$. Since we can look at the same propagator in spacetime regions with different shapes, this gives us a means of comparing boundary conditions on different types of hypersurfaces.

Concretely, we consider the standard vacuum of Klein-Gordon theory on Minkowski space. Using the notation of Section~\ref{sec:tec} its Feynman propagator can be written in the familiar form,
\begin{equation}
  G_F((t,x),(t',x'))=\im\int\frac{\xd^3 k}{(2\pi)^3 2E}
  \left(\theta(t-t') e^{-\im(E t - k x)} e^{\im(E t'-k x')} 
  +\theta(t'-t) e^{\im(E t - k x)} e^{-\im(E t' - k x')}\right) .
  \label{eq:fptint}
\end{equation}
This makes transparent the future and past boundary conditions that the propagator satisfies. As discussed in Section~\ref{sec:feynpf} we can impose these boundary conditions at the past and future boundary of a time-interval region $M=[t_1,t_2]\times\R^3$.

\subsection{Massless theory}
\label{sec:hcmassless}

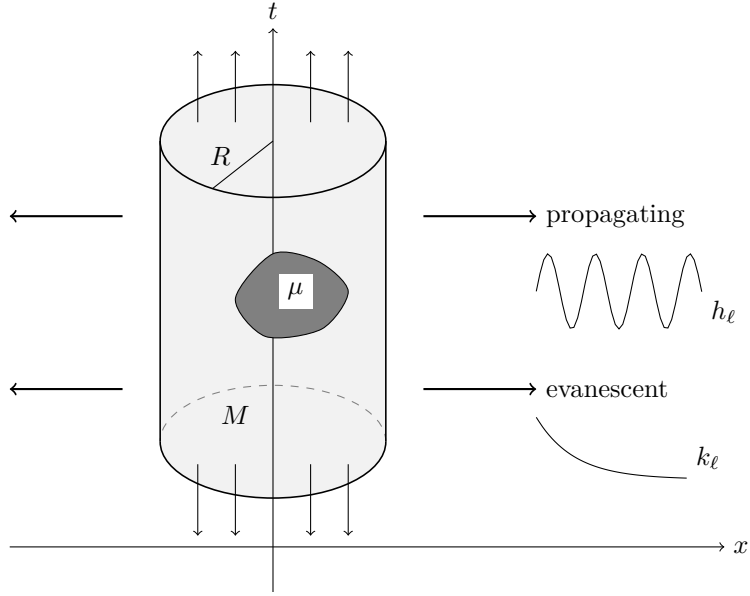
\begin{figure}[t]
\centering
\begin{tikzpicture}
\fill[gray!10] (0,0) arc (0:-180:1.5 and 0.75) -- (-3,4) arc (-180:-360:1.5 and 0.75);
\draw[->] (-1.5,-2) -- (-1.5,5.5) node [above] {$t$};
\draw[fill=gray,xshift=-4.5cm] plot [smooth cycle] coordinates { (2.5,1.9) (3,2.5) (3.6,2.4) (4,2)  (3.65,1.5) (2.9,1.4)};
\node at (-1.2,2) [fill=white] {$\mu$};
\draw[->] (-5,-1.4) -- (4.5,-1.4) node [right] {$x$};
\draw[->] (-2.5,4.25) -- (-2.5,5.2);
\draw[->] (-2,4.25) -- (-2,5.2);
\draw[->,xshift=1.5cm] (-2.5,4.25) -- (-2.5,5.2);
\draw[->,xshift=1.5cm] (-2,4.25) -- (-2,5.2);
\draw[<-,yshift=-5.5cm] (-2.5,4.25) -- (-2.5,5.2);
\draw[<-,yshift=-5.5cm] (-2,4.25) -- (-2,5.2);
\draw[<-,xshift=1.5cm,yshift=-5.5cm] (-2.5,4.25) -- (-2.5,5.2);
\draw[<-,xshift=1.5cm,yshift=-5.5cm] (-2,4.25) -- (-2,5.2);
	\draw[dashed,color=gray] (0,0) arc (0:180:1.5 and 0.75);
	\draw[semithick] (0,0) arc (0:-180:1.5 and 0.75);
	\draw[semithick] (-1.5,4) ellipse (1.5 and 0.75);
	\draw[semithick] (0,0) -- (0,4);
	\draw[semithick] (-3,0) -- (-3,4);
\node at (-2,0.35) {$M$};
\draw (-1.5,4) -- (-2.3,3.37);
\node at (-2.2,3.8) {$R$};
\draw [thick,->] (-3.5,3) -- (-5,3);
\draw [thick,->] (-3.5,0.7) -- (-5,0.7);
\draw [thick,->] (0.5,3) -- (2,3) node [right] {propagating};
\draw[variable=\y,samples at={0,0.05,...,2.25}] plot (\y+2,{2+0.5*sin(10*\y r)},0) node[below right]{$h_{\ls}$};
\draw [thick,->] (0.5,0.7) -- (2,0.7) node [right] {evanescent};
\draw[scale=0.5,domain=-0.5:3.5,smooth,variable=\x] plot ({\x+4.5},{-1+e^(-\x)}) node[above right]{$k_{\ls}$};
 \end{tikzpicture}
 \caption{The region $M$ is bounded by the hypercylinder of radius $R$. The interior solutions of Klein-Gordon theory are given in terms of spherical Bessel functions $j_{\ls}$, the vacuum on the exterior in terms of propagating $h_{\ls}$. For the massive case in the interior there are additionally the functions $a_{\ls}$ and in the exterior the evanescent $k_{\ls}$.}
\label{fig:st_hcyl_prop_evan}
\end{figure}

Let us consider now a region with timelike boundaries. We choose the \emph{timelike hypercylinder} that consists of a 3-ball $B^3_R$ of radius $R$, centered at the origin in space, and extended over all of time, i.e., $M=\R\times B^3_R$. See Figure~\ref{fig:st_hcyl_prop_evan}. The boundary of this region is the 2-sphere $S^2_R$ of radius $R$, centered at the origin in space, and extended over all of time, $\partial M=\R\times S^2_R$. This is a timelike and connected hypersurface. Suppose for the moment that the theory is massless, i.e., $m=0$. Then the space of complexified solutions $L_{\partial M}^\bC$ in a neighborhood of the boundary $\partial M$ may be parametrized in terms of modes as follows,
\begin{equation}
 \phi(t,r,\Omega)=\int_{-\infty}^\infty\xd E\, \frac{p}{4\pi}
  \sum_{\ls,\ms}\left(\phi_{\ls,\ms}^\text{a}(E) h_{\ls}(p r) e^{-\im E t} Y_{\ls}^{\ms}(\Omega)
  +\phi_{\ls,\ms}^\text{b}(E)\; \overline{h_{\ls}(p r)} e^{\im E t} Y_{\ls}^{-\ms}(\Omega)\right) .
  \label{eq:hcprop}
\end{equation}
Here $\Omega$ denotes angular coordinates on the sphere, $r$ is the radial distance from the origin, $p=|E|$. The functions $Y^{\ms}_{\ls}$ are the usual spherical harmonics. $h_{\ls}$ and $\overline{h_{\ls}}$ are the spherical Bessel functions of the third kind, also known as Hankel functions. The sum runs over the usual angular momentum quantum numbers, i.e., $\ls=0,1,\ldots,$ and $\ms=-\ls,-\ls+1,\ldots,\ls-1,\ls$. The integral runs over $\R$. Real solutions are those with $\phi_{\ls,\ms}^\text{a}(E)=\overline{\phi_{\ls,\ms}^\text{b}(E)}$. The \emph{symplectic potential} (\ref{eq:sympot}) is the bilinear form $L_{\partial M}^\bC\times L_{\partial M}^\bC\to\bC$ given by,
\begin{equation}
  [\phi,\xi]_{\partial M} = - R^2 \int  \xd t\,\xd\Omega\,
  \xi(t,R,\Omega) (\partial_r \phi)(t,R,\Omega) .
  \label{eq:hcsymp}
\end{equation}
The \emph{symplectic form} (\ref{eq:sympfrompotlin}) on $L_{\partial M}^\bC$ is its anti-symmetrization,
\begin{align}
  \omega_{\partial M}(\phi,\xi)  & =\frac{R^2}{2}\int \xd t\,\xd\Omega\,
 \left(\phi(t,R,\Omega) (\partial_r \xi)(t,R,\Omega)- \xi(t,R,\Omega) (\partial_r \phi)(t,R,\Omega)\right) \label{eq:symhc} \\
& = \int_{-\infty}^{\infty}\xd E\frac{\im p}{8\pi}\sum_{\ls,\ms}
 \left(\xi_{\ls,\ms}^\text{a}(E)\phi_{\ls,\ms}^\text{b}(E)
 -\xi_{\ls,\ms}^\text{b}(E)\phi_{\ls,\ms}^\text{a}(E)\right) .
 \label{eq:symhcprop}
\end{align}
Note that the \emph{global} (complexified) solution space can be identified not with $L_{\partial M}^\bC$, but with $L_M^\bC$, the space of solutions in the interior of the hypercylinder. Viewed as a subspace of $L_{\partial M}^\bC$, this consists of the solutions with the property $\phi_{\ls,\ms}^{\text{a}}(E)=\phi_{\ls,-\ms}^{\text{b}}(-E)$. To see this note that the Bessel functions $h_{\ls}$ and $\overline{h_{\ls}}$ are singular at the origin in space, i.e., at $r=0$. Only in the linear combination $j_{\ls}=(h_{\ls}+\overline{h_{\ls}})/2$ do the singularities cancel. $j_{\ls}$ is the spherical Bessel function of the first kind. It is also easy to verify using expression (\ref{eq:symhcprop}) that the subspace $L_M\subseteq L_{\partial M}$ (or $L_M^{\bC}\subseteq L_{\partial M}^{\bC}$) is a Lagrangian subspace, as required.

Using the modes of the expansion (\ref{eq:hcprop}) we may rewrite the Feynman propagator (\ref{eq:fptint}) equivalently as follows \cite{CoOe:smatrixgbf},
\begin{multline}
  G_F((t,r,\Omega),(t',r',\Omega'))
  =\int_{-\infty}^{\infty}\xd E\, \frac{\im p}{2 \pi} \sum_{\ls,\ms} Y_{\ls}^{\ms}(\Omega) Y_{\ls}^{-\ms}(\Omega')
  e^{-\im E t} e^{\im E t'} \\
  \left(\theta(r-r') h_{\ls}(p r) j_{\ls}(p r')
  +\theta(r'-r) h_{\ls}(p r') j_{\ls}(p r)\right) .
  \label{eq:fphcprop}
\end{multline}
From this expression it is easy to read off that the boundary condition on the hypercylinder is given by the restriction to the $h_{\ls}$ modes, i.e., the solutions that satisfy $\phi^{\text{b}}_{\ls,\ms}(E)=0$. We denote this subspace by $L_{\overline{\partial M}}^+$. It is easy to verify with the symplectic form (\ref{eq:symhcprop}) that this is indeed a Lagrangian subspace. The inner product (\ref{eq:iplc}) takes the form,
\begin{equation}
  (\phi,\xi)_{\overline{\partial M}}=
 \int_{-\infty}^{\infty}\xd E\frac{p}{2\pi}\sum_{\ls,\ms}
  \left(\overline{\phi_{\ls,\ms}^\text{a}(E)}\xi_{\ls,\ms}^\text{a}(E)
  -\overline{\phi_{\ls,\ms}^\text{b}(E)}\xi_{\ls,\ms}^\text{b}(E)\right) .
 \label{eq:iphcprop}
\end{equation}
(Note the opposite orientation of $\partial M$ and thus opposite sign compared to (\ref{eq:symhc}).) This is easily seen to be positive-definite on $L_{\overline{\partial M}}^+$. The corresponding complex structure $J_{\overline{\partial M}}$ multiplies the $h_{\ls}$ modes with $\im$ and the $\overline{h_{\ls}}$ modes by $-\im$. Thus, we nicely recover a notion of vacuum in the sense of Section~\ref{sec:stquant} as a boundary condition on the timelike hypercylinder. Moreover, this is the standard vacuum of Klein-Gordon theory in Minkowski space.

The definite Lagrangian subspace $L_{\overline{\partial M}}^+\subseteq L_{\partial M}^\bC$ exhibits the properties of a conventional vacuum in another important sense, namely that of behavior with respect to infinitesimal normal evolution, compare Section~\ref{sec:tec}. Consider the radial derivative $\partial_r$ as acting on the space of (complexified) germs of solutions $L_{\partial M}^\bC$. Since only the spherical Bessel functions in the expansion (\ref{eq:hcprop}) depend on the radius $r$ it is sufficient to focus on these exclusively. A convenient presentation is the following \cite[10.49(i)]{NIST:DLMF},
\begin{equation}
  h_{\ls}(pr)=e^{\im p r}\sum_{k=0}^{\ls} \frac{\im^{k-\ls-1} (\ls+k)!}{2^k k! (\ls-k)! (pr)^{k+1}},
  \qquad \overline{h_{\ls}(pr)}=e^{-\im p r}\sum_{k=0}^{\ls} \frac{(-\im)^{k-\ls-1} (\ls+k)!}{2^k k! (\ls-k)! (pr)^{k+1}} .
  \label{eq:hlexpansion}
\end{equation}
From this it is clear that applying $\partial_r$ to a mode containing $h_{\ls}$, yields modes containing $h_{\ls}$ (for different values of $\ls$). Similarly this happens for the modes containing $\overline{h_{\ls}}$. That is, the Lagrangian subspace $L_{\overline{\partial M}}^+$ and its complement $L_{\overline{\partial M}}^-=\overline{L_{\overline{\partial M}}^+}$ are both invariant subspaces of the operator $\partial_r$. What is more, for large radius $r$ the functions $r\mapsto h_{\ls}(pr)$ and $r\mapsto \overline{h_{\ls}(pr)}$ become approximate eigenfunctions of the derivative operator $\partial_r$ with imaginary eigenvalues $\im p$ and $-\im p$ respectively. Correspondingly, the operator $-J_{\partial M}\partial_r$ acquires the momentum $p$ as its approximate eigenvalue on all modes. Conversely, this criterion uniquely determines the complex structure $J_{\overline{\partial M}}$. This is in precise analogy to the role the time-evolution operator $\partial_0$ plays in the selection of the vacuum in traditional quantization, compare Section~\ref{sec:tec} and particularly expression (\ref{eq:tevolh}). (There is a difference in overall sign, due to the different sign for spacelike and timelike directions in the Lorentzian metric.)

\subsection{Massive theory}
\label{sec:hcmassive}

We continue to consider the Klein-Gordon theory in Minkowski space, but drop the restriction for the field to be massless, i.e., we allow $m\neq 0$. In that case the integrals over the energy $E$ have to be restricted to $E^2\ge m^2$ in expressions (\ref{eq:hcprop}), (\ref{eq:symhcprop}), (\ref{eq:fphcprop}), and (\ref{eq:iphcprop}). What is more, in these same expressions, as well as in those to follow, $p$ becomes $p=\sqrt{|E^2-m^2|}$. Correspondingly, the modes in the expansion (\ref{eq:hcprop}) do no longer describe the space of solutions near radius $R$ completely, but only a subspace $L_{\partial M}^{\text{p},\bC}\subseteq L_{\partial M}^{\bC}$. This is the subspace of \emph{propagating solutions} that show oscillatory behavior in space at large radius. We henceforth denote all corresponding solutions spaces from the previous section with a superscript $^{\text{p}}$ to indicate this. There are in addition \emph{evanescent solutions} that show exponential behavior in space. We denote the subspace of these by $L_{\partial M}^{\text{e},\bC}\subseteq L_{\partial M}^{\bC}$. To parametrize them we use modified spherical Bessel functions given by, $k_{\ls}(z)=-\im^{\ls}\pi h_{\ls}(\im z)/2$ and $\tilde{k}_{\ls}(z)=k_{\ls}(-z)$. For real arguments these are real functions. The $k_{\ls}$ modes decay exponentially with radius, while the $\tilde{k}_{\ls}$ modes grow exponentially. Complexified evanescent solutions near $\partial M$ (in fact anywhere away from $r=0$) may be parametrized as,
\begin{equation}
 \phi(t,r,\Omega)=\int_{-m}^{m}\xd E\, \frac{p}{4\pi} e^{-\im E t}
  \sum_{\ls,\ms} Y_{\ls}^{\ms}(\Omega) \left(\phi_{\ls,\ms}^\text{x}(E) k_{\ls}(p r)
+ \phi_{\ls,\ms}^\text{i}(E) \tilde{k}_{\ls}(p r) \right) .
\end{equation}
The real subspace $L_{\partial M}^{\text{e}}\subseteq L_{\partial M}^{\text{e},\bC}$ is given by the conditions $\phi_{\ls,\ms}^\text{x}(E)=\phi_{\ls,-\ms}^\text{x}(-E)$ and $\phi_{\ls,\ms}^\text{i}(E)=\phi_{\ls,-\ms}^\text{i}(-E)$. The symplectic form (\ref{eq:symhc}) on $L_{\partial M}^{\text{e},\bC}$ is,
\begin{equation}
  \omega_{\partial M}^{\text{e}}(\phi,\xi)
= \int_{-m}^{m}\xd E\frac{\pi p}{32}\sum_{\ls,\ms}
 \left(\xi_{\ls,\ms}^\text{x}(E)\phi_{\ls,-\ms}^\text{i}(-E)
 -\xi_{\ls,\ms}^\text{i}(E)\phi_{\ls,-\ms}^\text{x}(-E)\right) .
 \label{eq:symhcev}
\end{equation}
The symplectic form vanishes between propagating and evanescent solutions. Thus, the total symplectic form on $L_{\partial M}^{\bC}=L_{\partial M}^{\text{p},\bC}\oplus L_{\partial M}^{\text{e},\bC}$ is simply the sum of (\ref{eq:symhcprop}) and (\ref{eq:symhcev}). The solutions $k_{\ls}$ and $\tilde{k}_{\ls}$ are both singular at the origin, but the linear combination $a_{\ls}(z)\defeq\im^{-\ls+2}\pi j_{\ls}(\im z)=\tilde{k}_{\ls}(z)+(-1)^{\ls} k_{\ls}(z)$ is well behaved there. Thus, the subspace $L_M^{\text{e},\bC}\subseteq L_{\partial M}^{\text{e},\bC}$ of solutions in the interior of the hypercylinder is determined by the condition $\phi_{\ls,\ms}^\text{i}(E)=(-1)^{\ls} \phi_{\ls,\ms}^\text{x}(E)$. As is easy to see with (\ref{eq:symhcev}), $L_M^{\text{e}}\subseteq L_{\partial M}^{\text{e}}$ is a Lagrangian subspace. Correspondingly, $L_M=L_M^{\text{p}}\oplus L_M^{\text{e}}$ is a Lagrangian subspace of $L_{\partial M}$, as required. Note that there are no evanescent modes that are well defined and bounded in all of Minkowski space. So, in contradistinction to the propagating modes they do not appear in the global space of solutions.

In the massive case the Feynman propagator in terms of hypercylinder modes also receives a contribution from the evanescent solutions. Thus, in addition to the right hand side of expression (\ref{eq:fphcprop}), the evanescent contribution is given by \cite{CoOe:smatrixgbf},
\begin{multline}
  G_F^{\text{e}}((t,r,\Omega),(t',r',\Omega'))
  =-\int_{-m}^{m}\xd E\, \frac{p}{\pi^3} \sum_{\ls,\ms} Y_{\ls}^{\ms}(\Omega) Y_{\ls}^{-\ms}(\Omega')
  e^{-\im E t} e^{\im E t'} \\
  \left(\theta(r-r') k_{\ls}(p r) a_{\ls}(p r')
  +\theta(r'-r) k_{\ls}(p r') a_{\ls}(p r)\right) .
  \label{eq:fphcev}
\end{multline}
From this expression we can read off immediately that the boundary condition on the hypercylinder for the evanescent solutions is given by the restriction to the $k_{\ls}$ modes, i.e., the solutions that satisfy $\phi_{\ls,\ms}^\text{i}(E)=0$. We denote this subspace by $L_{\overline{\partial M}}^{\text{e},+}\subseteq L_{\partial M}^{\text{e},\bC}$. It is easy to verify with (\ref{eq:symhcev}) that this is indeed a Lagrangian subspace. The same is thus true for the corresponding subspace comprising both propagating and evanescent solutions, $L_{\overline{\partial M}}^+=L_{\overline{\partial M}}^{\text{p},+}\oplus L_{\overline{\partial M}}^{\text{e},+}\subseteq L_{\partial M}^{\bC}$. However, the inner product (\ref{eq:iplc}) takes on the evanescent solutions the form,
\begin{equation}
  (\phi,\xi)_{\overline{\partial M}}^{\text{e}}
= \int_{-m}^{m}\xd E\frac{\im \pi p}{8}\sum_{\ls,\ms}
 \left(\overline{\phi_{\ls,\ms}^\text{x}(E)}\xi_{\ls,\ms}^\text{i}(E)
 -\overline{\phi_{\ls,\ms}^\text{i}(E)}\xi_{\ls,\ms}^\text{x}(E)\right) .
 \label{eq:iphcev}
\end{equation}
This is clearly not positive-definite on $L_{\overline{\partial M}}^{\text{e},+}$. In fact, $L_{\overline{\partial M}}^{\text{e},+}$ is a \emph{neutral subspace} for this inner product, i.e., the inner product vanishes on any two elements from this subspace. In contrast to $L_{\overline{\partial M}}^{\text{p},+}\subseteq L_{\partial M}^{\text{p},\bC}$, $L_{\overline{\partial M}}^{\text{e},+}$ is the \emph{complexification} of a \emph{real} Lagrangian subspace of $L_{\partial M}^{\text{e}}$, recall that the $k_{\ls}$ modes are real modes. By inspection of the inner product (\ref{eq:iplc}) it is clear that in this case the isotropy property (\ref{eq:isotrop}) implies the neutrality property.

It is also instructive to consider the action of the radial derivative operator $\partial_r$ on the space $L_{\partial M}^{\text{e}}$ of evanescent germs. (On the space $L_{\partial M}^{\text{p}}$ of propagating germs the previous discussion of the massless case fully applies.) To this end we consider the following presentation of the relevant spherical Bessel functions of the third kind \cite[10.49(ii)]{NIST:DLMF},
\begin{equation}
  k_{\ls}(pr)= e^{-p r}\sum_{k=0}^{\ls} \frac{\pi\, (\ls+k)!}{2^{k+1} k! (\ls-k)! (pr)^{k+1}},
  \qquad \tilde{k}_{\ls}(pr)= e^{p r}\sum_{k=0}^{\ls} \frac{\pi\, (-1)^{k+1} (\ls+k)!}{2^{k+1} k! (\ls-k)! (pr)^{k+1}} .
\end{equation}
From this it is easy to see that the Lagrangian subspace $L_{\overline{\partial M}}^{\text{e},+}$ build out of Bessel functions $k_{\ls}$, as well as the complementary Lagrangian subspace $L_{\overline{\partial M}}^{\text{e},-}$ build out of Bessel functions $\tilde{k}_{\ls}$ are both invariant under $\partial_r$. What is more, for large radius $r$ the functions $r\mapsto k_{\ls}(pr)$ and $r\mapsto \tilde{k}_{\ls}(pr)$ become approximate eigenfunctions of the derivative operator $\partial_r$ with real eigenvalues $-p$ and $p$ respectively. (Note that $p$ is positive by definition.) This reflects the fact that the elements of $L_{\overline{\partial M}}^{\text{e},+}$ (which contain $k_{\ls}$-modes) are decaying solutions for large radius $r$, while those of $L_{\overline{\partial M}}^{\text{e},-}$ (which contain $\tilde{k}_{\ls}$-modes) are growing solutions for large radius $r$. We note that the standard vacuum corresponds to selecting the asymptotically decaying solutions rather than the growing ones.

We conclude that the boundary condition on the hypercylinder in the massless Klein-Gordon theory in Minkowski space induced by the standard Feynman propagator can be interpreted following Section~\ref{sec:stquant} precisely as a choice of vacuum. In particular, we can construct a corresponding Hilbert space of states associated to the hypercylinder following the prescription of Section~\ref{sec:modecs}. This is not so in the massive theory. There we encounter instead a Lagrangian subspace that is not definite. In particular, there is no corresponding Hilbert space in the sense of Section~\ref{sec:modecs} and the framework \cite{Oe:feynobs} referred to in Section~\ref{sec:stquant} does not apply. (For the incomplete subset of modes that are propagating, there is of course a definite Lagrangian subspace and corresponding Hilbert space, as in the massless theory.)

\section{The vacuum as a Lagrangian subspace}
\label{sec:vaclag}

We hope to have presented sufficient evidence in previous sections to convince the reader of the picture alluded to in the title of this work. We lay out this picture in the present section. It turns out to be fruitful to start with classical field theory.

\subsection{Classical field theory}
\label{sec:classglag}

\begin{figure}
\centering
\begin{tikzpicture}[scale=0.9]
\fill[gray!10] (-0.5,1) rectangle (4,3);
\draw[->] (-0.5,0) -- (4,0) node [right] {$x$};
\draw[->] (0,-0.5) -- (0,4) node [above] {$t$};
\draw [-] (-0.5,1) node [left] {$t_1$}  -- (4,1);
\draw [-] (-0.5,3) node [left] {$t_2$}  -- (4,3);
\node at (0.4,2) {$M$};
\node at (2.3, 2.3) {real Lagrangian};
\node at (2.3, 1.8) {subspace $L_M$};
\node at (2,-1) {(a)};
 \end{tikzpicture}
 \hspace{0.3cm}
\begin{tikzpicture}[scale=0.9]
\fill[gray!10] (-0.5,2) rectangle (4,4);
\draw[->] (-0.5,0) -- (4,0) node [right] {$x$};
\draw[->] (0,-0.5) -- (0,4) node [above] {$t$};
\draw [-] (-0.5,2)   -- (4,2) node [midway, below] {$\Sigma$};
\node at (0.4,2.7) {$\psi_0$};
\draw[->] (0.4,2.5) -- (0.4,2.1);
\node at (2.3, 3.7) {definite Lagrangian};
\node at (2.3, 3.2) {subspace $\tilde{L}_X$};
\node at (3.2,2.5) {$X$};
\node at (2,-1) {(b)};
 \end{tikzpicture}
 \hspace{0.3cm}
\begin{tikzpicture}[scale=0.9]
\fill[gray!10] (-0.5,2) rectangle (4,-0.5);
\draw[->] (-0.5,0) -- (4,0) node [right] {$x$};
\draw[->] (0,-0.5) -- (0,4) node [above] {$t$};
\draw [-] (-0.5,2)   -- (4,2) node [midway, above] {$\Sigma$};
\node at (0.4,1.3) {$\psi_0$};
\draw[->] (0.4,1.5) -- (0.4,1.9);
\node at (2.3, 0.8) {definite Lagrangian};
\node at (2.3, 0.3) {subspace $\tilde{L}_X$};
\node at (3.2,1.5) {$X$};
\node at (2,-1) {(c)};
 \end{tikzpicture}
  \caption{(a) Solutions in a time-interval region $M$ yield a \emph{real} Lagrangian subspace $L_M\subseteq L_{\partial M}$. (b) The future vacuum state $\psi_0$ on a spacelike hypersurface $\Sigma$ is encoded through a \emph{definite} Lagrangian subspace $\tilde{L}_X\subseteq L_{\partial X}^{\bC}=L_{\Sigma}^{\bC}$. (c) The corresponding past vacuum.}
  \label{fig:stdgeom_lagsub}
\end{figure}
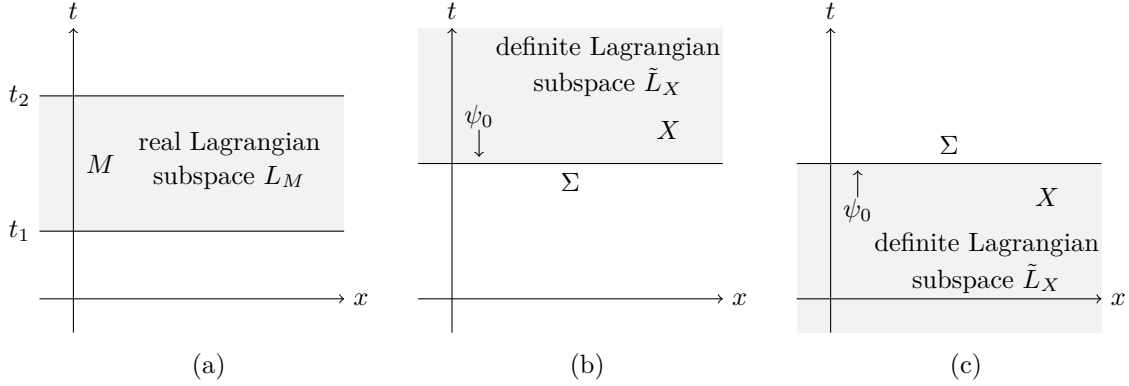

Recall from Section~\ref{sec:classlag} that given a spacetime region $M$ the space $L_M$ of solutions in the interior should give rise to a Lagrangian subspace of the space $L_{\partial M}$ of germs of solutions on the boundary. This works generally well if $M$ is a \emph{compact} spacetime region. This works also well for certain non-compact regions. The principal example we have discussed is a time-interval region in a globally hyperbolic spacetime, see Figure~\ref{fig:stdgeom_lagsub}.a. We have also seen in Section~\ref{sec:hypcyl} that this works well for Klein-Gordon theory on the hypercylinder in Minkowski space. On the other hand, suppose we consider the \emph{exterior} of a time-interval region in a globally hyperbolic spacetime. For simplicity just restrict to the future part, i.e., the region to the future of a given spacelike hypersurface. Generically, the solutions show an oscillating behavior and there seems to be no natural way to single out a subspace that is Lagrangian when restricted to germs on the hypersurface. Indeed, as we have previously emphasized, all germs on the hypersurface, which are nothing but initial data, correspond to valid global solutions. On the other hand, as recalled in Section~\ref{sec:revquant} we do attach a Lagrangian subspace to the hypersurface, encoding a vacuum for the quantum theory. What is more, the asymptotic perspective reviewed in Section~\ref{sec:feynpf} suggests that we may think of this Lagrangian subspace as associated to the region that is all of the future of the hypersurface, see Figure~\ref{fig:stdgeom_lagsub}.b. (Figure~\ref{fig:stdgeom_lagsub}.c shows the corresponding ``past vacuum''.) However, the Lagrangian subspace in question is a \emph{definite} Lagrangian subspace of the \emph{complexification} of the space of germs on the hypersurface.

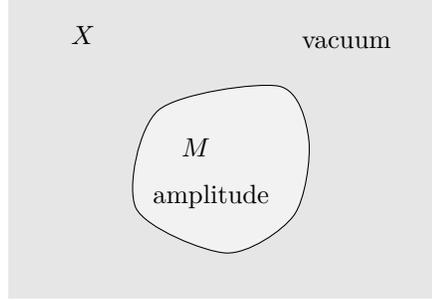
\begin{figure}
\centering
\begin{tikzpicture}
\fill[gray!20] (-1,-1) rectangle (4.75,3);
\draw[fill=gray!10] plot [smooth cycle] coordinates { (0.7,0.2) (1,1.5) (2.6,1.82) (3,1.1)  (2.8,.1) (1.9,-.4)};
\node at (1.7,0.35) {amplitude};
\node at (3.5,2.4) {vacuum};
\node at (1.5,1) {$M$};
\node at (0,2.5) {$X$};
 \end{tikzpicture}
  \caption{Generic setting with a compact region $M$ and non-compact exterior $X$. In $M$ we have an amplitude map while in $X$ we have a vacuum.}
  \label{fig:generic-ampl-vac}
\end{figure}

The intuitive picture that the reader should have in mind is the following. We consider a compact (or maybe just temporally compact) region $M$ and its complement $X$, see Figure~\ref{fig:generic-ampl-vac}. The physics in $M$ is encoded in its amplitude map $\rho_M$, corresponding to a real Lagrangian subspace $L_M\subseteq L_{\partial M}$. The physics in $X$ on the other hand gives rise to a vacuum state $\psi_0$ on $\partial M$, corresponding to a definite Lagrangian subspace $\tilde{L}_X\subseteq L^{\bC}_{\partial X}= L^{\bC}_{\partial M}$. We claim that both types of Lagrangian subspaces and correspondingly the two apparently very different notions of ``interior'' solution space and ``exterior'' vacuum should be seen as instances of the same general structure. Using the observation that a Lagrangian subspace of a real symplectic vector space canonically complexifies to give rise to Lagrangian subspace of the complexified symplectic vector spaces, we arrive at the following picture, generalizing that of Section~\ref{sec:classlag}. Thus, to each hypersurface $\Sigma$ we associate a complex symplectic vector space $(L_{\Sigma}^{\bC},\omega_{\Sigma})$ that arises as the complexification of the real symplectic vector space $(L_{\Sigma},\omega_{\Sigma})$ of germs of solutions previously considered. (We do not distinguish notationally between the symplectic form $\omega_{\Sigma}$ and its complexification.) To each region $M$ we associate a \emph{complex} vector space $\tilde{L}_M$ with a map $\tilde{L}_M\to L_{\partial M}^\bC$ such that the image is a complex Lagrangian subspace. All further ingredients of the classical field theory picture of Section~\ref{sec:classlag} are generalized accordingly. Notably, the composition rule encoded in the exact sequence (\ref{eq:exacomp}) is required to hold. The full generalization of the corresponding axiomatic system for classical field theory \cite[Section~4.1]{Oe:holomorphic} is provided in Appendix~\ref{sec:caxioms}.

It remains to specify how the spaces $\tilde{L}_M$ arise. Clearly, for the regions where the picture of Section~\ref{sec:classlag} works, we want to maintain it. That is, if the space $L_M$ is defined and gives rise to a Lagrangian subspace of $L_{\partial M}$, we take $\tilde{L}_M$ to be its complexification $L_{M}^\bC$. This should hold for all compact regions as well as for certain non-compact ones (such as the examples of the time interval and the hypercylinder).
The example of the massive Klein-Gordon theory in Section~\ref{sec:hcmassive} suggests how this generalizes to another important class of non-compact regions. Namely those regions where it makes sense to impose a condition of \emph{asymptotic decay} on the field. Indeed, we have seen that those evanescent waves that decay with large radius in the region \emph{exterior} to the hypercylinder form a real Lagrangian subspace. Moreover, this Lagrangian subspace precisely encodes the (evanescent part of) the standard vacuum of Klein-Gordon theory in Minkowski space.
On the other hand, $\tilde{L}_M$ should clearly give rise to a corresponding definite Lagrangian subspace when the region is one occurring in the context of conventional quantization. The prime example here is a region occurring as the future or past of a spacelike hypersurface in a globally hyperbolic spacetime. We return to a more dedicated discussion of the spaces $\tilde{L}_M$ in Section~\ref{sec:vchoice}.

\subsection{Quantum field theory}

The picture of unifying the real Lagrangian subspaces of solutions and the definite Lagrangian subspaces coming from quantization in a common framework is intriguing in the classical theory. It becomes compelling in the quantum theory. 

The quantum analog of a real Lagrangian subspace $L_M$ of solutions in a region $M$ is the \emph{amplitude map} $\rho_M$ for that region. Transferring the picture outlined above to the quantum theory, if for a region $M$ the space $\tilde{L}_M$ is not a real, but a definite Lagrangian subspace, its quantum analog is still a \emph{(generalized) amplitude map} for that region. Except in this case it is one traditionally called a \emph{vacuum}.
As recalled in Section~\ref{sec:schroedvac} we may conveniently encode the amplitude in terms of a Schrödinger wave function $\hat{\rho}_M$. There is a single formula that describes this wave function in all cases,
\begin{equation}
  \overline{\hat{\rho}_M(\varphi)}= Z_M(\varphi) = \exp\left(\frac{\im}{2} [\widetilde{\pol}_M(\varphi),\varphi]_{\partial M}\right) .
  \label{eq:univprop}
\end{equation}
Here, $\widetilde{\pol}_M:K_{\partial M}^\bC\to L_{\partial M}^\bC$ is the unique linear map such that $\widetilde{\pol}_M(K_{\partial M}^\bC)=\tilde{L}_{M}$ and $q_{\partial M}\circ\widetilde{\pol}_M=\id$. (We use the notation with the tilde to emphasize that we necessarily work with the complexified spaces now.) This reduces to the real case (\ref{eq:samplbdy}) of a traditional amplitude map if $\tilde{L}_M$ is a complexified real Lagrangian subspace. It reduces to the wave function of a traditional vacuum (\ref{eq:svacstd}) if $\tilde{L}_{M}$ is a definite Lagrangian subspace.

Note that the Schrödinger representation is chosen here for convenience and brevity of presentation. There is nothing special about it with respect to the unification of the concepts of amplitude and vacuum. This unification becomes equally manifest in other representations such as the holomorphic one. In the latter case this may be seen using the machinery developed in \cite{Oe:holomorphic,Oe:schroedhol,Oe:feynobs}. However, presenting the details of this is not essential to our argument and thus outside the scope of the present article.

The key formula for evaluating the amplitude for a spacetime region $M$ (which might be all of spacetime) in the vacuum state and with a Weyl observable inserted in $M$ remains (\ref{eq:veweyl}). The applicability of the formula is expanded, however. Consider the typical situation that $M$ is a spacetime region with $L_M\subseteq L_{\partial M}$ real Lagrangian and denote by $X$ its exterior, i.e., the complementary spacetime region which we take to be non-compact, see again Figure~\ref{fig:generic-ampl-vac}. Let $D:K_M\to\bC$ be a linear observable and $F=\exp(\im D)$ the corresponding Weyl observable. Denote as before (compare Section~\ref{sec:stquant}) by $A_M^D$ the solution space of the inhomogeneous equations of motion for the action $S_M+D$ and $A_M^D\oplus\im L_M$ its complexification. Now, the vacuum $\tilde{L}_X$ is a (definite or not) Lagrangian subspace of $L_{\partial X}^\bC=L_{\partial M}^\bC$. Then, the vacuum expectation value of the observable $F$ is given by formula (\ref{eq:veweyl}), where $\eta\in A_M^D\oplus\im L_M\cap \tilde{L}_X$ is unique. This gives the right result whether $\tilde{L}_X$ is a real or a definite Lagrangian subspace of $L_{\partial M}^\bC$ as we have seen in the hypercylinder example of Section~\ref{sec:hypcyl}.

The formula can further be extended to the boundary-less case that $M$ is all of spacetime if the ingredients are suitably interpreted. Then, $\eta\in A_M^D\oplus\im L_M\cap \tilde{L}_{\partial M}$, where $A_M^D$ and $L_M$ are now interpreted as global inhomogeneous and respectively homogeneous solutions. On the other hand we take $L_{\partial M}^\bC$ to be the complexified space of \emph{asymptotic} solutions, i.e., the space of solutions that are defined only near the ``boundary'' $\partial M$ of spacetime ``at infinity''. $\tilde{L}_M\subseteq L_{\partial M}^\bC$ is the Lagrangian subspace implementing the choice of vacuum. We are deliberately vague here about what we mean by this ``boundary''. One might think of boundaries of conformal compactifications for example. We leave this to be made precise in future work.

\section{The choice of vacuum}
\label{sec:vchoice}

In this section we turn to the question of how a space $\tilde{L}_M$ of (complexified) solutions is associated to a spacetime region $M$. As already discussed in Sections~\ref{sec:classlag} and \ref{sec:classglag}, for compact spacetime regions and for certain non-compact ones the space $L_M$ of solutions in $M$ is naturally a Lagrangian subspace of the space $L_{\partial M}$ of germs on the boundary $\partial M$. $\tilde{L}_M$ is then just its complexification $L_M^\bC$. As also indicated in Section~\ref{sec:classglag} this generalizes to non-compact regions where it makes sense to impose a condition of asymptotic decay on the field.
For other non-compact regions a decay condition does not make sense. The prime example is the future or past half of a globally hyperbolic spacetime bounded by a spacelike hypersurface. Indeed, as this is the home for the conventional notion of vacuum we do not expect to obtain a real Lagrangian subspace, but rather a definite Lagrangian subspace. In Section~\ref{sec:tec} we have reviewed how the choice of this Lagrangian subspace can be addressed via an infinitesimal approach tied to a notion of time evolution. In the present section we integrate this into a more general infinitesimal approach for the selection of (generalized) vacua (Section~\ref{sec:infvac}). Subsequently we seek to formalize an approach based on generalized asymptotic decay conditions (Section~\ref{sec:fpvac}).

\subsection{Infinitesimal approach}
\label{sec:infvac}

Suppose we have a non-compact spacetime region $X$ where imposing an asymptotic decay condition on the field yields a real Lagrangian subspace $L_X\subseteq L_{\partial X}$. Then, we might try to determine if a given germ in $L_{\partial X}$, when extended into $X$ decays or not by looking at its behavior in a neighborhood of the hypersurface $\partial X$. More concretely, we suppose we have a normal derivative operator $\partial_n$, pointing to the interior of $X$, acting on $L_{\partial X}$. Then, spectrally decomposing $\partial_n$, we should expect it to exhibit positive or negative eigenvalues for solutions depending on whether they grow or decay under continuation into $X$. Thus, $L_X$ should correspond to the subspace of $L_{\partial X}$ that is spanned by the eigenspaces for negative eigenvalues of $\partial_n$.

The evanescent solutions of the massive Klein-Gordon field provide an example for precisely this situation. Recall Section~\ref{sec:hcmassive}. The region in question is the exterior $X$ of the solid hypercylinder $M$ and the normal derivative is the radial one, which we denoted $\partial_r$. The Lagrangian subspace $L_X^{\text{e}}$ ($^{\text{e}}$ indicates restriction to evanescent solutions) is called $L_{\overline{\partial M}}^{\text{e},+}=L_{\partial X}^{\text{e},+}$ there. As we have seen there, it is the subspace of $L_{\partial X}=L_{\partial M}$ where $\partial_r$ has negative asymptotic eigenvalues, namely $-p$, where $p$ is the total momentum.

In general, a normal derivative operator $\partial_n$ on the space $L_{\Sigma}$ of germs on a hypersurface (boundary or not) does not need to be spectrally decomposable. Recall in particular the situation corresponding to a conventional vacuum, see Section~\ref{sec:tec} and \ref{sec:quanthyp}. More specifically, consider Klein-Gordon theory on an equal-time hypersurface $\Sigma$ in Minkowski space. Then the square $\partial_0^2$ of the time-derivative operator $\partial_0$ has negative eigenvalues. Thus, to decompose $\partial_0$ we need to complexify $L_{\Sigma}$ to $L_{\Sigma}^\bC$ and the eigenvalues are imaginary. The solutions are oscillatory. So, how does this fit together with the idea of decaying solutions? On the face of it it does not, but we can ``make'' some solutions decaying by \emph{Wick rotation}. That is, we rotate $\partial_0$ by $90^\circ$ in the complex plane by multiplying it with $-\im$ (or $\im)$. We then declare those solutions ``decaying'' that correspond to negative eigenvalues of $-\im\partial_0$ (or $\im\partial_0$). Recall from Section~\ref{sec:tec} that choosing $-\im$ makes the corresponding Hamiltonian positive. Multiplying instead by $\im$ makes it negative. The right choice depends in general on conventions and coherence. The Lagrangian subspace of the selected solutions is what we declare $\tilde{L}_M$ to be if $\Sigma$ is the boundary of $M$. This subspace is by construction not the complexification of a real subspace, but rather a complex subspace that is complementary to its complex conjugate. What is more, looking at the explicit structure of the symplectic form (\ref{eq:sfkget}) this should even be a definite subspace with respect to the inner product (\ref{eq:iplc}).

This approach should work in somewhat more generality, e.g., for spacelike hypersurfaces in globally hyperbolic spacetime (bounding a half of spacetime). It also works in the example of the Klein-Gordon theory on the hypercylinder for propagating solutions, recall Section~\ref{sec:hcmassless}. There, the asymptotic eigenvalues of the radial derivative operator $\partial_r$ are $\im p$ and $-\im p$, with $p$ the total momentum. Indeed, the vacuum Lagrangian subspace $\tilde{L}_X$ in the exterior region $X$, called there $L_{\overline{\partial M}}^+=L_{\partial X}^+$, corresponds precisely to selecting the negative asymptotic eigenspaces of $\im \partial_r$. (The choice of $\im\partial_r$ rather than $-\im\partial_r$ ensures consistency with the standard vacuum and is related to the opposite signature between spacelike and timelike directions in the metric.)

Summarizing, the infinitesimal approach starts with a normal derivative operator $\partial_n$ on the space $L_{\Sigma}$ of germs on the hypersurface $\Sigma$ in question. If the operator has real eigenvalues, the eigenspaces with negative eigenvalues are chosen to form a real Lagrangian subspace. If the operator has imaginary eigenvalues, the space $L_{\Sigma}$ is complexified to $L_{\Sigma}^{\bC}$, the derivative operator is \emph{Wick rotated} to $\im \partial_n$ (or $-\im \partial_n$) and the eigenspaces with negative eigenvalues of the rotated operator are taken to form a Lagrangian subspace. Of course, the derivative operator might not have eigenvalues that are restricted to being either only real or only imaginary. One might imagine that in some cases one can define sensible generalizations of the Wick rotation procedure. On the other hand, it is clear that the presented method has important limitations. We also have not elaborated on how exactly one would choose a normal derivative operator in general. There might not even be a good choice for such an operator.

\subsection{Asymptotic field propagator approach}
\label{sec:fpvac}

We proceed to describe a method for selecting a vacuum that relies on a particular formalization of the notion of asymptotic boundary conditions. To this end we suppose that the spacetime region $X$ in question is foliated into a collection of hypersurfaces $\Sigma_s$, indexed by a real parameter $s\in [0,\infty)$, such that $\Sigma_0$ coincides with the boundary $\partial X$. Denote by $X_s$ the subregion of $X$ enclosed between the hypersurfaces $\Sigma_0$ and $\Sigma_s$. That is, the boundary $\partial X_s$ decomposes into the disjoint union $\Sigma_0\sqcup \overline{\Sigma_s}$. Note that we orient $\Sigma_0$ and $\Sigma_s$ both as boundaries of a region spanned by the larger values of $s$. Thus, as a boundary component of $X_s$, $\Sigma_s$ has opposite orientation. Correspondingly, the space of germs on the boundary decomposes as $L_{\partial X_s}=L_{\Sigma_0}\times L_{\overline{\Sigma_s}}$. We consider the Schrödinger field propagator (\ref{eq:pifp}), which we write as a function of two arguments, separating $K_{\partial X_s}=K_{\Sigma_0}\times K_{\Sigma_s}$,
\begin{equation}
  Z_{X_s}(\varphi,\varphi')=\int_{K_{X_s}, \phi|_0=\varphi, \phi|_s=\varphi'}\xD\phi\, e^{\im S_{X_s}(\phi)} .
\end{equation}

The basic idea of the asymptotic method for obtaining the field propagator for the region $X$ is rather simple. We want to impose that the field vanish at ``infinity'', that is for $s\to\infty$. To this end we cut the region off to obtain $X_s$, make the field vanish on $\Sigma_s$, then send $s$ to infinity. That is, we define $Z_X$ as,
\begin{equation}
  Z_X(\varphi)\defeq \lim_{s\to\infty} Z_{X_s}(\varphi,0) .
  \label{eq:fplim}
\end{equation}
Using further ingredients of the Schrödinger representation (compare Section~\ref{sec:schroedvac}) we can make this more precise.
We assume the Lagrangian subspace property $L_{X_s}\subseteq L_{\partial X_s}$ for the region $X_s$. This (together with genericity) implies then that we have a linear map $\pol_{X_s}:K_{\partial X_s}\to L_{\partial X_s}$ satisfying $\pol_{X_s}(K_{\partial X_s})= L_{X_s}$ and $q_{\partial X_s}\circ \pol_{X_s}=\id$. We rewrite $\pol_{X_s}$ in terms of components as, $\pol_{[0,s]}^0:K_{\Sigma_0}\times K_{\Sigma_s}\to L_{\Sigma_0}$ and $\pol_{[0,s]}^s:K_{\Sigma_0}\times K_{\Sigma_s}\to L_{\Sigma_s}$. We can then reexpress the field propagator with (\ref{eq:samplbdy}) as,
\begin{equation}
    Z_{X_s}(\varphi,\varphi')= \exp\left(\frac{\im}{2} \left([\pol_{[0,s]}^0(\varphi,\varphi'),\varphi]_{0}-[\pol_{[0,s]}^s(\varphi,\varphi'),\varphi']_{s}\right)\right) .
\end{equation}
Then, the limit (\ref{eq:fplim}) works out to,
\begin{equation}
  Z_X(\varphi)= \lim_{s\to\infty} \exp\left(\frac{\im}{2} [\pol_{[0,s]}^0(\varphi,0),\varphi]_{0}\right) .
  \label{eq:fplimsp}
\end{equation}
On the other hand, a Lagrangian subspace $L_X\subseteq L_{\partial X}$ yields a corresponding linear map $\pol_X:K_{\partial X}\to L_{\partial X}$ satisfying $\pol_X(K_{\partial X})=L_X$ and $q_{\partial X}\circ\pol_X=\id$. (Note $\partial X=\Sigma_0$.) Conversely, $\pol_X$ uniquely determines $L_X$. Thus, comparing the field propagator (\ref{eq:samplbdy}) determined by $L_X$ with the limit (\ref{eq:fplimsp}) given above we get,
\begin{equation}
  \pol_X(\varphi)=\lim_{s\to\infty} \pol_{[0,s]}^0(\varphi,0) .
  \label{eq:limpol}
\end{equation}
In particular, the limit on the right hand side determines thus the Lagrangian subspace $L_X$.

We consider again the example of the evanescent waves for massive Klein-Gordon theory on the hypercylinder in Minkowski space (Section~\ref{sec:hcmassive}). Thus, fix a radius $R$ and let $X$ be the region exterior to the hypercylinder of radius $R$. We foliate this region by radius so that the leaf indexed by a parameter value $s\in [0,\infty)$ is the hypercylinder of radius $R+s$. One can then show that the relevant expression in the exponential of the right hand side of (\ref{eq:fplimsp}) is \cite{CoOe:smatrixgbf},
\begin{multline}
  [\pol_{[0,s]}^0(\varphi,0),\varphi]_{0} \\
  =R^2\int\xd t\,\xd\Omega\, \varphi(t,\Omega)
  \left(\frac{p\left(k_{\ls}\left(p (R+s)\right)\tilde{k}_{\ls}'\left(p R\right)-\tilde{k}_{\ls}\left(p (R+s)\right) k_{\ls}'\left(p R\right)\right)}%
       {k_{\ls}\left(p (R+s)\right)\tilde{k}_{\ls}\left(p R\right) - \tilde{k}_{\ls}\left(p (R+s)\right) k_{\ls}\left(p R\right)}\, \varphi\right)(t,\Omega) .
  \label{eq:ehcpol}
\end{multline}
We use here a rather compact notation where the expression forming the fraction is understood as acting as an operator on the field configuration $\varphi$ on the right-hand side. More precisely, the expression represents the eigenvalue of the operator when $\varphi$ is decomposed into spherical harmonics (indexed by $\ls$ and $\ms$) in space and into Fourier modes in time (indexed by the energy $E$ with $p=\sqrt{m^2-E^2}$). Noting that $k_{\ls}$ decays exponentially in radius we have
\begin{equation}
  \lim_{s\to\infty} k_{\ls}\left(p(R+s)\right)=0 .
\end{equation}
Thus, for the limit of (\ref{eq:ehcpol}) we get,
\begin{equation}
  [\pol_X(\varphi),\varphi]_0=\lim_{s\to\infty}[\pol_{[0,s]}^0(\varphi,0),\varphi]_{0}
  =R^2\int\xd t\,\xd\Omega\, \varphi(t,\Omega)
  \left(\frac{p\, k_{\ls}'\left(p R\right)}%
{k_{\ls}\left(p R\right)}\, \varphi\right)(t,\Omega) .
\end{equation}
Comparing this with expression (\ref{eq:hcsymp}) for the symplectic potential it is easy to see that the fraction is precisely the radial derivative if $\pol_X:K_{\partial X}\to L_{\partial X}$ maps onto the $k_{\ls}$-modes. Thus, the Lagrangian subspace $L_X\subseteq L_{\partial X}$ is precisely that of the decaying $k_{\ls}$-modes, called $L_{\overline{\partial M}}^{\text{e},+}=L_{\partial X}^{\text{e},+}$ in Section~\ref{sec:hcmassive}, recovering the standard vacuum. As expected, this is also in agreement with the infinitesimal approach.

We proceed to the situation where imposing a vanishing boundary condition at $s\to\infty$ does not make sense, because solutions show an oscillating behavior in the $s$-direction. The prime example is of course that of a globally hyperbolic spacetime with $s$ representing a time parameter, i.e., the setting of the conventional notion of vacuum. As we have seen in Section~\ref{sec:infvac}, we can deal with this in the infinitesimal setting through a Wick rotation. We proceed to describe an asymptotic version of the notion of \emph{Wick rotation}. Here, rather than to the corresponding derivative operator we apply the Wick rotation to the evolution parameter $s$. Thus, we consider the same limit (\ref{eq:fplim})
for the field propagator, but only after rotating $s$ to $-\im s$ (or $\im s$). As is easy to see, choosing $-\im s$ (or $\im s$) here corresponds precisely to choosing $-\im \partial_s$ (or $\im\partial_s$) in the infinitesimal approach. That is, $-\im s$ corresponds to a positive (generalized) Hamiltonian. Repeating the same steps as above, we arrive for the field propagator at the expression,
\begin{equation}
  Z_X(\varphi)= \lim_{s\to\infty} \exp\left(\frac{\im}{2} [\pol_{[0,-\im s]}^0(\varphi,0),\varphi]_{0}\right) .
  \label{eq:fplimisp}
\end{equation}
The analogue of equation (\ref{eq:limpol}) is now,
\begin{equation}
  \widetilde{\pol}_X(\varphi)=\lim_{s\to\infty} \pol_{[0,-\im s]}^0(\varphi,0) .
  \label{eq:limipol}
\end{equation}
We have used here the notation $\widetilde{\pol}_X:K_{\partial X}^\bC\to L_{\partial X}^\bC$ with a tilde to emphasize that we necessarily work with the complexified spaces now. This determines the complex Lagrangian subspace $\tilde{L}_X$ through the properties of $\widetilde{\pol}_X$. These are, $\widetilde{\pol}_X(K_{\partial X})=\tilde{L}_X$ and $q_{\partial X}\circ \widetilde{\pol}_X=\id$.

The recovery of the standard vacuum from a time foliation of Minkowski space for Klein-Gordon theory is the prime example for the Wick rotated setting. Thus, let $X$ be the region to the future of the spacelike hypersurface at $t=0$ and set $s=t$. With the conventions of Section~\ref{sec:tec} we then have \cite{Oe:timelike},
\begin{equation}
  [\pol_{[0,s]}^0(\varphi,0),\varphi]_{0}=\int \xd^3 x\, \varphi(x) \left(\frac{E \cos(E s)}{\sin(E s)}\, \varphi\right)(x) .
  \label{eq:mslhprop}
\end{equation}
Again, we understand the fractional expression as an operator specified through its eigenvalues on plane wave modes labeled by $E$. Performing the Wick rotation, with $s$ replaced by $-\im s$ we get,
\begin{equation}
  [\pol_{[0,-\im s]}^0(\varphi,0),\varphi]_{0}=\int \xd^3 x\, \varphi(x) \left(\frac{\im E \cosh(E s)}{\sinh(E s)}\, \varphi\right)(x) .
\end{equation}
Taking the limit $s\to \infty$, we obtain,
\begin{equation}
  [\widetilde{\pol}_X(\varphi),\varphi]_0=\lim_{s\to\infty} [\pol_{[0,-\im s]}^0(\varphi,0),\varphi]_{0}=\int \xd^3 x\, \varphi(x) \left(\im E\, \varphi\right)(x) .
  \label{eq:mslhlim}
\end{equation}
Comparing this with expression (\ref{eq:etsymp}) for the symplectic potential we recognize $-\im E$ as the eigenvalues of the time-derivative operator $\partial_0$ for the positive energy modes. Thus, we recover precisely the definite Lagrangian subspace $L^+\subseteq L_{\partial X}^{\bC}$ as the image of $\widetilde{\pol}_X$ corresponding to the standard vacuum of Klein-Gordon theory in Minkowski space. In particular, the Schrödinger wave function for the vacuum is (\ref{eq:svackg}).

As a second example, already discussed in the infinitesimal approach, we consider the propagating waves of Klein-Gordon theory on the hypercylinder in Minkowski space (Section~\ref{sec:hcmassless}). Geometrically the setting is the same as that for the evanescent waves. That is, the region $X$ is the exterior of the hypercylinder and we foliate it by radius. Thus, $X_s$ is bounded by hypercylinders of radius $R$ and $R+s$. Then \cite{Oe:kgtl},
\begin{multline}
  [\pol_{[0,s]}^0(\varphi,0),\varphi]_{0}\\
  = R^2\int\xd t\,\xd\Omega\, \varphi(t,\Omega)
  \left(\frac{p\left(h_{\ls}\left(p (R+s)\right)\overline{h_{\ls}'\left(p R\right)}-\overline{h_{\ls}\left(p (R+s)\right)} h_{\ls}'\left(p R\right)\right)}%
       {h_{\ls}\left(p (R+s)\right)\overline{h_{\ls}\left(p R\right)} - \overline{h_{\ls}\left(p (R+s)\right)} h_{\ls}\left(p R\right)}\, \varphi\right)(t,\Omega) .
  \label{eq:phcpol}
\end{multline}
We Wick rotate $s$. In accordance with the infinitesimal approach we replace $s$ by $\im s$ (rather than by $-\im s$). As can be read off from expression (\ref{eq:hlexpansion}) we then have,
\begin{equation}
  \lim_{s\to\infty} h_{\ls}\left(p(R+\im s)\right)=0 .
\end{equation}
This yields,
\begin{equation}
  [\widetilde{\pol}_X(\varphi),\varphi]_0=\lim_{s\to\infty}[\pol_{[0,\im s]}^0(\varphi,0),\varphi]_{0}
  = R^2\int\xd t\,\xd\Omega\, \varphi(t,\Omega)
  \left(\frac{p\, h_{\ls}'\left(p R\right)}{h_{\ls}\left(p R\right)}\, \varphi\right)(t,\Omega) .
\end{equation}
Comparison with (\ref{eq:hcsymp}), taking into account the opposite sign due to opposite orientation lets us conclude that $\widetilde{\pol}_X$ maps onto the $h_{\ls}$-modes. That is, the Lagrangian subspace $\tilde{L}_X$ is the space called $L_{\overline{\partial M}}^+=L_{\partial X}^+$ in Section~\ref{sec:hcmassless}, spanned by the $h_{\ls}$-modes. As expected, this is in agreement with the infinitesimal approach and recovers the standard vacuum.

As for the infinitesimal method we remark that the asymptotic method will have a limited range of applicability. In some cases a suitable modification might be apparent. We consider an example of such a modification in Section~\ref{sec:Rindlerhyp}.

\section{Further examples}
\label{sec:examples}

In the present section we showcase the applicability of our framework for defining and encoding vacua (Section~\ref{sec:vaclag}) to different types of regions and on different types of hypersurfaces, complementing the hypercylinder example of Section~\ref{sec:hypcyl}. At the same time, we demonstrate the methods for vacuum selection of Section~\ref{sec:vchoice}. All the presented examples concern Klein-Gordon theory in simple and well understood settings, including among other elements timelike hypersurfaces, Euclidean space and curved spacetime.

\subsection{Minkowski space: vacuum on timelike hyperplanes}
\label{sec:tlhp}

In preceding sections the standard Minkowski vacuum was reviewed on spacelike hyperplanes (Section~\ref{sec:revquant}) and constructed on the hypercylinder (Section~\ref{sec:hypcyl}). In this section we consider the vacuum on a timelike hyperplane. This type of hypersurface in Klein-Gordon theory was first considered in \cite{Oe:timelike}, where only propagating modes where taken into account. As before, we consider coordinates $(t,x_1,x_2,x_3)$ where $t$ is time and $x_i$ are spatial coordinates. Consider the hyperplane $\Sigma$ given by $x_1=0$. We are interested in the vacuum ``to the right'', i.e., corresponding to the region with $x_1\ge 0$. The space of (complexified) solutions around $\Sigma$ contains both propagating and evanescent solutions. The space $L_{\Sigma}^{\text{p},\bC}$ of complexified \emph{propagating solutions} may be parametrized as,
\be
\phi(t,x_1,x_2,x_3) =  \int \frac{\xd k_2 \, \xd k_3}{(2 \pi)^{3/2}} \int_{|E|>\tilde{k}} \xd E\, e^{\im E t - \im k_2 x_2 - \im k_2 x_3} \left( \phi^{\text{a}}(E, k_2,k_3)  e^{\im k x_1} + \phi^{\text{b}}(E, k_2,k_3)  e^{-\im k x_1} \right),
\label{eq:KG-exp}
\ee
where $\tilde{k} = \sqrt{k_2^2+k_3^2+m^2}$,  $k = \sqrt{|E^2-\tilde{k}^2|}$ and the integrals in $k_2$ and $k_3$ are over $\R$. These solutions are characterized by an oscillatory behavior in the $x_1$ direction.
The space $L_{\Sigma}^{\text{e},\bC}$ of complexified \emph{evanescent solutions},
\be
\phi(t,x_1,x_2,x_3) = \int \frac{\xd k_2 \, \xd k_3}{(2 \pi)^{3/2}} \int_{-\tilde{k}}^{\tilde{k}} \xd E \, e^{\im E t - \im k_2 x_2 - \im k_2 x_3} \left( \phi^\text{x}(E, k_2,k_3)  e^{- k x_1} + \phi^\text{i}(E, k_2,k_3)  e^{ k x_1} \right),
\label{eq:KG-exp-ev}
\ee
consists of real exponential modes in $x_1$.
The symplectic potential (\ref{eq:sympot}) on the hypersurface $\Sigma$ (as a boundary of a region with $x_1\ge 0$) is the bilinear form
$L_{\Sigma} \times L_{\Sigma} \to \R$  given by
\begin{equation}
  [\phi,\phi'] =   \int \xd t \, \xd x_2 \, \xd x_3 \, \phi'(t,x_1,x_2,x_3) (\partial_1 \phi) (t,x_1,x_2,x_3),
  \label{eq:tlsympot}
\end{equation}
where $\partial_1$ denotes the partial derivative w.r.t.\ the coordinate $x_1$, namely it is the normal derivative w.r.t.\ the hyperplane $\Sigma$. The symplectic form (\ref{eq:sympfrompotlin}) on $\Sigma$ is
\be
\omega(\phi,\xi) = \frac12 \int \xd t \, \xd x_2 \, \xd x_3 \left( \xi(t,x_1,x_2,x_3) (\partial_1 \phi)(t,x_1,x_2,x_3) - \phi(t,x_1,x_2,x_3) (\partial_1 \xi)(t,x_1,x_2,x_3)\right) .
\label{eq:tlhpsymf}
\ee
It is easy to show that it vanishes between propagating and evanescent solutions. We thus have a corresponding orthogonal decomposition $L_{\Sigma}^{\bC} =  L_{\Sigma}^{\text{p},\bC} \oplus L_{\Sigma}^{\text{e},\bC}$ of the complexified solution space, analogous to the massive theory on the hypercylinder, compare Section~\ref{sec:hcmassive}. In contrast to the hypercylinder, however, evanescent waves occur on the timelike hyperplane even in the massless theory.

We proceed to consider the Lagrangian subspaces determining the vacuum on the timelike hyperplane. For propagating modes, the symplectic form (\ref{eq:tlhpsymf}) reads,
\begin{multline}
\omega^{\text{p}}(\phi,\xi) = \im  \int \xd k_2 \, \xd k_3 \, k \int_{|E|>\tilde{k}}   \xd E \\ \left(\phi^{\text{a}}(E,k_2,k_3) \xi^{\text{b}}(-E,-k_2,-k_3) - \phi^{\text{b}}(E,k_2,k_3) \xi^{\text{a}}(-E,-k_2,-k_3)\right) .
\end{multline}
This yields the inner product (\ref{eq:iplc}),
\be
(\phi,\xi)^{\text{p}} = 4 \int \xd k_2 \, \xd k_3 \, k \int_{|E|>\tilde{k}}   \xd E \, \left( \overline{\phi^{\text{a}}(E,k_2,k_3)} \xi^{\text{a}}(E,k_2,k_3) - \overline{\phi^{\text{b}}(E,k_2,k_3)} \xi^{\text{b}}(E,k_2,k_3)\right) .
\ee
This results to be positive-definite for $\phi^{\text{b}}(E,k_2,k_3)=0$, which defines a positive-definite Lagrangian subspace $L_{\Sigma}^{\text{p},+}\subseteq L_{\Sigma}^{\text{p},\bC}$ which we take to define the propagating part of the vacuum.

On the other hand, in the whole spacetime region with $x_1\ge 0$, of the evanescent modes only those decaying as $e^{-kx_1}$ are well behaved. This corresponds to the condition $\phi^\text{i}(E,k_2,k_3)=0$. These solutions form a Lagrangian subspace $L_{\Sigma}^{\text{e},+} \subseteq L_{\Sigma}^{\text{e},\bC}$ as can be seen from the expression of the symplectic form,
\begin{equation}
\omega^{\text{e}}(\phi,\xi) = \int \xd k_2 \, \xd k_3 \, k \int_{-\tilde{k}}^{\tilde{k}} \xd E \, \left(\phi^\text{i}(E,k_2,k_3) \xi^\text{x}(-E,-k_2,-k_3) - \phi^\text{x}(E,k_2,k_3) \xi^\text{i}(-E,-k_2,-k_3)\right).
\end{equation}
As for the case of the massive hypercylinder studied in Section (\ref{sec:hcmassive}), the subspace $L_{\Sigma}^{\text{e},+}$ turns out to be a neutral subspace since the inner product (\ref{eq:iplc}) vanishes when computed on two elements of this subspace.

With the Lagrangian subspaces determined we can almost immediately write down the corresponding presentation of the \emph{Feynman propagator}. Its propagating part reads,
\begin{multline}
  G_F^{\text{p}}((t,x_1,x_2,x_3),(t',x_1',x_2',x_3')) = \int \xd k_2 \, \xd k_3\,\frac{\im}{2k (2 \pi)^3} \int_{|E|>\tilde{k}} \xd E \\ e^{\im Et}e^{-\im Et'} 
  e^{-\im k_2 x_2}e^{\im k_2 x_2'}e^{-\im k_3 x_3}e^{\im k_3 x_3'}
\left( \theta(x_1-x_1') e^{\im k x_1}e^{-\im k x_1'} + \theta(x_1'-x_1) e^{\im k x_1'}e^{-\im k x_1} \right) .
\end{multline}
The evanescent part is,
\begin{multline}
  G_F^{\text{e}}((t,x_1,x_2,x_3),(t',x_1',x_2',x_3')) = \int \xd k_2 \, \xd k_3\,\frac{1}{2k (2 \pi)^3} \int_{-\tilde{k}}^{\tilde{k}} \xd E \\
  e^{\im Et}e^{-\im Et'}e^{-\im k_2 x_2}e^{\im k_2 x_2'}e^{-\im k_3 x_3}e^{\im k_3 x_3'}
\left( \theta(x_1-x_1') e^{- k x_1}e^{ k x_1'} + \theta(x_1'-x_1) e^{- k x_1'}e^{ k x_1} \right).
\end{multline}
The total propagator is the sum, $G_F=G_F^{\text{p}}+ G_F^{\text{e}}$. This is easily verified to be the Feynman propagator of the standard vacuum, equivalent to expression (\ref{eq:fptint}) as well as the sum of (\ref{eq:fphcprop}) and (\ref{eq:fphcev}).

Let us consider the infinitesimal approach to vacuum selection presented in Section \ref{sec:infvac}. As already mentioned, the evanescent modes $e^{-kx_1}$ that form the complexified real Lagrangian subspace $L_{\Sigma}^{\text{e},+}$ satisfy an asymptotic decay condition. They are indeed precisely the eigenmodes of the normal derivative operator $\partial_1$ with negative eigenvalues. In contrast, the propagating modes show oscillatory behavior. But, Wick rotating the derivative operator to $\im\partial_1$ yields negative eigenvalues precisely for the modes in the positive-definite Lagrangian subspace $L_{\Sigma}^{\text{p},+}$. Thus, we obtain exactly the expected results.

We proceed to the asymptotic field propagator approach for vacuum selection (Section~\ref{sec:fpvac}). Let $X$ be the region to the right of $\Sigma_0=\Sigma$, i.e., with $x_1\ge 0$. This region can be foliated in terms of timelike hyperplanes. In particular we consider the subregion $X_s= \R \times [0,s] \times \R^2$ relative to the coordinates $(t,x_1,x_2,x_3)$, namely the region bounded by two hyperplanes defined by constant values of the spatial coordinate $x_1$: the hyperplane $\Sigma_0$ and the hyperplane $\Sigma_s$ given by $x_1=s$. 
The parameter $s$ used in Section~\ref{sec:fpvac} now represents the spatial coordinate $x_1$. The boundary solution space $L_{\partial X_s}$ decomposes as $L_{\partial X_s}=L_{\Sigma_0} \oplus L_{\overline{\Sigma_s}}$.

For the evanescent modes, we expect to get the Schrödinger vacuum wave function via the limit (\ref{eq:fplimsp}) of the field propagator. The relevant quantity is,
\begin{equation}
[\text{pol}^0_{[0,s]}(\varphi,0),\varphi]_0 =  \int \xd t \, \xd \tilde{x}\, \varphi(t,\tilde{x}) \left( \frac{- k \cosh (k s)}{\sinh (k s)} \varphi \right) (t,\tilde{x}),
\end{equation}
where $\tilde{x}$ is a compact notation for the coordinates $(x_2,x_3)$ and the fraction represents an operator acting on a mode expansion of the field configuration. Then limit (\ref{eq:fplimsp}) yields the evanescent part of the field propagator for $X$,
\begin{equation}
Z_X^{\text{e}}(\varphi) =\exp\left(-\frac{\im}{2}\int \xd t \, \xd \tilde{x} \, \varphi(t,\tilde{x}) \left( k  \varphi \right) (t,\tilde{x})\right) .
\end{equation}
From this we can read off by comparison to the negative of (\ref{eq:tlsympot}) the eigenvalues $-k$ of the decaying modes, in accordance with our previous results.

For propagating modes, an asymptotic vanishing condition at spatial infinity can not be imposed due to the oscillatory character of solutions. Instead, we Wick rotate the parameter $s$ describing spatial evolution to $\im s$, corresponding to rotating $\partial_1$ to $\im\partial_1$. That is, we consider the field propagator in $X$ as the limit (\ref{eq:fplimisp}), except with $\im s$ instead of $-\im s$. This is,
\begin{align}
Z_X^{\text{p}}(\varphi) &= \lim_{s \to \infty} \exp \left(\frac{\im}{2} [\text{pol}^0_{[0, \im s]}(\varphi, 0),\varphi]_0 \right), \\
&=  \lim_{s \to \infty} \exp \left( \frac{\im}{2}  \int \xd t \, \xd \tilde{x} \, \varphi(t,\tilde{x}) \left( \frac{- k \cos ( k \im s)}{\sin ( k \im s)} \varphi \right) (t,\tilde{x}) \right),\\
&=  \lim_{s \to \infty} \exp \left( \frac{\im}{2}  \int \xd t \, \xd \tilde{x} \, \varphi(t,\tilde{x}) \left( \frac{- k \cosh ( k s)}{\im\sinh ( k s)} \varphi \right) (t,\tilde{x}) \right),\\
&= \exp \left( - \frac{1}{2}  \int \xd t \, \xd \tilde{x} \, \varphi(t,\tilde{x}) \left( k  \varphi \right) (t,\tilde{x}) \right).
\end{align}
Again, this can be seen to be in accordance with our previous results.

\subsection{Rindler space}
\label{sec:Rindler}

We consider the massive Klein-Gordon theory in 2-dimensional Rindler space, which corresponds to the right wedge of 2-dimensional Minkowski space. It can be described by the metric $\xd s^2= \rho^2 \xd \eta^2- \xd \rho^2$, where $\rho \in \R^+$ and $\eta \in \R$ are the spatial and temporal Rindler coordinates respectively. The Klein-Gordon equation takes the form,
\be
\left( -\rho \partial_{\rho}\rho \partial_{\rho}  +\partial_{\eta}^2 + m^2 \rho^2\right) \phi(\eta,\rho)=0,
\ee
where $\partial_{\eta}$ and $\partial_{\rho}$ denote partial derivatives w.r.t.\ the coordinates $\eta$ and $\rho$ respectively.
The temporal part is solved by the modes $e^{p \eta}$ with $p \in \bC$. Propagating modes (in time) correspond to imaginary values of $p$. The solution of the spatial part is given by the modified Bessel functions of the first and second kind, $I_{p}(m \rho)$ and $K_{p}(m \rho)$ respectively, which constitute a pair of linear independent solutions of the modified Bessel equation.\footnote{The linear independence of $I_p$ and $K_p$ is manifested by the Wronskian between these function (see expression~10.28.2 of \cite{NIST:DLMF}),
 \be
I_{ p}' (z) K_{ p}(z) - I_{ p} (z) K_{ p}'(z) = \frac{1}{z},
 \ee
 with the prime denoting the derivative of the Bessel function with respect to its argument. Notice that another pair of linear independent solutions is provided by $I_{-p}$ and $I_{p}$, that satisfy (formula 10.28.1 of \cite{NIST:DLMF}),
\be
 I_{ p} (z) I_{- p}'(z) -  I_{ p}' (z) I_{ - p}(z) = -\frac{2 \sin(p \pi)}{z \pi}.
 \ee 
}
These functions present the following asymptotic behavior: According to formulas~10.30.1 and 10.30.2 of \cite{NIST:DLMF}, for $\Re p>0$ and small values of their argument ($z \ll 1$) respectively,
\be
I_p(z) \sim \frac{(z/2)^p}{\Gamma(p+1)}, \qquad K_p(z) \sim \frac12 \Gamma(p) (z/2)^{-p} .
\label{eq:as-mBf-0}
\ee
Their behavior for large values of the argument ($z \gg 1$) is given by formulas~10.34.1 and 10.34.2 of \cite{NIST:DLMF},
 \be
 I_{ p}(z) \sim \frac{e^{z}}{\sqrt{2 \pi z}}, \quad K_{p}(z) \sim \sqrt{\frac{\pi}{2z}} \, e^{-z} .
 \label{eq:as-mBf-inf}
 \ee

\subsubsection{Region bounded by equal Rindler-time hyperplanes}

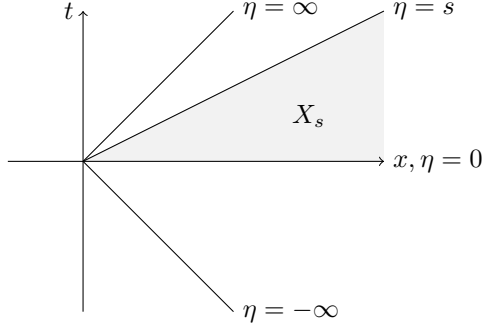
\begin{figure}
\centering
\begin{tikzpicture}[scale=1]
\fill[gray!10] (4,0) -- (0,0) -- (4,2) -- cycle;
\draw[->] (0,-2) -- (0,2) node [left] {$t$};
\draw[->] (-1,0) -- (4,0) node [right] {$x, \eta=0$};
\draw (2,-2) node [right] {$\eta=-\infty$} -- (0,0) -- (2,2)node [right] {$\eta=\infty$};
\draw (0,0) -- (4,2)node [right] {$\eta=s$};
\node at (3,0.6) {$X_s$};
\end{tikzpicture}
\caption{Rindler space corresponds to the right wedge of Minkowski space; it is bounded by the half-lines $\eta=\pm \infty$. In this spacetime we consider the region $X_s=[0,s] \times \R^+$.}
\label{fig:Rin-1}
\end{figure}

In order to obtain the vacuum on a half-line of constant Rindler time, say $\Sigma_0:\{\eta=0\}$, we consider the region $X$ in the future of $\Sigma_0$ and the subregion $X_s$ bounded by $\Sigma_0$ and the half-line $\Sigma_s : \{ \eta=s\}$, namely the region  $X_s=[0,s] \times \R^+$ (see Figure~\ref{fig:Rin-1}). The region is unbounded in space and therefore only the modified Bessel functions of the second kind, $K_p$, are admissible modes for the field due to the asymptotic decay (\ref{eq:as-mBf-inf}); moreover, because of the behavior (\ref{eq:as-mBf-0}) imaginary values of the momentum $p$ must be selected. So, complexified  solutions contain only propagating modes and can be written as,
\be
\phi(\eta,\rho) = \int_0^\infty \xd p \left( \phi^{\text{a}} (p) e^{-\im p \eta} + \phi^{\text{b}}(p) e^{\im p \eta} \right) K_{\im p} (m \rho).
\label{eq:rindlerhpmodes}
\ee
The symplectic potential (\ref{eq:sympot}) and symplectic form (\ref{eq:sympfrompotlin}) on a half-line $\Sigma$ (oriented as the boundary of a region to the future) take the form, respectively,
\begin{align}
[\phi,\xi] &= - \int_0^\infty \frac{\xd \rho}{\rho} \xi(\eta,\rho) (\partial_{\eta} \phi)(\eta,\rho),\\
\omega(\phi,\xi) &= \frac12 \int_0^\infty \frac{\xd \rho}{\rho} \left( \phi(\eta,\rho) (\partial_{\eta} \xi)(\eta,\rho) - \xi(\eta,\rho) (\partial_{\eta} \phi)(\eta,\rho)\right),\\
&= \frac{\im}{2} \int_0^\infty \xd p\, \frac{\pi^2}{\sinh(p \pi)} \left(\phi^{\text{a}}(p) \xi^{\text{b}}(p) - \xi^{\text{a}}(p) \phi^{\text{b}}(p) \right) .
\end{align}
Here, the following identity has been used,
\be
\int_0^{\infty} \frac{\xd \rho}{\rho} K_{\im p}(\rho)  K_{\im p'}(\rho) = \frac{\pi^2}{2 p \sinh(p \pi)} \delta(p-p').
\ee
Consequently, the inner product (\ref{eq:iplc}) is,
\begin{align}
(\phi,\xi) &= 2 \im \int_0^{\infty} \frac{\xd \rho}{\rho}  \left( \overline{\phi(\eta,\rho)} \partial_{\eta} \xi(\eta,\rho) - \xi(\eta, \rho) \partial_{\eta} \overline{\phi(\eta,\rho)} \right), \\
&=2 \int_0^{\infty} \xd p\, \frac{\pi^2}{ \sinh(p \pi)} \left( \overline{\phi^{\text{a}}(p)} \xi^{\text{a}}(p) - \overline{\phi^{\text{b}}(p)} \xi^{\text{b}}(p)\right).
\end{align}
As the solutions (\ref{eq:rindlerhpmodes}) show oscillatory behavior in Rindler time, we expect a vacuum determined by a definite Lagrangian subspace. Indeed, the inner product is positive-definite on the subspace $L_{\Sigma}^+=\{ \phi \in L_{\Sigma}^{\bC} : \phi^{\text{b}}(p)=0 \}$ consisting of the modes $e^{-\im p \eta}K_{\im p}(m\rho)$. On these, the Wick rotated derivative operator $-\im\partial_{\eta}$ has negative eigenvalues $-p$. We are thus in a situation completely analogous to the standard vacuum on an equal-time hyperplane in all of Minkowski space, see Section~\ref{sec:tec}.
The expression of the corresponding Feynman propagator has been obtained in \cite{CoRa:qftrindler},
\begin{multline}
G_F((\eta,\rho),(\eta',\rho')) \\ = \im \int_0^{\infty} \xd p \, \frac{\sinh(p \pi)}{\pi^2} \left( \theta(\eta-\eta') e^{-\im p \eta} e^{\im p \eta'} + \theta(\eta'-\eta) e^{-\im p \eta'} e^{\im p \eta} \right)
K_{\im p}(m\rho) K_{\im p}(m\rho') .
\label{eq:Fp-R-eta}
\end{multline}

The alternative determination of the vacuum via the asymptotic limit of the field propagator in $X_s$ proceeds also in analogy to the Minkowski case, see Section~\ref{sec:fpvac}.
The relevant quantity here is,
\begin{equation}
[\text{pol}^0_{[0,s]}(\varphi,0),\varphi]_0 =  \int_0^{\infty} \frac{\xd \rho }{ \rho}\, \varphi(\rho) \left( \frac{p \cos (p s)}{\sin (p s)} \varphi \right) (\rho),
\end{equation}
where $p$ is to be understood as the operator $p=\sqrt{(\rho \partial_{\rho})^2 -m^2}$. We implement the Wick rotation $s \mapsto -\im s$ and obtain according to formula~(\ref{eq:fplimisp}),
\begin{align}
Z_X(\varphi) &= \lim_{s \to \infty} \exp \left( \frac{\im}{2} [\text{pol}^0_{[0, -\im s]}(\varphi, 0),\varphi]_0\right), \\
&=  \lim_{s \to \infty} \exp \left(-\frac{1}{2} \int_0^{\infty} \frac{\xd \rho }{ \rho} \, \varphi(\rho) \left( \frac{p \cosh ( p s)}{\sinh ( p s)} \varphi \right) (\rho) \right),\\
&= \exp \left( - \frac{1}{2}  \int_0^{\infty} \frac{\xd \rho }{ \rho} \, \varphi(\rho) \left( p  \varphi \right) (\rho) \right).
\end{align}
The last expression coincides, apart from a normalization factor, with formula~(35) of \cite{CoRa:qftrindler} giving the Schrödinger vacuum wave function of the scalar field on the equal Rindler time half-line. (Note that complex conjugation is trivial here.)

\subsubsection{Region bounded by hyperbolas}
\label{sec:Rindlerhyp}

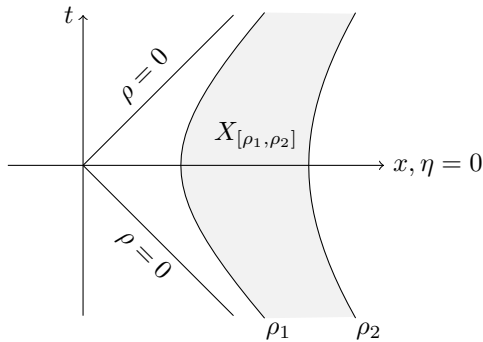
\begin{figure}
\centering
\begin{tikzpicture}[scale=1]
\fill[gray!10] plot[domain=-1.23:1.23] ({1.3*cosh(\x)}, {1.3*sinh(\x)}) -- plot[domain=0.63:-0.63] ({3*cosh(\x)}, {3*sinh(\x)}); 
\draw[->] (0,-2) -- (0,2) node [left] {$t$};
\draw[->] (-1,0) -- (4,0) node [right] {$x, \eta=0$};
\draw plot[domain=-1.23:1.23] ({1.3*cosh(\x)}, {1.3*sinh(\x)});
\draw plot[domain=-0.634:0.634] ({3*cosh(\x)}, {3*sinh(\x)});
\draw (2,-2) --(0,0) node [midway, below, sloped] {$\rho=0$};
\draw (0,0) -- (2,2) node [midway, above, sloped] {$\rho=0$};
\node at (2.3,0.4) {$X_{[\rho_1,\rho_2]}$};
\node at (2.6,-2.2) {$\rho_1$};
\node at (3.8,-2.2) {$\rho_2$};
\end{tikzpicture}
\caption{Spacetime region bounded by two hyperbola of constant Rindler spatial coordinates at $\rho_1$ and $\rho_2$ with $\rho_2>\rho_1$. The boundary of Rindler corresponds to $\rho=0$.}
\label{fig:Rin-2}
\end{figure}

A curve described by a constant value of the Rindler spatial coordinate corresponds to a hyperbola. Consider the region $X_{[\rho_1,\rho_2]}$ bounded by two hyperbolas $\Sigma_\rho$ defined by $\rho=\rho_i$ with $i=1,2$, namely the region $ \R \times [\rho_1,\rho_2]$ (see Figure~\ref{fig:Rin-2}); the evolution parameter $s$ now represents the Rindler spatial coordinate. With the purpose of avoiding exponentially growing solutions in time, $p$ is restricted to take imaginary values. On the other hand, both types of Bessel function enter in the expansion of complexified solutions:\footnote{The absolute value of $p$ that appears in the index of the modified Bessel function of the first kind is due to guarantee the independence of these functions, since the modified Bessel function of the second kind can be expressed as in expression 10.27.4 of \cite{NIST:DLMF},
\be
K_{\nu}(z) = \frac{\pi}{2} \frac{I_{- \nu}(z) - I_{\nu}(z)}{\sin(\nu \pi)},
\label{eq:mBf2k}
\ee
valid for $\nu \neq 0, \pm 1, \pm 2, \dots$ and, for purely imaginary order, $\overline{I_{\im p}(z)} = I_{-\im p}$.
}
\be
\phi(\eta,\rho) = \int_{-\infty}^\infty \xd p \, e^{\im p \eta} \left( \phi^{\text{a}}(p)  I_{\im |p|}(m \rho) + \phi^{\text{b}}(p) K_{\im p} (m \rho) \right) .
\label{eq:phiR}
\ee
Notice that both modified Bessel functions of imaginary order present an oscillatory behavior for small values of their argument, according to (\ref{eq:as-mBf-0}) and an exponential one for large values of their argument, as shown by expression (\ref{eq:as-mBf-inf}). The symplectic potential (\ref{eq:sympot}) on a hyperbola $\Sigma_{\rho}$ (as a boundary of a region of smaller values of $\rho$) is defined in terms of the normal derivative to $\Sigma_{\rho}$, namely the derivative operator $\rho \partial_{\rho}$,
\be
   [\phi,\xi]_{\rho} = - \int_{-\infty}^{\infty} \xd \eta \, \xi(\eta,\rho) (\rho \partial_{\rho} \phi)(\eta,\rho).
   \label{eq:rindhsymp}
\ee
Consequently, the symplectic form (\ref{eq:sympfrompotlin}) on $\Sigma_{\rho}$ is
\begin{align}
\omega_{\rho}(\phi,\xi) &= \frac12 \int_{-\infty}^{\infty} \xd \eta \left( \phi(\eta,\rho) (\rho \partial_{\rho} \xi)(\eta,\rho) - \xi(\eta,\rho) (\rho \partial_{\rho} \phi)(\eta,\rho)  \right),\\
&=  \pi \int_{-\infty}^{\infty} \xd p \left( \xi^{\text{a}}(p)\phi^{\text{b}}(-p) - \xi^{\text{b}}(p)\phi^{\text{a}}(-p)\right).
\label{eq:sfRh}
\end{align}
The inner product (\ref{eq:iplc}) takes the form
\be
(\phi,\xi)_{\rho} = - 4 \pi \im \int_{-\infty}^{\infty} \xd p \left( \overline{\phi^{\text{a}}(p)} \left( \xi^{\text{b}}(p) - \frac{2 \im \sinh(|p| \pi)}{\pi} \xi^{\text{a}}(p) \right) - \overline{\phi^{\text{b}}(p)} \xi^{\text{a}}(p) \right).
\label{eq:ipRh}
\ee
The difficulty to read off from this expression a positive-definite subspace suggests to chose a different parametrization of the field. In particular, the field can be expanded in the basis provided by the modified Bessel function $I_{\im |p|}$ and its complex conjugate,
\be
\phi(\eta,\rho) = \int_{-\infty}^{\infty} \xd p \, e^{\im p \eta} \left( \phi^{\text{a}}(p)  I_{\im |p|} (m \rho) + \phi^{\text{b}}(p)   \overline{I_{\im |p|}(m \rho)} \right).
\ee
The symplectic form (\ref{eq:sfRh}) and the inner product (\ref{eq:ipRh}) in this parametrization result to be equal respectively to,
\begin{align}
\omega_{\rho}(\phi,\xi) &=  2 \im \int_{-\infty}^{\infty} \xd p \, \sinh(|p| \pi) \left( \xi^{\text{a}}(p) \phi^{\text{b}}(-p) - \xi^{\text{b}}(p)\phi^{\text{a}}(-p)\right),\\
(\phi,\xi)_{\rho} &= 8 \int_{-\infty}^{\infty} \xd p \, \sinh(|p| \pi)\left( \overline{\phi^{\text{b}}(p)} \xi^{\text{b}}(p) - \overline{\phi^{\text{a}}(p)} \xi^{\text{a}}(p) \right).
\end{align}
These relations show that the subspace $L^+_{\Sigma_{\rho}}$ defined by the condition $\phi^{\text{a}}(p)=0$ is a positive-definite Lagrangian subspace.

In order to obtain the expression of the symplectic potential with Dirichlet boundary conditions in $X_{[\rho_1,\rho_2]}$ it is convenient to express classical solutions in terms of the boundary field configurations $\varphi_1$ and $\varphi_1$ at $\rho_1$ and $\rho_2$ respectively,
\be
\phi(\eta, \rho) = \left( \frac{\Delta(p,\rho,\rho_2)}{\Delta(p,\rho_1,\rho_2)} \varphi_1\right) (\eta) + \left( 
\frac{\Delta(p,\rho_1,\rho)}{\Delta(p,\rho_1,\rho_2)} \varphi_2\right) (\eta),
\ee
where the quotients are understood as operators acting on a mode decomposition of the field configurations, $\varphi(\eta) = \int \xd p \, \varphi(p) e^{\im p \eta}$ and 
\be
\Delta(p,\rho_1,\rho_2) = I_{\im |p|}(m \rho_1)   \overline{I_{\im |p|}(m \rho_2)} - I_{\im |p|}(m \rho_2)   \overline{I_{\im |p|}(m \rho_1)}.
\ee
The symplectic potentials (\ref{eq:rindhsymp}) on $\rho_1$ and $\rho_2$ with these boundary conditions are then,
\begin{multline}
[\text{pol}^{\rho_1}_{[\rho_1,\rho_2]}(\varphi_1, \varphi_2),\varphi_1]_{\rho_1}\\  = \int_{-\infty}^{\infty} \xd \eta 
\left( \varphi_1(\eta) \left( \frac{\rho_1 \, m \, \sigma (p,\rho_2,\rho_1)}{\Delta (p,\rho_1,\rho_2)} \varphi_1 \right) (\eta) + \varphi_1(\eta) \left( \frac{2 \im \sinh(|p| \pi)}{\pi  \Delta (p,\rho_1,\rho_2)} \varphi_2 \right) (\eta) \right),
\end{multline}
\begin{multline}
[\text{pol}^{\rho_2}_{[\rho_1,\rho_2]}(\varphi_1, \varphi_2),\varphi_2]_{\rho_2}\\ = -   \int_{-\infty}^{\infty} \xd \eta 
\left( \varphi_2(\eta) \left( \frac{\rho_2 \, m \, \sigma (p,\rho_1,\rho_2)}{\Delta (p,\rho_1,\rho_2)} \varphi_2 \right) (\eta) + \varphi_2(\eta) \left( \frac{2 \im \sinh(|p| \pi)}{\pi \Delta (p,\rho_1,\rho_2)} \varphi_1 \right) (\eta) \right),
\label{eq:secondsympotRin}
\end{multline}
where $ \sigma(p,\rho_1,\rho_2) = I_{\im |p|}(m \rho_1) \overline{I'_{\im |p|}(m \rho_2)}- I_{\im |p|}'(m \rho_2) \overline{I_{\im |p|}(m \rho_1)}$.

In order to obtain the vacua on each side of the hyperbola we apply the asymptotic propagator method (Section~\ref{sec:fpvac}) and consider the two corresponding limits: $\rho_1 \to 0$ and $\rho_2 \to \infty$.

\begin{itemize}
	\item Limit $\rho_1 \to 0$.

According to (\ref{eq:as-mBf-0}), both modified Bessel functions present an oscillatory behavior in this limit because of the imaginary index. Notice that,
\begin{align}
 \Delta(p,\rho_1,\rho_2)\big|_{\rho_1 \ll 1} &\sim  \overline{I_{\im |p|}(m \rho_2)} \frac{\left( m \rho_1/2\right)^{\im |p|}}{ \Gamma(1+\im |p|)} - I_{\im |p|}(m \rho_2) \frac{\left( m \rho_1/2\right)^{- \im |p|}}{ \Gamma(1-\im |p|)},\\ 
 \sigma(p,\rho_1,\rho_2)\big|_{\rho_1 \ll 1} &\sim  \overline{I'_{\im |p|}(m \rho_2)} \frac{\left( m \rho_1/2\right)^{\im |p|}}{ \Gamma(1+\im |p|)} - I_{\im |p|}'(m \rho_2) \frac{\left( m \rho_1/2\right)^{- \im |p|}}{ \Gamma(1-\im |p|)}.
\end{align}
A simple Wick rotation $\rho_1\mapsto \im\rho_1$ (or with negative sign) clearly does not lead to a well defined limit when $\rho_1\to 0$. However, inspection suggests to implement a \emph{modified Wick rotation} with respect to $\ln \rho_1$ instead of $\rho_1$, or equivalently to rotate as $\rho_1 \mapsto \rho_1^{-\im}$, and then to consider the limit $\rho_1 \to 0$ of (\ref{eq:secondsympotRin}) with $\varphi_1=0$.

The limit of the field propagator results to be,
\begin{align}
Z_X(\varphi_2) &= \lim_{\rho_1 \to 0} \exp \left( \frac{\im}{2} [\text{pol}^{\rho_2}_{[\rho_1^{-\im},\rho_2]}(0, \varphi_2),\varphi_2]_{\rho_2} \right), \\
&=   \exp \left(  - \frac{\im}{2} \int_{-\infty}^{\infty}  \xd \eta \, \varphi_2(\eta) \left( \frac{\rho_2\, m\, \overline{I'_{\im |p|}(m \rho_2)} }{ \overline{I_{\im |p|}(m \rho_2)} } \varphi_2 \right) (\eta) \right) ,
\end{align}
where $X$ is the region to the left of $\rho_2$. This coincides with the expression of the vacuum wave function obtained in \cite{CoRa:qftrindler}, see formula~(46). At the same time it reproduces the definite Lagrangian subspace $L_{\Sigma_\rho}^+$ considered previously, that consists of the modes $\overline{I_{\im |p|}(m \rho)} e^{\im p\eta}$ only.

\item Limit $\rho_2 \to \infty$.

  According to the asymptotic behavior (\ref{eq:as-mBf-inf}), the modes $K_{\im p}(m\rho)$ show exponential decay for large values of $\rho$. Furthermore, as can be read off from expression (\ref{eq:sfRh}) for the symplectic form, they form a real Lagrangian subspace. Correspondingly, we consider the limit of the field propagator without necessity for a Wick rotation. To this end notice first,
\be 
\lim_{\rho_2 \to \infty} \frac{\sigma(p, \rho_2,\rho_1)}{\Delta(p,\rho_1, \rho_2)} = -
 \frac{ \overline{I'_{\im |p|}(m \rho_1)}  - I_{\im |p|}'(m \rho_1)     }{ \overline{I_{\im |p|}(m \rho_1)}  - I_{\im |p|}(m \rho_1) } = - \frac{K_{\im p}'(m \rho_1)}{K_{\im p}(m \rho_1)}.
\ee
Then, the limit of the field propagator is, compare equation (\ref{eq:fplimsp}),
 \begin{align}
Z_X(\varphi_1) &= \lim_{\rho_2 \to \infty} \exp \left(- \frac{\im}{2} [\text{pol}^{\rho_1}_{[\rho_1, \rho_2]}(\varphi_1,0),\varphi_1]_{\rho_1} \right), \\
&=   \exp \left(  \frac{\im}{2} \int_{-\infty}^{\infty}  \xd \eta \, \varphi_1(\eta) \left( \frac{\rho_1 m K_{\im p}'(m \rho_1)}{K_{\im p}(m \rho_1)} \varphi_1 \right) (\eta) \right),
 \end{align}
 where $X=\R \times [\rho_1,+\infty)$ is the region to the right of $\rho_1$. (Note the opposite orientation of the hypersurface at $\rho_1$.) As expected, we obtain the subspace of the modes $K_{\im p}(m\rho) e^{\im p\eta}$.
\end{itemize}

  The vacua on the two sides of the hyperbola are completely different. On the left hand side we have a traditional vacuum corresponding to a positive-definite Lagrangian subspace. On the right hand side we have a traditional amplitude corresponding to a real Lagrangian subspace. The presentation of the Feynman propagator that exhibits these two vacua is given by expression~(89) of \cite{CoRa:qftrindler},
\begin{multline}
G_F((\eta,\rho),(\eta',\rho'))\\ = \int_{-\infty}^{\infty} \frac{\xd p}{2 \pi} \left( \theta(\rho-\rho') K_{\im p} (m \rho) \overline{I_{\im |p|}(m \rho')} + \theta(\rho'-\rho) K_{\im p} (m \rho') \overline{I_{\im |p|}(m \rho)} \right) e^{\im p (\eta-\eta')}.
\end{multline}
This expression for the Feynman propagator is completely equivalent to (\ref{eq:Fp-R-eta}), see Section~VI.C. of \cite{CoRa:qftrindler}. It also reproduces the expression given by Boulware \cite[(3.22)]{Bou:qftschwrind} when the correct open boundary conditions for Rindler space are chosen.

\subsection{Euclidean space}
\label{sec:2deucl}

We review the free quantum theory of a massive scalar field in 2-dimensional Euclidean space \cite{CoOe:2deucl}. We use cartesian coordinates $(\tau,x)$, where $\tau$ is taken to play the role of time when convenient, although the theory is fully invariant under rotations. The analog of the Klein-Gordon equation in this context is the Helmholtz equation $\left( \partial_{\tau}^2 + \partial_x^2+m^2\right)\phi(\tau,x)=0$ with corresponding action,
\be
S(\phi)=\frac{1}{2}\int \xd \tau \, \xd x\, \left((\partial_{\tau}
\phi)^2+ (\partial_x \phi)^2 - m^2\phi^2\right) .
\ee
As in \cite{CoOe:2deucl} we consider two types of region and corresponding vacua: The (Euclidean) time-interval region, analogous to the Minkowski case, and the disk and annulus regions bounded by one and two circles respectively.

\subsubsection{Hyperplane}

To determine the vacuum state on an equal-time hyperplane (i.e., a straight line), say $\Sigma_{0}$ at $\tau=0$, we consider the region $X$ in the future to this hypersurface and foliate it in terms of constant-$\tau$ hypersurfaces. Then, the time-interval region $X_s=[0,s] \times \R$ has boundary components $\Sigma_0$ and $\Sigma_s$ (at $\tau=s$). Complexified propagating solutions in this region can be expanded as
\be
\phi(\tau,x) = \int_{-m}^{m} \frac{\xd \nu}{2 \pi} \left( \phi^{\text{a}}(\nu) e^{-\im \omega_{\nu} \tau} + \phi^{\text{b}}(\nu) e^{\im \omega_{\nu} \tau}\right) e^{\im \nu x},
\ee
where $\omega_{\nu} = \sqrt{m^2-\nu^2}$. Evanescent solutions, which are the only ones in the massless case, read
\be
\phi(\tau,x) = \int_{|\nu| > m} \frac{\xd \nu}{2 \pi} \left( \phi^\text{x}(\nu) e^{- \omega_{\nu} \tau} + \phi^\text{i}(\nu) e^{ \omega_{\nu} \tau}\right) e^{\im \nu x},
\ee
where $\omega_{\nu} = \sqrt{\nu^2-m^2}$. 
The symplectic potential (\ref{eq:sympot}) on $\Sigma_0$ (oriented as a boundary component of $X_s$) is,
\be
[\phi,\xi] = -\int \xd x \, \xi(\tau,x) (\partial_{\tau} \phi)(\tau,x) .
\ee
The symplectic form (\ref{eq:sympfrompotlin}) is then,
\be
\omega(\phi, \xi) = -\frac12 \int \xd x \left( \xi(\tau,x) (\partial_{\tau} \phi)(\tau,x) - \phi(\tau,x) (\partial_{\tau} \xi)(\tau,x) \right).
\label{eq:2desymf}
\ee
It is easy to verify that $\omega$ vanishes between propagating and evanescent solutions, as for the field in Minkowski space, and the solution space decomposes as $L_{\Sigma_0}^{\bC} = L_{\Sigma_0}^{\text{p},\bC} \oplus L_{\Sigma_0}^{\text{e},\bC}$. On $L_{\Sigma_0}^{\text{p},\bC}$ the symplectic form (\ref{eq:2desymf}) is
\be
\omega^{\text{p}}(\phi, \xi) =  \im \int_{-m}^{m} \frac{\xd \nu}{2\pi}\, \omega_{\nu}\left(\phi^{\text{a}}(\nu) \xi^{\text{b}}(-\nu) - \phi^{\text{b}}(\nu) \xi^{\text{a}}(-\nu) \right),
\ee 
and the inner product (\ref{eq:iplc}),
\be
(\phi, \xi)^{\text{p}} = 4 \int_{-m}^m \frac{\xd \nu}{2\pi} \, \omega_{\nu}\left(\overline{\phi^{\text{a}}(\nu)} \xi^{\text{a}}(\nu) - \overline{\phi^{\text{b}}(\nu)} \xi^{\text{b}}(\nu) \right).
\ee
This shows that the solutions determined by $\phi^{\text{b}}(\nu)=0$ form a positive-definite Lagrangian subspace $L_{\Sigma_0}^{\text{p},+}\subseteq L_{\Sigma_0}^{\text{p},\bC}$. At the same time for this subspace the Wick rotated derivative operator $-\im\partial_{\tau}$ has negative eigenvalues.

On evanescent solutions the symplectic form (\ref{eq:2desymf}) is,
\be
\omega^{\text{e}}(\phi, \xi) =  \int_{|\nu| > m} \frac{\xd \nu}{2\pi}\, \omega_{\nu}\left(\phi^{\text{x}}(\nu) \xi^{\text{i}}(-\nu) - \phi^{\text{i}}(\nu) \xi^{\text{x}}(-\nu) \right) .
\ee
A vanishing boundary condition in the asymptotic future (i.e.\ in the limit $\tau \to \infty$) is satisfied by solutions such that $\phi^{\text{i}}(\nu)=0$. The subspace $L_{\Sigma_0}^{\text{e},+}\subseteq L_{\Sigma_0}^{\text{e},\bC}$ of these solutions is a Lagrangian subspace and corresponds to the negative eigenvalues of the derivative operator $\partial_{\tau}$.

The Feynman propagator, given by formula (107) of \cite{CoOe:2deucl}, receives contribution from both the propagating and evanescent solutions,
\begin{align}
G_F((\tau,x), (\tau',x')) &= \im \int_{-m}^{m} \frac{\xd \nu}{2 \pi}  \frac{1}{2 \omega_{\nu}}\left( \theta(\tau-\tau') e^{-\im \omega_{\nu} \tau}e^{\im \omega_{\nu} \tau'} + \theta(\tau'-\tau) e^{-\im \omega_{\nu} \tau'}e^{\im \omega_{\nu} \tau}\right) e^{\im \nu (x-x')} \nonumber\\
& -\int_{|\nu| > m} \frac{\xd \nu}{2 \pi}  \frac{1}{2 \omega_{\nu}}\left( \theta(\tau-\tau') e^{- \omega_{\nu} \tau}e^{ \omega_{\nu} \tau'} + \theta(\tau'-\tau) e^{- \omega_{\nu} \tau'}e^{ \omega_{\nu} \tau}\right) e^{\im \nu (x-x')}.
\label{eq:2defeyn1}
\end{align}
This expression shows that the boundary conditions for the propagating and evanescent solutions are given precisely by the subspaces $L_{\Sigma_0}^{\text{p},+}$ and $L_{\Sigma_0}^{\text{e},+}$.

In order to apply the asymptotic method of Section~\ref{sec:fpvac} we note for the propagating solutions the complete analogy to the Minkowski case. Concretely, we perform a Wick rotation $s\mapsto -\im s$,
\begin{align}
  [\text{pol}^0_{[0,s]}(\varphi,0),\varphi]_0 &=  \int  \xd x\, \varphi(x) \left( \frac{ \omega_{\nu}  \cos ( \omega_{\nu}  s)}{\sin ( \omega_{\nu}  s)} \varphi \right) (x), \\
    [\widetilde{\pol}_X(\varphi),\varphi]_0 & =\lim_{s\to\infty} [\pol_{[0,-\im s]}^0(\varphi,0),\varphi]_{0}=\im \int \xd x\, \varphi(x) \left(\omega_{\nu}\, \varphi\right)(x) .
\end{align}
Compare equations~(\ref{eq:mslhprop}) and (\ref{eq:mslhlim}). Inserting the latter expression into the field propagator (\ref{eq:univprop}) and going to momentum space recovers precisely the propagating vacuum wave function~(31) of \cite{CoOe:2deucl}.
For the evanescent solutions on the other hand there is no Wick rotation involved,
\begin{align}
  [\text{pol}^0_{[0,s]}(\varphi,0),\varphi]_0 &=  \int  \xd x\, \varphi(x) \left( \frac{ \omega_{\nu}  \cosh ( \omega_{\nu}  s)}{\sinh ( \omega_{\nu}  s)} \varphi \right) (x), \\
    [\widetilde{\pol}_X(\varphi),\varphi]_0 & =\lim_{s\to\infty} [\pol_{[0,s]}^0(\varphi,0),\varphi]_{0}=\int \xd x\, \varphi(x) \left(\omega_{\nu}\, \varphi\right)(x) .
\end{align}
This is (up to a sign) completely analogous to the evanescent solutions on the timelike hyperplane in Minkowski space, compare Section~\ref{sec:tlhp}.

\subsubsection{Circle}

\begin{figure}
\centering
\begin{tikzpicture}[scale=1]
\draw[fill=gray!10] (0,0) circle (2cm);
\draw[fill=white] (0,0) circle (1cm);
\draw[->] (0,-2.5) -- (0,2.5) node [above] {$\tau$};
\draw[->] (-3,0) -- (3,0) node [right] {$x$};
\draw[->] (1.6,1.6) -- (2,2);
\draw[->] (-1.6,1.6) -- (-2,2);
\draw[->] (1.6,-1.6) -- (2,-2);
\draw[->] (-1.6,-1.6) -- (-2,-2);
\node at (-1.2,0.8) {$X_{[R,\hat{R}]}$};
\node at (2.2,0.3) {$\hat{R}$};
\node at (1.2,0.3) {$R$};
\end{tikzpicture}
\caption{Annulus region bounded by two circles of radii $R$ and $\hat{R}$ in Euclidean space.}
\label{fig:Es}
\end{figure}
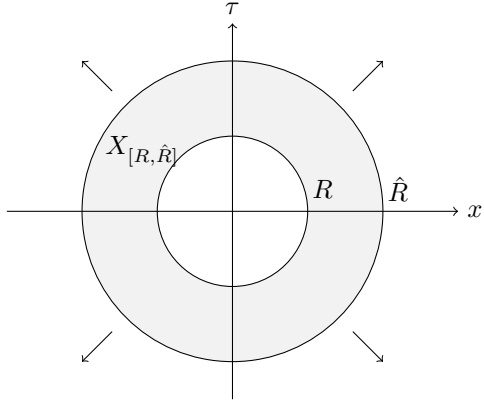

In order to obtain the vacuum state on the circle it is convenient to use polar coordinates $(r,\theta)$, with $\tau= r \sin \theta$ and $x = r \cos \theta$. The Helmholtz equation, $\left( \partial_r^2 +\frac{1}{r} \partial_r + \frac{1}{r^2} \partial_{\theta}^2 +m^2 \right)\phi=0$, is solved in terms of Bessel functions,
\be
\phi(r, \theta) = \sum_{n=-\infty}^{\infty} \left( \phi^{\text{a}}_n H_n(mr) + \phi^{\text{b}}_n \overline{H_n(mr)}\right) e^{\im n \theta},
\ee
where $H_n(z)$ is the Hankel function, related the Bessel functions of the first and second kind, $J_{n}(z)$ and $Y_{n}(z)$ respectively, as $H_n(z)=J_n(z) + \im Y_n(z)$.
The symplectic potential (\ref{eq:sympot}) on a circle of radius $R$ (as a boundary of the region exterior to the disk) has the form,
\be
[\phi,\xi]_R = -\int_0^{2\pi} \xd \theta \, \xi(r, \theta) (r \partial_r \phi)(r,\theta) .
\ee
The symplectic form (\ref{eq:sympfrompotlin}) and inner product (\ref{eq:iplc}) are,
\begin{align}
\omega_R(\phi, \xi) &=-\frac12 \int_0^{2 \pi} \xd \theta \left( \xi(r,\theta)(r \partial_r \phi) (r,\theta) - \phi(r,\theta)(r \partial_r \xi) (r,\theta) \right),\\
&= 4 \im \sum_{n=-\infty}^{\infty} (-1)^n \left( \xi_n^{\text{a}} \phi_{-n}^{\text{b}} - \xi_n^{\text{b}} \phi_{-n}^{\text{a}} \right),\\
(\phi,\xi)_R & = 16 \sum_{n=-\infty}^{\infty} \left( \xi_n^{\text{b}} \overline{\phi_{n}^{\text{b}}} - \xi_n^{\text{a}} \overline{\phi_{n}^{\text{a}}} \right).
\end{align}
It is clear from these expressions that the subspace $L_{R}^+$ defined by $\phi^{\text{a}}_n=0$ is a positive-definite Lagrangian subspace of the space of complexified solutions $L_R^{\bC}$ in a neighborhood of the circle. On the other hand, the solutions well defined in the interior of the disk $D_R$ of radius $R$ are given in terms of the Bessel functions of the first kind $J_n(z)=\frac12(H_n(z)+\overline{H_n(z)})$. These form a real Lagrangian subspace $L_{D_R}\subseteq L_R$ which determines the amplitude, i.e., generalized vacuum, in $D_R$. The corresponding field propagator (\ref{eq:samplbdy}) can thus be written as,
\begin{equation}
  Z_{D_R}(\varphi)=\exp\left(-\frac{\im}{2}[\pol_{D_R}(\varphi,\varphi)]_{R}\right)=\exp \left(\frac{\im}{2} \int_0^{2\pi} \xd \theta \, \varphi(\theta) \left( \frac{R\, m J_n'(m R)}{J_n(m R)} \varphi \right) (\theta) \right) .
\end{equation}

In order to obtain the vacuum on the exterior of the circle we apply the asymptotic propagator method (Section~\ref{sec:fpvac}). To this end consider the annulus region bounded by two circles with radii $\hat{R}>R$, i.e., $X_{[R,\hat{R}]}=[R,\hat{R}] \times 2 \pi$, see Figure~\ref{fig:Es}, where both Bessel functions are well defined. Denoting with $\varphi$ and $\hat{\varphi}$ the field configurations on the circles $r=R$ and $r=\hat{R}$ respectively, the solution in $X_{[R,\hat{R}]}$ can be written as,
\be
\phi(r,\theta) = \left( \frac{\Delta(n,mr,m\hat{R})}{\Delta(n,mR,m\hat{R})} \varphi \right)(\theta)  + \left(\frac{\Delta(n,mR,mr)}{\Delta(n,mR,m\hat{R})} \hat{\varphi}\right)(\theta),
\ee
where $\Delta(n,z,\hat{z}) = H_n(z) \overline{H_n(\hat{z})} - \overline{H_n(z)} H_n(\hat{z})$.
The relevant field propagator is written in terms of the symplectic potential as,
\begin{equation}
[\text{pol}^{R}_{[R,\hat{R}]}(\varphi, \hat{\varphi}),\varphi]_{R} = \int_0^{2\pi} \xd \theta 
\left( \varphi(\theta) \left( \frac{R \, m \, \sigma (n,m\hat{R},mR)}{\Delta(n,mR,m\hat{R})} \varphi \right) (\theta) + \varphi(\theta) \left( \frac{4\im}{\pi \Delta(n,mR,m\hat{R})} \hat{\varphi} \right) (\theta) \right), 
\end{equation}
where $ \sigma(n,z,\hat{z}) = H_{n}( z) \overline{H_{n}'(\hat{z})}- H_n'(\hat{z}) \overline{H_n(z)}$, the prime denoting derivative w.r.t.\ the argument. 

Now the external radius $\hat{R}$ plays the role of the parameter $s$. The asymptotic behavior of the Hankel functions, for large values of their argument $z$, can be derived from expression~10.7.8 of \cite{NIST:DLMF},
	\be
	H_n(z) \sim \sqrt{\frac{2}{\pi z}} e^{\im (z-n \pi/2-\pi/4)}, \quad \overline{H_n(z)} \sim \sqrt{\frac{2}{\pi z}} e^{-\im (z-n \pi/2-\pi/4)}.
 	\label{eq:asHH}
 	\ee
The oscillatory behavior of these functions suggests to implement a Wick rotation in order to perform the limit. So we consider the rotation $\hat{R}  \mapsto -\im \hat{R}$ and then the limit $\hat{R} \to \infty$, as in expression~(\ref{eq:fplimisp}). This yields,
 \begin{align}
Z_X(\varphi) &= \lim_{\hat{R} \to \infty} \exp \left(\frac{\im}{2} [\text{pol}^R_{[R, \im \hat{R}]}(\varphi, 0),\varphi]_R \right),\\
&=  \lim_{\hat{R} \to \infty} \exp \left(  \frac{\im}{2} \int_0^{2\pi} \xd \theta \, \varphi(\theta)
 \left( R\, m  \frac{H_{n}(- m \im \hat{R}) \, \overline{H_{n}'(m R)}  - H_{n}'(m R) \, \overline{H_{n}(- m \im \hat{R})}  }{H_{n}(m R) \, \overline{H_{n}(- m \im \hat{R})} - \overline{H_{n}(m R)} \, H_{n}(- m \im \hat{R})}
  \varphi \right) (\theta) \right),\\
  &= \exp \left(- \frac{\im}{2} \int_0^{2\pi} \xd \theta \,  \varphi(\theta) \left( R\, m \frac{\overline{H_{n}'(m R)}}{\overline{H_{n}(mR)}}\right)
 \varphi(\theta) \right).
\end{align}
This expression coincides precisely with the vacuum wave function in the Schrödinger representation obtained in \cite{CoOe:2deucl}, see expression (36) there. What is more, we can read off that the selected Lagrangian subspace is the positive-definite subspace $L_R^+$ of the $\overline{H_n(m r)} e^{\im n \theta}$ modes remarked previously.

The different vacua obtained on the two sides of the circle define the expression of the Feynman propagator, computed in \cite{CoOe:2deucl}, see expression (130), (with a minor reformulation),
\be
G_F((r,\theta),(r',\theta')) = \frac{\im}{4} \sum_{n=-\infty}^{\infty} e^{\im n (\theta - \theta')} \left( \theta(r-r') J_{n}(mr') \overline{H_n(mr)} + \theta(r'-r) J_{n}(mr) \overline{H_n(mr')} \right).
\ee
As shown there this is equivalent to expression~(\ref{eq:2defeyn1}).

\subsection{de~Sitter space}
\label{sec:deSitter}

\subsubsection{Standard vacuum}

We use the coordinate system in which de~Sitter metric is given by
\be
\xd s^2 = \frac{R^2}{t^2} \left( \xd t^2 - (\xd x^1)^2 - (\xd x^2)^2- (\xd x^3)^2 \right)= \frac{R^2}{t^2} \left( \xd t^2 - (\xd \underline{x})^2  \right),
\label{eq:dSmetric}
\ee
with $t \in (0, \infty)$ and $\underline{x} \in \mathbb{R}^3$. The coordinates chosen describe only half of spacetime; the other half is recovered by extending $t$ to negative values ($t=0$ represents a coordinate singularity).  
Complexified solutions of the massive Klein-Gordon equation describing a minimally coupled scalar field are
\be
\phi(t, \underline{x}) =\int \frac{\xd ^3 \underline{k}}{(2 \pi)^{3/2}} \left( \phi^{\text{a}}(\underline{k}) t^{3/2} H_{\nu}(kt) \, e^{\im \underline{k} \cdot \underline{x}} +  \phi^{\text{b}}(\underline{k}) t^{3/2} \overline{ H_{\nu}(kt)} \, e^{-\im \underline{k} \cdot \underline{x}} \right),
\label{eq:sol-dS}
\ee
where $k=|\underline{k}|$,\footnote{In order for the solution to be well defined (i.e.\ not divergent) in the whole spacetime the components of the 3-vector $\underline{k}$ have to be real. Consequently the modulus $k$ is necessarily positive, $k>0$.} and $H_{\nu}(z)$ is the Hankel function (introduced in the previous section) with index $\nu = \sqrt{\frac{9}{4} - (m R)^2}$. We will assume that $\nu$ is real and consequently $\nu < 3/2$. Then, the modes $t^{3/2} H_{\nu}(kt)$ as well as their complex conjugates vanish in the limit $t \to 0$, due to the asymptotic behavior of the Hankel functions, see expression~10.7.7 of \cite{NIST:DLMF},
\be
	H_{\nu}(z) \sim - \overline{H_{\nu}(z)} \sim - \frac{\im }{\pi} \Gamma({\nu}) \left( \frac{z}{2}\right)^{-{\nu}}
	\label{eq:asymHankel}
\ee
valid for positive ${\nu}$. On the other hand, for large values of $t$ these modes oscillate according to (\ref{eq:asHH}).

We consider the spacetime region $X_t$ to the future of the hypersurface $\Sigma_t$ of constant de~Sitter time. The symplectic potential (\ref{eq:sympot}) on $\Sigma_t$ oriented as a boundary of $X_t$ is,
\be
[\phi,\xi]_t = - \int \xd^3 \underline{x} \frac{R^2}{t^2} \xi(t, \underline{x}) (\partial_t \phi) (t, \underline{x}) .
\ee
The symplectic form (\ref{eq:sympfrompotlin}) and inner product (\ref{eq:iplc}) are,
\begin{align}
\omega_t(\phi, \xi) &= \frac12 \int \xd^3 \underline{x} \frac{R^2}{t^2} \left( \phi(t, \underline{x}) (\partial_t \xi) (t, \underline{x}) - \xi(t, \underline{x}) (\partial_t \phi) (t, \underline{x}) \right), \label{eq:dssympf} \\
&= \frac{2 \im}{\pi} \int \xd^3 \underline{k} \, R^2 \left( \phi^{\text{b}}(\underline{k}) \xi^{\text{a}}(\underline{k}) - \phi^{\text{a}}(\underline{k}) \xi^{\text{b}}(\underline{k}) \right), \\
(\phi,\xi)_t & = \frac{8}{\pi} \int \xd^3 \underline{k} \, R^2 \left( \xi^{\text{b}}(\underline{k}) \overline{\phi^{\text{b}}(\underline{k})} - \xi^{\text{a}}(\underline{k}) \overline{\phi^{\text{a}}(\underline{k})} \right) .
\end{align}
We note that  the subspace $L_{\Sigma_t}^+$ defined by $\phi^{\text{a}}(\underline{k})=0$ is a positive-definite Lagrangian subspace of the space of solution associated with $\Sigma_t$.

In order to apply the asymptotic method (Section~\ref{sec:fpvac}) for vacuum selection, we foliate spacetime with equal-time hypersurfaces and consider the region between two of the surfaces, namely $X_{[s,s']}=[s,s'] \times \R^3$. It is convenient to express solutions (\ref{eq:sol-dS}) in terms of the boundary configurations $\varphi$ and $\varphi'$ taken by the field at $t=s$ and $t=s'$ respectively,
\be
\phi(t, \underline{x}) = \left( \frac{\Delta(t,s')}{\Delta(s,s')} \varphi \right)(\underline{x}) + \left( \frac{\Delta(s,t)}{\Delta(s,s')} \varphi' \right)(\underline{x}),
\ee
where $\Delta(t,s) = t^{3/2} H_{\nu}(kt) s^{3/2} \overline{H_{\nu}(ks)} - t^{3/2} \overline{ H_{\nu}(kt)} s^{3/2} H_{\nu}(ks)$.

In the region $X_{[s,s']}$ the symplectic potential that appears in the expression of the field propagator and determines the vacuum in an appropriate limit is,
\begin{multline}
[\text{pol}^s_{[s,s']}(\varphi, \varphi'),\varphi]_s 
= - \int \xd^3 \underline{x} \left( \varphi(\underline{x})
 \left( \frac{R^2}{s^2} \left( \frac{3}{2s} + k \frac{H_{\nu}'(k s) \, \overline{H_{\nu}(k s')} - \overline{H_{\nu}'(k s)} \, H_{\nu}(k s')}{ H_{\nu}(k s) \, \overline{H_{\nu}(k s')} - \overline{H_{\nu}(k s)} \, H_{\nu}(k s') }\right)
  \varphi \right) (\underline{x}) \right.\\
  \left.
  + \varphi(\underline{x})   \left( \frac{-4 \im R^2 }{\pi (ss')^{3/2} \left(    H_{\nu}(k s) \, \overline{H_{\nu}(k s')} - \overline{H_{\nu}(k s)} \, H_{\nu}(k s') \right)}  \varphi' \right) (\underline{x}) \right), 
 \end{multline}
where a prime over the Bessel function indicates the derivative with respect to the argument.

Because of the oscillatory behavior (\ref{eq:asHH}) of the Hankel functions for large argument we expect to obtain the vacuum via the limit (\ref{eq:fplimisp}), implementing the Wick rotation $s' \to -\im s'$,
 \begin{align}
   & Z_{X_s}(\varphi) \nonumber \\
   &= \lim_{s' \to \infty} \exp \left(\frac{\im}{2} [\text{pol}^s_{[s, -\im s']}(\varphi, 0),\varphi]_s \right),\\
&=  \lim_{s' \to \infty} \exp \left(-  \frac{\im}{2} \int \xd^3 \underline{x} \, \varphi(\underline{x})
 \left( \frac{R^2}{s^2} \left( \frac{3}{2s} + k \frac{H_{\nu}'(k s) \, \overline{H_{\nu}( -\im k s')} - \overline{H_{\nu}'(k s)} \, H_{\nu}( -\im ks')}{ H_{\nu}(k s) \, \overline{H_{\nu}( -\im k s')} - \overline{H_{\nu}(k s)} \, H_{\nu}( -\im k s') }\right)
  \varphi \right) (\underline{x}) \right),\\
  &= \exp \left(-  \frac{\im}{2}  \int \xd^3 \underline{x} \,  \varphi(\underline{x}) \frac{R^2}{s^2} \left( \frac{3}{2s} + k \frac{ \overline{H_{\nu}'(k s)} }{ \overline{H_{\nu}(k s)} }\right)
 \varphi(\underline{x}) \right).
\end{align}
This limit recovers the vacuum wave function in the Schrödinger representation (compare with expression~(56) of \cite{Col:desitterpaper} and apply complex conjugation). At the same time it corresponds precisely to selecting the positive-definite Lagrangian subspace $L_{\Sigma_s}^+$ of modes of the form $t^{3/2}\overline{H_{\nu}(k t)}e^{-\im \underline{k}\cdot \underline{x}}$. These are precisely the modes that define the Bunch-Davies vacuum \cite{Muk:introqeffg}.

As for the previous examples, the Feynman propagator incorporates the vacuum as,
\begin{multline}
G_F((t,\underline{x}), (t',\underline{x}')) \\ = \frac{\im \pi}{4 R^2} \int \frac{\xd^3 \underline{k}}{(2 \pi)^3 } (tt')^{3/2} \left( \theta(t-t') \overline{H_{\nu}(kt)} H_{\nu}(kt') + \theta(t'-t) \overline{H_{\nu}(kt')} H_{\nu}(kt) \right) e^{\im \underline{k} (\underline{x} - \underline{x}')}.
\end{multline}
It has been shown in \cite{Col:desitterpaper} that this expression is equivalent to the one obtained in \cite{Schomblond:1976xc}. We can also read off that, as in Minkowski space, the past vacuum is given by complex conjugate modes of the future vacuum. Choosing a coordinate system covering negative time values the derivation would be analogous with our methods.

\subsubsection{$\alpha$-vacua}

The vacuum state obtained in the preceding section corresponds to the Lagrangian subspace of the space of solutions of the equations of motion determined by the modes $u_{\underline{k}}(t,\underline{x}) = t^{3/2} \overline{H_{\nu}( kt)} e^{-\im \underline{k} \cdot \underline{x}}$. These are a special case of the more general solutions, parametrized by a real number $\alpha$, $u_{\underline{k}}^{\alpha} (t, \underline{x})= t^{3/2} \left( e^{\alpha} J_{\nu}(kt) - \im e^{-\alpha}Y_{\nu}(kt) \right) e^{-\im \underline{k} \cdot \underline{x}}$. In particular $u_{\underline{k}}(t,\underline{x})= u_{\underline{k}}^0(t,\underline{x})$. The field expanded in this new basis reads,
\be
\phi(t, \underline{x}) = \int \frac{\xd^3 \underline{k}}{(2 \pi)^{3/2}} \left( \widetilde{\phi^{\text{a}}} (\underline{k}) u_{\underline{k}}^{\alpha} (t, \underline{x}) +  \widetilde{\phi^{\text{b}}} (\underline{k}) \overline{u_{\underline{k}}^{\alpha} (t, \underline{x})} \right).
\label{eq:sol-dS-alpha}
\ee
The symplectic form (\ref{eq:dssympf}) and inner product (\ref{eq:iplc}) are then,
\begin{align}
\omega_t(\phi,\xi) & = \frac{2 \im}{\pi} \int \xd^3 \underline{k} \, R^2 \left( \widetilde{\phi^{\text{a}}}(\underline{k}) \widetilde{\xi^{\text{b}}}(\underline{k}) - \widetilde{\phi^{\text{b}}}(\underline{k}) \widetilde{\xi^{\text{a}}}(\underline{k})\right), \\
(\phi, \xi)_t & = \frac{8}{\pi} \int \xd^3 \underline{k} \, R^2 \left( \overline{\widetilde{\phi^{\text{a}}}(\underline{k})} \widetilde{\xi^{\text{a}}}(\underline{k}) - \overline{\widetilde{\phi^{\text{b}}}(\underline{k})} \widetilde{\xi^{\text{b}}}(\underline{k}) \right) .
\end{align}
In particular, the modes $u_{\underline{k}}^{\alpha}(t, \underline{x})$ form a positive-definite Lagrangian subspace $L_{\Sigma_t}^+$ with $\widetilde{\phi^{\text{b}}}(\underline{k})=0$.

The relation between the expansions (\ref{eq:sol-dS}) and (\ref{eq:sol-dS-alpha}) amounts to a Bogoliubov transformation of the coefficients:
\begin{align}
\phi^{\text{a}}(\underline{k}) &= \sinh \alpha \, \widetilde{\phi^{\text{a}}}(-\underline{k}) + \cosh \alpha \, \widetilde{\phi^{\text{b}}}(\underline{k}),\\
\phi^{\text{b}}(\underline{k}) &= \sinh \alpha \, \widetilde{\phi^{\text{b}}}(-\underline{k}) + \cosh \alpha \, \widetilde{\phi^{\text{a}}}(\underline{k}). 
\end{align} 
The modes $u_{\underline{k}}^{\alpha}$ define a one-parameter family of vacua known as $\alpha$-vacua \cite{Chernikov:1968zm,Mottola:1984ar,Allen:1985ux}. Notice that requiring that these modes form a Lagrangian subspace implies selecting the same value of the parameter $\alpha$. Indeed,
\be
\omega_t \left( u_{\underline{k}}^{\alpha}, u_{\underline{k}'}^{\alpha'} \right) = \frac{(2 \pi)^3}{\pi} R^2 \im \sinh \left(\alpha'-\alpha \right) \delta(\underline{k} + \underline{k}'),
\ee
which reduces to zero only for $\alpha = \alpha'$.

\section{Discussion and outlook}
\label{sec:outlook}

In the present section we provide some context for our results as well as an outlook on future development. We start with some perspective on the notion of Wick rotation as invoked in Section~\ref{sec:vchoice}.

\paragraph{Wick rotation}
The term \emph{Wick rotation} refers originally to the rotation of a contour of integration in the complex plane in the integral representation of wave functions or propagators \cite{Wic:bethesalpeter}. In the guise of a rotation of the time coordinate this leads to the notion of \emph{Euclidean propagator} \cite{Sch:euclidrel,Nak:qfteuclid}. We recall this for the case of Klein-Gordon theory in Minkowski space. Thus, replace in the Feynman propagator (\ref{eq:fptint}) $t$ by $-\im\tau$ and $t'$ by $-\im\tau'$,
\begin{multline}
  G_F^{(E)}((\tau,x),(\tau',x'))=\im\int\frac{\xd^3 k}{(2\pi)^3 2E}
  \left(\theta(\tau-\tau') e^{-E\tau + \im k x} e^{E \tau'- \im k x'}\right. \\
  \left. +\theta(\tau'-\tau) e^{E \tau - \im k x} e^{- E \tau' +\im k x'}\right) .
  \label{eq:euclprop}
\end{multline}
Note that we think of the multiplication by $-\im$ as the limit of a multiplication by $e^{-\im \theta}$ where the angle $\theta$ is moved continuously from $0$ to $\pi/2$. That is, we are really doing an analytic continuation. (This also determines the behavior of the $\theta$-functions under this rotation.) The Euclidean propagator is a solution to a Wick rotated version of the equations of motion (in this case the Klein-Gordon equation). These in turn are natural equations of motion associated to a Wick rotated metric, which is just the metric of Euclidean space. In Section~\ref{sec:2deucl} we have considered precisely such a theory (in two dimensions), except for the fact that that theory corresponds more precisely to a Wick rotation of the spatial coordinates. Nevertheless, one may appreciate the coincidence in the massless case of the corresponding Feynman propagator (\ref{eq:2defeyn1}) with expression (\ref{eq:euclprop}), up to a factor of $\im$.
Wick rotation leads to a whole \emph{Euclidean formulation} of quantum field theory that can be brought into correspondence with Minkowski quantum field theory in a precise way \cite{OsSc:axeucl}. An advantage of the Euclidean formulation is that its ingredients, including the Euclidean path integral, are generally better behaved mathematically. Most relevant in this respect for our present considerations is the asymptotic behavior of the Euclidean propagator. Indeed, fixing one point, say $(\tau',x')$, we can read off an exponential decay to the (Euclidean) future $e^{-E\tau}$ (where $\tau>\tau'$) and to the past $e^{E\tau}$ (where $\tau<\tau'$).\footnote{Note that the Euclidean propagator shows this same decay behavior in all Euclidean spacetime directions as it is in fact invariant under Euclidean rotations, but this is not manifest in the representation (\ref{eq:euclprop}).} It is precisely this behavior of asymptotic decay that we put at the center of vacuum selection in Section~\ref{sec:vchoice}. However, instead of Wick rotating the actual object of the theory (propagator, solution, etc.)\ we Wick rotate either the normal derivative operator that would detect this decay (Section~\ref{sec:infvac}) or the variable that parametrizes the approach to the asymptotic boundary (Section~\ref{sec:fpvac}). In the first case and when restricted to a standard context of time evolution this recovers and provides a new perspective on the usual recipe for quantization in curved spacetime (Section~\ref{sec:revquant}). Wick rotation in the Euclidean formulation is strictly limited to Minkowski space due to its global nature. In contrast, our notion of Wick rotation is local (to hypersurfaces or regions). This gives our methods (Section~\ref{sec:vchoice}) a much wider applicability (e.g., Sections~\ref{sec:hypcyl} and \ref{sec:examples}) while retaining some of the ideas and motivations of the Euclidean formulation.

\paragraph{Types of Lagrangian subspaces}
We have proposed a unification of the traditional notions of amplitude and vacuum. The mathematical structure at the center of this unification is that of a Lagrangian subspace of the complexified space of germs of solution on the relevant spacetime hypersurface. While in the traditional case of a vacuum on a spacelike hypersurface this is a \emph{definite Lagrangian subspace}, for an amplitude this is a complexified \emph{real Lagrangian subspace} (compare Section~\ref{sec:modecs} or Appendix~\ref{sec:mathlag}). In the cases that we have considered that fall outside the traditional framework (Sections~\ref{sec:hypcyl} and \ref{sec:examples}), notably when timelike hypersurfaces are concerned, we have encountered exactly three types of Lagrangian subspaces: definite, real or mixed. While finding the second (real) type when dealing with a non-compact region that physically should induce a vacuum is already intriguing, the occurrence of the mixed type is even more interesting. By mixed we refer to the following situation. Consider a symplectic vector space $L$, $L^\bC$ its complexification, and $H\subseteq L^\bC$ a Lagrangian subspace. We say $H$ is (properly) \emph{mixed} if there is a (non-trivial) decomposition of $L$ into orthogonal symplectic subspaces $L=L_1\oplus L_2$ and $H=H_1\oplus H_2$ such that $H_1\subseteq L_1^\bC$ is a complexified real Lagrangian subspace and $H_2\subseteq L_2^\bC$ a definite Lagrangian subspace. We have seen the proper mixed type in various examples (Sections~\ref{sec:hcmassive}, \ref{sec:tlhp}, \ref{sec:2deucl}), generally in accordance with a distinction between \emph{evanescent} ($L_1$) and \emph{propagating} waves ($L_2$).
The generalized notion of vacuum we propose (Section~\ref{sec:vaclag}) does not require a limitation to these observed types of Lagrangian subspaces. On the other hand, the methods we propose for vacuum selection (Section~\ref{sec:vchoice}) do. It is thus a relevant question of whether the occurrence of the mixed type (and its degenerate cases real and definite) is a generic phenomenon. Analyzing the Lagrangian subspaces occurring in a variety of field theories and in a range of geometric settings should shed light on this question.

\paragraph{Geometric quantization}
Lagrangian subspaces have played an important role in quantization for quite some time not only implicitly (Section~\ref{sec:revquant}), but also explicitly. This is particularly the case in \emph{geometric quantization} \cite{Woo:geomquant}. There, an important step in the quantization of a classical phase space is the choice of a \emph{polarization}. To this end we consider the complexified tangent bundle of the phase space manifold. Roughly, a choice of polarization consists in selecting a Lagrangian subspace of the complexified tangent space at each point of phase space. There are also integrability conditions that have to be satisfied. In our case of linear field theory, the phase space manifold is a real vector space, canonically identified with all of its tangent spaces. So the choice of a polarization reduces precisely to a choice of Lagrangian subspace. The Hilbert space of states is then constructed as a space of square-integrable functions on phase space that are invariant under the flows generated by the Lagrangian subspace. For example, if we chose the subspace of ``momenta'' $P_{\Sigma}^{\bC}\subseteq L_{\Sigma}^{\bC}$, compare Definition~(\ref{eq:defmn}) in Section~\ref{sec:schroedvac}, we obtain the \emph{Schrödinger representation} with functions that depend only on configurations. If on the other hand we choose a definite Lagrangian subspace, we obtain a \emph{holomorphic representation} with holomorphic functions on phase space.
A priori, the choice of Lagrangian subspace and thus representation in geometric quantization does not need to carry any physical meaning. In contrast, the Lagrangian subspaces considered in this work have the physical meaning of choosing a vacuum (or a dynamics in the case of amplitudes), but are independent of the representation. For example, the standard vacuum state in Klein-Gordon theory on Minkowski space is determined by a definite Lagrangian subspace. But, it can be perfectly well expressed in the Schrödinger representation, compare formula (\ref{eq:svackg}) of Section~\ref{sec:schroedvac}, which is determined by a real Lagrangian subspace. That being said, it is customary in the literature when using the holomorphic representation to choose the definite Lagrangian subspace that determines the representation to be precisely the same as the one that determines the vacuum. (The remarks at the end of Section~\ref{sec:modecs} with respect to constructing the Fock space starting from an inner product in the phase space amount to precisely that.) This has conceptual advantages in exposing the choice of vacuum in the representation as well as advantages of simplicity.
The different role of the Lagrangian subspaces aside, the mathematical similarities suggest to take advantage of parts of the well developed apparatus of geometric quantization for the our present purposes of developing a generalized notion of vacuum. In particular, one might speculate that in a theory that is non-linear even asymptotically (i.e., beyond the perturbative S-matrix paradigm), a polarization or similar structure might serve to encode a vacuum.

\paragraph{Vacuum selection}
The methods for vacuum selection discussed in Section~\ref{sec:vchoice} are inspired by the Euclidean formulation and represent only a very particular approach to the problem.
In any case, the view of the vacuum as encoding asymptotic boundary conditions (Section~\ref{sec:vaclag}) suggests to formalize a space of \emph{asymptotic solutions}. These would be the solutions of the equations of motion that need only be defined near the ``boundary of spacetime''. More concretely, one might imagine these as defined in ``neighborhoods of the boundary'', i.e., in complements of ``sufficiently large'' compact regions. The vacuum is then determined by a Lagrangian subspace of the complexification of this asymptotic solution space. If the spacetime admits symmetries it is natural to require \emph{invariance of the vacuum}, i.e., of this Lagrangian subspace under the symmetries. In the S-matrix paradigm for Minkowski space, the asymptotic solution space is simply taken to be the product of two copies of global solutions (one for early times and one for late times). A defining feature of the standard vacuum then is of course precisely its invariance under Poincaré transformations. (In fact, this vacuum factorizes into a past and future vacuum that are separately invariant.) The situation becomes more interesting when the asymptotic solution space admits more symmetries than the global solution space. This can be the case with spacetime that admit asymptotic symmetries that do not extend to global symmetries. A well known example is the BMS group for asymptotically flat spacetime \cite{BoBuMe:gravwaves7,Sac:gravwaves8}. To exploit such kinds of symmetries for vacuum selection appears thus particularly promising in our formalism.

\paragraph{Gauge symmetries and fermions}
We have focused in this work on linear bosonic field theories which admit non-degenerate symplectic forms on relevant spaces of germs of solutions. In the presence of \emph{gauge symmetries}, however, we are faced instead with a degenerate pre-symplectic form. It is a natural question then to ask for a possible generalization of our framework to this case. In the case of \emph{abelian} gauge symmetries it seems plausible that we can make things work by applying \emph{symplectic reduction}. More specifically, it was shown in \cite{DiOe:qabym} that (a suitable generalization of) the axiomatic framework for classical field theory outlined in Section~\ref{sec:classlag} applies to this reduced setting. One may thus expect the same for an adapted version of the novel framework of Section~\ref{sec:classglag} and Appendix~\ref{sec:caxioms}. We leave the task of working this out in detail for the future.
Another important question concerns the notion of generalized vacua in \emph{fermionic} field theories. There, instead of a symplectic form, the spaces of germs on hypersurfaces are equipped with a (non-degenerate) \emph{symmetric bilinear form} \cite{Woo:gqbogoliubov}. In the spirit of the considerations in \cite{Oe:freefermi}, one should expect the Lagrangian subspace be replaced by the notion of a \emph{hypermaximal neutral subspace} as encoding a generalized vacuum. Again, we leave the task of concretizing this for the future.

\paragraph{State space}
While the unification of the notions of amplitude and vacuum put forward in this work appears compelling, our proposal remains incomplete in an important sense. We specify a formula, namely formula (\ref{eq:veweyl}), that can be used in principle to calculate expectation values of almost arbitrary observables in any (generalized) vacuum. Sometimes this is all one needs. However, we do not say anything general about other states. The reason is of course that the standard quantization prescription reviewed in Section~\ref{sec:modecs} only works in the standard case that the vacuum is determined by a \emph{definite} Lagrangian subspace. The same restriction underlies the more general local and functorial framework for quantizing field theory \cite{Oe:feynobs} partially reviewed in Section~\ref{sec:stquant}. As we have observed at the end of Section~\ref{sec:hypcyl}, even the standard vacuum of Klein-Gordon theory in Minkowski space induces for certain hypersurfaces (here the hypercylinder) Lagrangian subspaces that are not definite.
The difficulties are particularly apparent in the case of a (complexified) \emph{real} Lagrangian subspace. One may observe that in this case the Schrödinger wave function (\ref{eq:svacstd}) is a pure phase, i.e., the exponential of an imaginary quantity. In contrast to the definite case, this does not lead to a Gaussian exponential term as required to obtain a well defined inner product.
The minimalist way around this problem is to simply ignore the degrees of freedom corresponding to the ``non-definite part'' of the vacuum and quantize the other ones as usual \cite{Oe:timelike,Oe:kgtl}. This is sometimes physically correct and satisfactory \cite{CoOe:smatrixgbf}. However, in the interest of a complete and coherent local description of quantum field theory, the construction of state spaces over such non-standard vacua is a necessity. This problem will be addressed in a subsequent work by the authors.

\subsection*{Acknowledgments}

This work was partially supported by CONACYT project grant 259258 and UNAM-PAPIIT project grant IA-106418.

\appendix

\section{Ingredients from Lagrangian field theory}
\label{sec:lagingreds}

We recall a few basic ingredients of Lagrangian field theory \cite[Chapter~7]{Woo:geomquant}. We adhere mostly to the conventions of \cite{Oe:affine}.

Consider a first-order Lagrangian field theory specified in terms of a \emph{Lagrangian} density $\Lambda(\varphi,\partial\varphi,x)$ as an $n$-form (for spacetime dimension $n$). Here $x$ denotes a point in spacetime, $\varphi$ a field configuration at that point, and $\partial\varphi$ the first jet, i.e., a first field derivative. The \emph{action} for a field configuration $\phi$ in a spacetime region $M$ is the integral of $\Lambda$,
\begin{equation}
 S_M(\phi)\defeq \int_M \Lambda(\phi(\cdot),\partial\phi(\cdot),\cdot) .
\end{equation}
Given a hypersurface $\Sigma$ we denote by $L_\Sigma$ the space of \emph{germs of solutions} of the Euler-Lagrange equations in a neighborhood of $\Sigma$. The \emph{symplectic potential} is the one-form $\theta_{\Sigma}$ on $L_\Sigma$ defined as,\footnote{We use here the opposite sign convention compared to \cite{Oe:holomorphic,Oe:affine,Oe:feynobs}. This also affects expressions~(\ref{eq:sympl}), (\ref{eq:relspact}), (\ref{eq:sfspact}), (\ref{eq:actsymp}) and (\ref{eq:afactsymp}).}
\begin{equation}
 (\theta_{\Sigma})_{\phi}(X)\defeq \int_\Sigma X^a \left.\partial_\mu\lrcorner\frac{\delta \Lambda}{\delta\, \partial_\mu\varphi^a}\right|_\phi .
\label{eq:sympot}
\end{equation}
Here $\phi\in L_\Sigma$ while $X$ is a tangent vector to $\phi$, i.e., an element of the tangent space $T_\phi L_\Sigma$ of solutions linearized around $\phi$. $\partial_\mu$ is a coordinate derivative understood as a vector field and $\lrcorner$ denotes the contraction between vector fields and forms.
The \emph{symplectic form} $\omega_{\Sigma}$ is the exterior derivative of the symplectic potential on $L_\Sigma$,
\begin{multline}
(\omega_\Sigma)_\phi(X,Y)  =(\xd\theta_\Sigma)_\phi(X,Y)
 =\frac{1}{2}\int_\Sigma\left( (X^b Y^a-Y^b X^a)\left.\partial_\mu\lrcorner
 \frac{\delta^2\Lambda}{\delta\varphi^b\delta\,\partial_\mu\varphi^a}\right|_\phi\right. \\
 \left. + (Y^a\partial_\nu X^b-X^a \partial_\nu Y^b)\left.\partial_\mu\lrcorner
 \frac{\delta^2\Lambda}{\delta\,\partial_\nu\varphi^b\delta\,\partial_\mu\varphi^a}\right|_\phi\right) .
\label{eq:sympl}
\end{multline}

Given a spacetime region $M$ and a solution $\phi$ of the Euler-Lagrange equations in $M$ the first variation of the action $S_M$ around $\phi$ vanishes up to a boundary term. This boundary term is precisely the symplectic potential $\theta_{\partial M}$. That is, for an infinitesimal field configuration $X$ we have,
\begin{equation}
 (\theta_{\partial M})_{\phi}(X)=(\xd S_M)_\phi(X) .
\label{eq:relspact}
\end{equation}
Note that this implies,
\begin{equation}
  (\omega_{\partial M})_\phi(X,Y)=(\xd\xd S_M)_\phi(X,Y)=0 .
  \label{eq:sfspact}
\end{equation}
That is, the manifold of solutions in $M$ is \emph{isotropic} with respect to the symplectic form $\omega_{\partial M}$ when restricted to germs on the boundary $\partial M$.

In the present work, we restrict to linear field theory. Then, $L_{\Sigma}$ becomes canonically isomorphic to its tangent spaces and the symplectic potential may be viewed as a bilinear form. We us the notation,
\begin{equation}
  [\phi,X]_{\Sigma}\defeq (\theta_{\Sigma})_{\phi}(X) .
\end{equation}
The symplectic form $\omega_{\Sigma}$ on the other hand loses its dependence on the base point. It is then simply the anti-symmetric part of the symplectic potential,
\begin{equation}
  \omega_{\Sigma}(\phi,\phi')=\frac{1}{2}[\phi,\phi']_{\Sigma}-\frac{1}{2}[\phi',\phi]_{\Sigma} .
  \label{eq:sympfrompotlin}
\end{equation}
The action $S_M$ is quadratic in the linear case and its value on a solution $\phi\in L_M$ in a spacetime region $M$ may then be expressed in terms of the symmetric part of the symplectic potential,
\begin{equation}
  S_M(\phi)=\frac{1}{2} [\phi,\phi]_{\partial M} .
  \label{eq:actsymp}
\end{equation}
Note, crucially, that the right hand side only depends on $\phi$ as an element of $L_{\partial M}$, i.e., does not explicitly depend on $\phi$ in the interior of $M$. If we modify the action to $S_M+D$, where $D$ is linear in field configurations, a solution $\eta\in A^D_M$ of the modified equations of motion satisfies,\footnote{See \cite[(49)]{Oe:feynobs}, but with opposite sign convention for the symplectic potential.}
\begin{equation}
  S_M(\eta)=\frac{1}{2} [\eta,\eta]_{\partial M}-\frac{1}{2} D(\eta) .
  \label{eq:afactsymp}
\end{equation}

\section{Lagrangian subspaces, inner product, and complex structure}
\label{sec:mathlag}

We collect here some relevant elementary statements at the intersection of symplectic vector spaces, indefinite inner product spaces and compatible complex structures \cite{Woo:geomquant, Bog:indipspaces}.

Let $L$ be a real vector space. We call $\omega:L\times L\to\R$ a \emph{symplectic form} if it is bilinear, anti-symmetric and non-degenerate. We denote by $L^{\bC}$ the \emph{complexification} of $L$. $\omega$ extends to a complex anti-symmetric bilinear form $L^{\bC}\times L^{\bC}\to \bC$ that we also denote by $\omega$. A subspace $V\subseteq L$ is called \emph{isotropic} iff for all $v,w\in V$ we have $\omega(v,w)=0$. $V$ is called \emph{coisotropic} iff for any $v\in L\setminus V$ there exists $w\in V$ such that $\omega(v,w)\neq 0$. $V$ is called \emph{Lagrangian} iff it is both isotropic and coisotropic. Consider the hermitian sesquilinear form $(v,w)\defeq 4\im\omega(\overline{v},w)$ on $L^{\bC}$. This defines a non-degenerate indefinite inner product on $L^{\bC}$.
Given a subspace $V\subseteq L^{\bC}$ the set $V^\perp$ of vectors that are orthogonal to all elements of $V$ is a subspace called the \emph{orthogonal companion} of $V$. We say that a subspace $V\subseteq L^\bC$ is \emph{ortho-complemented} iff $L^{\bC}$ is spanned by $V$ together with its orthogonal companion, i.e., $L^\bC=V+V^\perp$.

\begin{lem}
  \label{lem:conjnd}
  Let $V\subseteq L^{\bC}$ be a positive-definite subspace. Then $\overline{V}$ is a negative-definite subspace.
\end{lem}
\begin{proof}
  Let $v\in \overline{V}\setminus\{0\}$. Then, $\overline{v}\in V\setminus\{0\}$ and thus $(\overline{v},\overline{v})>0$. Therefore $(v,v)=\overline{(v,v)}=-(\overline{v},\overline{v})<0$.
\end{proof}

\begin{cor}
  \label{cor:noint}
  Let $V\subseteq L^{\bC}$ be a positive-definite subspace. Then, $V\cap \overline{V}=\{0\}$.
\end{cor}

\begin{lem}
  \label{lem:orthiso}
  Let $V\subseteq L^{\bC}$ be a subspace. Then, $\overline{V}\subseteq V^\perp$ iff $V$ is isotropic.
\end{lem}
\begin{proof}
  Let, $v\in V$ and $w\in \overline{V}$. Then $\overline{w}\in V$. Isotropy of $V$ means that $\omega(\overline{w},v)=0$ for all such choices of $v$ and $w$. But this is equivalent to $(w,v)=0$ which implies the orthogonality of $V$ and $\overline{V}$.
\end{proof}

\begin{lem}
  \label{lem:orthcoiso}
  Let $V\subseteq L^{\bC}$ be a subspace. Then, $V^\perp\subseteq \overline{V}$ iff $V$ is coisotropic.
\end{lem}
\begin{proof}
  $w\in V^\perp$ is equivalent to $(w,v)=0$ for all $v\in V$. This in turn is equivalent to $\omega(\overline{w},v)=0$ for all $v\in V$. Coisotropy would imply $\overline{w}\in V$, i.e., $w\in\overline{V}$. That is $V^\perp\subseteq \overline{V}$. Conversely the latter property would imply coisotropy.
\end{proof}

\begin{cor}
  Let $V\subseteq L^{\bC}$ be a coisotropic and positive-definite subspace. Then, $V$ is a maximal positive-definite subspace.
\end{cor}
\begin{proof}
  Suppose that $V$ is not maximally positive-definite. Then there exists a positive-definite subspace $W\subseteq L^{\bC}$ such that $V$ is a proper subspace of $W$. Take a non-zero vector $w\in W$ that is orthogonal to $V$. Then, $w\in V^\perp\subseteq \overline{V}$ by Lemma~\ref{lem:orthcoiso}. So by Lemma~\ref{lem:conjnd} $(w,w)<0$, a contradiction.
\end{proof}

\begin{lem}[{\cite[Corollary~11.9]{Bog:indipspaces}}]
  Let $V\subseteq L^\bC$ be a finite-dimensional non-degenerate subspace. Then, $V$ is ortho-complemented.
\end{lem}

\begin{lem}
  \label{lem:dlagdec}
  Let $V\subseteq L^\bC$ be an ortho-complemented positive-definite Lagrangian subspace. Then, $L^\bC$ admits an orthogonal decomposition $L^\bC=V\oplus \overline{V}$.
\end{lem}
\begin{proof}
  Since $V$ is ortho-complemented we have $L^\bC=V+V^\perp$. Since $V$ is isotropic and coisotropic we have $\overline{V}=V^\perp$ by combing Lemmas~\ref{lem:orthiso} and \ref{lem:orthcoiso}. Using positive-definiteness Corollary~\ref{cor:noint} completes the proof.
\end{proof}

We say that a positive-definite Lagrangian subspace $V\subseteq L^\bC$ is \emph{complete} if it is ortho-complemented and if $V$ is complete with respect to the inner product $(\cdot,\cdot)$. This makes $L^\bC=V\oplus \overline{V}$ into a \emph{Krein space}.

We call a linear map $J:L\to L$ a \emph{complex structure} if it satisfies $J^2=-\id$. We call $J$ \emph{compatible} if it is a symplectomorphism, i.e., if $\omega(Jv,Jw)=\omega(v,w)$ for all $v,w\in L$. We call $J$ \emph{positive-definite} if the hermitian sesquilinear form $\{v,w\}\defeq 2\omega(v,J w)+2\im\omega(v,w)$ is positive-definite on $L$ viewed as a complex vector space. We call $J$ \emph{complete} if $L$ is complete with respect to this positive-definite inner product.

\begin{prop}
  \label{prop:dlagtoj}
  Let $V\subseteq L^\bC$ be a complete positive-definite Lagrangian subspace. Let $J:L^{\bC}\to L^{\bC}$ be the operator with eigenvalues $\im$ and $-\im$ on $V$ and $\overline{V}$ respectively. Then, $J$ is the complexification of a complete positive-definite compatible complex structure on $L$ (also denoted by $J$).
\end{prop}
\begin{proof}
  By Lemma~\ref{lem:dlagdec} we can write any element of $L^\bC$ uniquely as $v+\overline{w}$ for some $v,w\in V$. Then, $\overline{J(v+\overline{w})}=\overline{\im v - \im \overline{w}}=-\im \overline{v} +\im w=J(\overline{v+\overline{w}})$. That is, $J$ commutes with complex conjugation, i.e., is the complexification of a real linear map $L\to L$. Since $J$ has eigenvalues $\im$ and $-\im$ it is also clear that it satisfies $J^2=-\id$, i.e., it is a complex structure. Let $v,w,v',w'\in V$. Then $\omega(J(v+\overline{w}),J(v'+\overline{w'}))=\omega(\im v-\im \overline{w},\im v'-\im \overline{w'})= \omega(\im v,-\im\overline{w'})+\omega(-\im \overline{w},\im v')=\omega(v,\overline{w'})+\omega(\overline{w},v')=\omega(v+\overline{w},v'+\overline{w'})$. That is, $J$ is compatible. It remains to check that the inner product $\{\cdot,\cdot\}$ on $L$ is positive-definite. Note that any element of $L$ is uniquely represented as $v+\overline{v}$ with $v\in V$. It is sufficient to consider the real part of the inner product. Indeed, $\Re\{v+\overline{v},v+\overline{v}\}=2\omega(v+\overline{v},J(v+\overline{v}))=2\omega(v+\overline{v},\im v -\im\overline{v})=2\omega(v,-\im\overline{v})+2\omega(\overline{v},\im v)=4\im\omega(\overline{v},v)=(v,v)>0$ if $v\neq 0$. (This relation between the inner products also ensures that $\{\cdot,\cdot\}$ is complete.)
\end{proof}

\begin{prop}
  \label{prop:jtodlag}
  Let $J$ be a complete positive-definite compatible complex structure on $L$. Let $V\subseteq L^\bC$ be the eigenspace of the complexification of $J$ for the eigenvalue $\im$. Then, $V\subseteq L^\bC$ is a complete positive-definite Lagrangian subspace.
\end{prop}
\begin{proof}
  Denote by $V'\subseteq L^\bC$ the eigenspace corresponding to the eigenvalue $-\im$ of $J$. If $v\in V$ then $J\overline{v}=\overline{J v}=-\im \overline{v}$. That is, $\overline{V}\subseteq V'$. Similarly, we get $\overline{V'}\subseteq V$. Then, $\overline{V}=V'$. In particular, $L^\bC=V\oplus \overline{V}$.
  Let $v,w\in V$. Then, $\omega(v,w)=\omega(Jv,Jw)=-\omega(v,w)$. That is, $\omega(v,w)=0$, and $V$ is isotropic. Now let $v\in V\setminus\{0\}$. Then, $(v,v)=4\im\omega(\overline{v},v)=\Re\{v+\overline{v},v+\overline{v}\}>0$. That is, $V$ is positive-definite. (This relation between the inner products also implies that $(\cdot,\cdot)$ is complete.) Since $L$ is isotropic to show that it is also coisotropic it suffices to find for any non-zero element $w$ in the complement $\overline{V}$ an element $v\in V$ such that $\omega(v,w)\neq 0$. Indeed, we may choose $v=\overline{w}$ since $4\im\omega(\overline{w},w)=(w,w)>0$. This completes the proof.
\end{proof}

\begin{lem}[{\cite[Lemma~4.1]{Oe:holomorphic}}]
  \label{lem:jcompl}
  Let $J$ be a complete positive-definite compatible complex structure on $L$. Let $W\subseteq L$ be a Lagrangian subspace. Then, $L$ decomposes as a direct sum $L=W\oplus J W$, orthogonal with respect to the real inner product $\Re\{\cdot,\cdot\}$. Moreover, $L^\bC$ decomposes as a complex direct sum $L^\bC=W^\bC\oplus J W^\bC$.
\end{lem}
\begin{proof}
  Let $v,w\in W$. Then $\Re\{v,J w\}=-2\omega(v,w)=0$ since $W$ is Lagrangian. That is, $W$ and $JW$ are orthogonal. Now let $w\in L$ be orthogonal to $V$. Thus, $0=\Re\{v, w\}=2\omega(v,J w)=0$ for all $v$ in $W$. Since $W$ is coisotropic this implies $J w\in W$ and thus $w\in JW$. Thus, the orthogonal complement of $W$ with respect to $\Re\{\cdot,\cdot\}$ is $JW$.
\end{proof}

\begin{prop}
  \label{prop:rdlagcompl}
  Let $V\subseteq L^\bC$ be a complete positive-definite Lagrangian subspace. Let $W\subseteq L$ be a Lagrangian subspace. Then, $L^\bC$ admits a decomposition as a direct sum $L^\bC=V\oplus W^\bC$.
\end{prop}
\begin{proof}
  Let $J$ be the corresponding complete positive-definite compatible complex structure by Proposition~\ref{prop:dlagtoj}.
  By Lemma~\ref{lem:jcompl} we can write any element of $L^\bC$ as $v+Jw$ with $v,w\in W^\bC$. But $v+Jw=v-\im w + \im w+J w$, where clearly $v-\im w\in W^\bC$ while $\im w + J w\in V$.
\end{proof}

\section{Axioms for classical field theory}
\label{sec:caxioms}

We provide here an axiomatization of the classical part of the framework for generalized vacua in terms of Lagrangian subspaces, see Section~\ref{sec:classglag}. The axiomatic system is essentially a generalization of the one provided in \cite{Oe:holomorphic}, without complex structures, compare also Section~\ref{sec:classlag}.

\begin{itemize}
\item[\textbf{(C1)}] Associated to each hypersurface $\Sigma$ is a real vector space $L_{\Sigma}$. $L_{\Sigma}$ is equipped with a non-degenerate symplectic form $\omega_{\Sigma}$.
\item[\textbf{(C2)}] Associated to each hypersurface $\Sigma$ there is an (implicit) linear involution $L_\Sigma\to L_{\overline\Sigma}$, such that $\omega_{\overline{\Sigma}}=-\omega_{\Sigma}$.
\item[\textbf{(C3)}] Suppose the hypersurface $\Sigma$ decomposes into a union of hypersurfaces $\Sigma=\Sigma_1\cup\cdots\cup\Sigma_n$. Then, there is an (implicit) isomorphism $L_{\Sigma_1}\oplus\cdots\oplus L_{\Sigma_n}\to L_\Sigma$. The isomorphism preserves the symplectic form.
\item[\textbf{(C4)}] Associated to each region $M$ is a complex vector space $\tilde{L}_M$.
\item[\textbf{(C5)}] Associated to each region $M$ there is a complex linear map $r_M:\tilde{L}_M\to L_{\partial M}^\bC$. The image $r_M(\tilde{L}_M)$ is a Lagrangian subspace of $L_{\partial M}^\bC$.
\item[\textbf{(C6)}] Let $M_1$ and $M_2$ be regions and $M= M_1\sqcup M_2$ be their disjoint union. Then $\tilde{L}_M$ is the direct sum $\tilde{L}_{M}=\tilde{L}_{M_1}\oplus \tilde{L}_{M_2}$. Moreover, $r_M=r_{M_1} + r_{M_2}$.
\item[\textbf{(C7)}] Let $M$ be a region with its boundary decomposing as a union $\partial M=\Sigma_1\cup\Sigma\cup \overline{\Sigma'}$, where $\Sigma'$ is a copy of $\Sigma$. Let $M_1$ denote the gluing of $M$ to itself along $\Sigma,\overline{\Sigma'}$ and suppose that $M_1$ is a region. Then, there is an injective complex linear map $r_{M;\Sigma,\overline{\Sigma'}}:\tilde{L}_{M_1}\toi \tilde{L}_{M}$ such that
\begin{equation}
 \tilde{L}_{M_1}\toi \tilde{L}_{M}\rightrightarrows L_\Sigma^\bC
\end{equation}
is an exact sequence. Here the arrows on the right hand side are compositions of the map $r_M$ with the complexified projections of $L_{\partial M}$ to $L_\Sigma$ and $L_{\overline{\Sigma'}}$ respectively (the latter identified with $L_\Sigma$). Moreover, the following diagram commutes, where the bottom arrow is the projection.
\begin{equation}
\xymatrix{
  \tilde{L}_{M_1} \ar[rr]^{r_{M;\Sigma,\overline{\Sigma'}}} \ar[d]_{r_{M_1}} & & \tilde{L}_{M} \ar[d]^{r_{M}}\\
  L_{\partial M_1}^\bC  & & L_{\partial M}^\bC \ar[ll]}
\end{equation}
\end{itemize}

\bibliographystyle{stdnodoi} 
\bibliography{stdrefsb,ref2}
\end{document}